\documentclass[nonblindrev]{informs3}

\usepackage[numbers]{natbib}
 \bibpunct[, ]{(}{)}{,}{a}{}{,}%
 \def\bibfont{\small}%

\usepackage{booktabs} 
\usepackage[ruled]{algorithm2e} 

\SetAlFnt{\small}
\SetAlCapFnt{\small}
\SetAlCapNameFnt{\small}
\SetAlCapHSkip{0pt}
\IncMargin{-\parindent}

\usepackage{utopia}  
\usepackage[utf8]{inputenc}
\usepackage{mathtools}
\usepackage{enumerate}
\usepackage{amsmath, amsfonts}
\usepackage{float}
\usepackage{booktabs}
\usepackage{hyperref}
\usepackage{relsize}
\usepackage[american]{babel}
\usepackage[normalem]{ulem}

\usepackage{xcolor}
\definecolor{blue}{RGB}{28,132,218}
\definecolor{green}{RGB}{29,163,35}
\definecolor{orange}{RGB}{237,125,49}
\definecolor{grey}{RGB}{150,150,150}

\newcommand{\argdot}{\,\cdot\,}

\newtheorem{theorem}{Theorem}[section]
\newtheorem{proposition}{Proposition}[section]
\newtheorem{lemma}{Lemma}[section]
\newtheorem{definition}{Definition}[section]
\newtheorem{assumption}{Assumption}[section]

\allowdisplaybreaks

\newcommand\Tstrut{\rule{0pt}{2.0ex}}         

\AtBeginDocument{}
\AtBeginDocument{}
\AtBeginDocument{}

\newboolean{includeHidden}
\setboolean{includeHidden}{false}
\newcommand{\hide}[1]{\ifthenelse{\boolean{includeHidden}}{{\tiny\textbf{HIDDEN:~}#1}}{}}

\usepackage{enumitem}
\newlist{todolist}{itemize}{2}
\setlist[todolist]{label=$\square$}
\usepackage{pifont}

\newenvironment{proof1}{\paragraph{Proof:}}{\hfill$\square$}

\usepackage[most]{tcolorbox}
\newtcolorbox{box_orange}{
	enhanced,breakable,
	boxrule=0pt,frame hidden,
	borderline west={4pt}{0pt}{white!60!orange},
	colback=orange!10!white,
	sharp corners
}
\newtcolorbox{box_green}{
	enhanced,breakable,
	boxrule=0pt,frame hidden,
	borderline west={4pt}{0pt}{white!60!green},
	colback=green!10!white,
	sharp corners
}
\newtcolorbox{box_red}{
	enhanced,breakable,
	boxrule=0pt,frame hidden,
	borderline west={4pt}{0pt}{white!60!red},
	colback=red!10!white,
	sharp corners
}

\usepackage{comment}
\usepackage{subfigure}
\usepackage{amsmath}

\usepackage{algpseudocode}

\makeatletter
\def\namedlabel#1#2{\begingroup
	#2%
	\def\@currentlabel{#2}%
	\phantomsection\label{#1}\endgroup
}

\algnewcommand{\Inputs}[1]{%
	\State \textbf{Inputs:}
	\Statex \hspace*{\algorithmicindent}\parbox[t]{.8\linewidth}{\raggedright #1}
}
\algnewcommand{\Initialize}[1]{%
	\State \textbf{Initialize:}
	\Statex \hspace*{\algorithmicindent}\parbox[t]{.8\linewidth}{\raggedright #1}
}

\usepackage{booktabs}
\usepackage{float}
\usepackage{titlesec}
\usepackage{capt-of}
\usepackage{array}
\usepackage{arydshln}
\allowdisplaybreaks

\usepackage[nameinlink,noabbrev]{cleveref}

\newcommand{\note}[1]{\textcolor{orange}{\emph{\textbf{NOTE:} #1}}}
\newcommand{\todo}[1]{\textcolor{red}{\emph{\textbf{TODO:} #1}}}

\definecolor{darkviolet}{rgb}{0.58, 0.0, 0.83}

\newcommand{\FP}[1]{\textcolor{blue}{FP: #1}}
\newcommand{\NK}[1]{\textcolor{orange}{NK: #1}}

\usepackage{booktabs}
\usepackage{float}
\usepackage{titlesec}
\usepackage{capt-of}
\usepackage{array}
\usepackage{arydshln}
\MANUSCRIPTNO{}

\begin{document}
\RUNAUTHOR{Pieroth, Kohring, Bichler}

\RUNTITLE{Equilibrium Computation in Multi-Stage Auctions and Contests}

\TITLE{Equilibrium Computation in Multi-Stage Auctions and Contests}
\ARTICLEAUTHORS{%
\AUTHOR{Fabian R. Pieroth, Nils Kohring, Martin Bichler}
\AFF{School of Computation, Information, and Technology\\ Technical University of Munich, 85748 Garching, Germany}

} 

\ABSTRACT{%
We compute equilibrium strategies in multi-stage games with continuous signal and action spaces as they are widely used in the management sciences and economics. Examples include sequential sales via auctions, multi-stage elimination contests, and Stackelberg competitions.
In sequential auctions, analysts performing equilibrium analysis are required to derive not just single bids but bid functions for all possible signals or values that a bidder might have in multiple stages. Due to the continuity of the signal and action spaces, these bid functions come from an infinite dimensional space.
While such models are fundamental to game theory and its applications, equilibrium strategies are rarely known. The resulting system of non-linear differential equations is considered intractable for all but elementary models. This has been limiting progress in game theory and is a barrier to its adoption in the field. 
We show that Deep Reinforcement Learning and self-play can learn equilibrium bidding strategies for various multi-stage games.
We find equilibrium in models that have not yet been explored analytically and new asymmetric equilibrium bid functions for established models of sequential auctions. 
The verification of equilibrium is challenging in such games due to the continuous signal and action spaces. We introduce a verification algorithm and prove that the error of this verifier decreases when considering Lipschitz continuous strategies with increasing levels of discretization and sample sizes. 
}
\KEYWORDS{deep reinforcement learning, Nash equilibrium, auctions, contests}


\maketitle

\newpage

\section{Introduction}
Auction theory studies allocation and prices on markets with self-interested participants in equilibrium. Nobel laureate William \citet{vickrey1961counterspeculation} was the first to model markets as games of incomplete information, using the Bayes-Nash equilibrium concept. Participants do not have complete but only distributional information about the valuations of competing market participants. An equilibrium bid function determines how much they bid based on their value draw and their knowledge of the prior distribution. Today, auction theory is arguably the best-known application of game theory~\citep{klemperer2000every}. 
However, incomplete-information game-theoretical models with continuous signals (e.g., values) and actions (e.g., bids) are not restricted to the study of auctions. They are used in a variety of important problems studied in the economic and management sciences today. For example, such models have been used to study competitive situations as they can be found in crowd-sourcing, procurement, or R\&D contests \citep{konrad2009strategy, vojnovicContestTheoryIncentive2016} as well as in established models in oligopoly pricing (e.g., Stackelberg competition) and many more. 

Despite the enormous academic attention that this line of game theory has received, equilibrium strategies are only known for very restricted, simple market models, such as single-object auctions with independent prior value distributions. Analytical derivations of equilibrium strategies do not only require strong assumptions but are usually also the result of months or even years of work. For multi-object auctions, bidders with non-uniform or interdependent valuation distributions or non-quasilinear (i.e., payoff-maximizing) utility functions, as they are analyzed in behavioral economics, we are typically not aware of any explicit equilibrium bid function. 
In these market models, the equilibrium problem can be described as a system of non-linear differential equations for which we have no general mathematical solution theory. 

Beyond analytical solutions, numerical techniques for differential equations have turned out to be challenging and have received only limited attention. \citet{fibich2011numerical} criticize the inherent numerical instability of standard techniques used in this field. 
Equilibrium learning is an alternative numerical approach to finding equilibrium. This literature examines what kind of equilibrium arises as a consequence of {a relatively simple process of learning and adaptation}, in which agents are trying to maximize their payoff while learning about the actions of other agents \citep{fudenberg1998theory, hartUncoupledDynamicsNot2003, young2004strategic}. Learning provides an intuitive, tractable model of how equilibrium emerges but does not always end up in a Nash equilibrium. There are games in which \textit{any} learning dynamics fail to converge to a Nash equilibrium \citep{milionis2022nash, mazumdar2019policy, daskalakisLearningAlgorithmsNash2010, vlatakis2020no}. Learning algorithms can cycle, diverge, or be chaotic, even in zero-sum games, where the Nash equilibrium is tractable \citep{mertikopoulos2018cycles, baileyMultiplicativeWeightsUpdate2018, cheung2020chaos}. \citet{sanders2018prevalence} argue that chaos is typical behavior for more general matrix games. 
However, in certain types of games, such as potential games, equilibrium can be learned \citep{monderer1996potential, fudenberg1998theory}. 

Recently, neural equilibrium learning algorithms were introduced for single-stage Bayesian auction games \citep{bichler2021learning}. These algorithms are based on neural networks and self-play, and they found equilibrium in a wide variety of auctions and markets. 
While partial results exist \citep{bichler2021learning}, analyzing the reasons for the convergence of neural networks and self-play in continuous Bayesian auction games and many other games remains challenging \citep{milionis2022nash, vi2023, ahunbay2024uniqueness}. Still, these methods allowed finding equilibrium in a number of challenging game-theoretical problems \citep{bichlerLearningEquilibriaAsymmetric2023}. 

So far, all applications of neural equilibrium learning are simultaneous-move, single-stage auction games \citep{bichlerLearningEquilibriaAsymmetric2023}. Many managerially relevant game-theoretical problems are multi-stage games. Sequential sales of multiple objects or Stackelberg price competition describe textbook examples of such problems~\citep{krishna2009auction}. Dynamic or multi-stage games are ubiquitous in management. For instance, requests for quotes in procurement are often followed by an auction, and elimination contests involve multiple stages. Equilibrium analysis for such games has proven to be even more challenging. An equilibrium is only known for a few very restricted multi-stage models that require strong assumptions.
\citet{myersonPerfectConditionalEEquilibria2020} introduced the notion of multi-stage games with infinite sets of signals and actions and discussed the existence of equilibrium for a large class of such games. We refer to these games as \textit{continuous multi-stage games} and always assume continuous signal and action spaces in this paper. This game class includes Bayesian games, stochastic games with finite horizons, and combinations thereof. As such, it can be considered one of the most general game classes that allows us to analyze a large set of managerially relevant strategic situations. Unfortunately, not much is known about the computation of equilibrium in such games.

\subsection{Contributions}

We find equilibrium in continuous multi-stage games via deep reinforcement learning (DRL) algorithms and self-play. In contrast to earlier work on neural equilibrium learning in single-stage games, this requires us to consider the state of the game and therefore also different learning algorithms. Reinforcement learning was developed for single-agent learning in discrete-time stochastic control processes, and it was successfully applied in a wide range of applications ranging from robot control to elevator scheduling \citep{sutton2018ReinforcementLearningIntroduction}. 
Multi-agent reinforcement learning (MARL) has achieved veritable successes in achieving superhuman performance in finite games with complete information, such as chess, shogi, and Go~\citep{silverMasteringGameGo2016, Silver2018}, and in games with imperfect information, such as poker~\citep{Brown18:Superhuman, Brown19:Superhuman}.
It has also reached superhuman levels in games with continuous signals and actions, such as Dota~2~\citep{bernerDotaLargeScale2019} and professional level in Starcraft~II~\citep{alphastar}. 
While these successes demonstrate that MARL is capable of learning good policies, it remains unknown how far the learned strategies are from an equilibrium or whether the applied algorithms generally converge towards equilibrium (see Section~\ref{sec:discussion-convergence-in-games}). 

We draw on policy gradient methods such as \textsc{Reinforce} and Proximal Policy Optimization (PPO) and use neural networks to approximate the equilibrium bid functions. 
In our first contribution, we show that policy gradient methods converge to equilibrium strategies with self-play in a wide variety of models including sequential auctions, elimination contests, and Stackelberg-Bertrand competitions with appropriate versions of DRL algorithms. 
The fact that DRL with self-play does converge to such a strategy profile is remarkable. 
These algorithms are originally designed for single-agent tasks, and convergence results rely on the environment being stationary and Markovian \citep{suttonPolicyGradientMethods1999}. Both of these properties do not hold in game-theoretical problems with multiple agents and imperfect information. 
In general, these standard algorithms do not converge even in zero-sum settings \citep{littmanMarkovGamesFramework1994, srinivasan2018actor} or also linear quadratic games \citep{mazumdarPolicyGradientAlgorithmsHave2020}. 

We show how these techniques help us finding equilibrium in environments, where we did not have equilibrium predictions so far. For each environment, we start our empirical study with a \emph{standard} model from the literature that is simple enough to derive an analytical equilibrium. We use this as an unambiguous baseline. We then increase the complexity along three dimensions to provide equilibria in environments for which no equilibrium is known so far. These dimensions reflect relaxations of strong assumptions that analytical models make, but that are rarely warranted in the field. 

First, we drop the standard assumption that symmetric agents always follow \textit{symmetric strategies} \citep{krishna2009auction}. For example, we find new asymmetric equilibrium strategies in the well-known sequential sales model, even though agents have symmetric prior distributional information. Additionally, we find a new approximate equilibrium in an elimination contest with asymmetric prior distributions.
Second, we drop the assumption of \textit{independent prior distributions}. While interdependencies have been explored in single-stage auctions, we are not aware of equilibrium predictions for interdependent priors in multi-stage games. 
We find new equilibrium strategies in sequential auctions, elimination contests, and Stackelberg-Bertrand competitions under affiliated and correlated prior distributions.
Third, we drop the assumption of \textit{quasi-linear utilities} and consider the common behavioral effect of risk aversion in continous multi-stage games. Again, we find novel approximate equilibrium strategies in all three settings.
This provides extensive empirical evidence that even continous multi-stage economic possess a particular structure that is conducive to learning equilibrium. Importantly, it opens the way to practical equilibrium solvers that allows us to explore equilibrium making different assumptions in a few hours or minutes. 

The second main contribution, that was instrumental for our new equilibrium predictions, is an algorithm that allows us to verify an approximate global Nash equilibrium in continuous multi-stage games. 
Deciding whether a strategy profile is a global Nash equilibrium in a finite, complete-information game can be done in polynomial time. However, in the context of continuous Bayesian games, where agents may have different signals, determining whether a strategy profile is a Bayesian Nash Equilibrium is computationally hard. Even when allowing for a constant error to verify an approximate equilibrium, verification is NP-hard~\citep{caiSimultaneousBayesianAuctions2014}. 
It is costly in particular due to the infinite space of alternative strategies. Small changes in the strategy of one player in the first stage can have ample effects on the strategies of all other players in subsequent stages. Given this infinite-dimensional function space, we prove that the error of our verifier grows small if the level of discretization and the number of samples increases. This is not obvious, because even small deviations in the strategy of one round could have a large impact on the equilibrium strategy in later rounds. Our proof highlights the very assumptions that we need to make for such guarantees. We want to emphasize the generality of the verifier that is able to certify equilibrium for the broad class of continuous multi-stage games. 

Overall, we show that one can learn and verify equilibrium in a large variety of managerially relevant multi-stage games. While the computational hardness results on equilibrium computation put limits on the size of the games that can be analyzed, most games that have been discussed in the literature have a small number of stages and low-dimensional action space for each player, which renders equilibrium analysis tractable. 
Our results illustrate that for elimination contests, multi-stage auctions, and Stackelberg competitions we can learn and verify equilibrium in due time, which pushes the boundaries of equilibrium analysis and allows us to study models that have so far been considered intractable.

\section{Related Work}

Our work draws on different lines of research related to game theory and machine learning. We discuss differentiable economics, equilibrium computation in single-stage continuous Bayesian games, computational complexity of learning, and equilibrium verification.

\subsection{Differentiable economics}
In recent years, deep learning approaches to designing auction mechanisms have received significant attention~\citep{golowichDeepLearningMultifacility2018, feng2018deep, duettingOptimalAuctionsDeep2019, zhangComputingOptimalEquilibria2023, wang24}. These efforts aim to design mechanisms that satisfy desirable properties by incorporating constraints into the deep learning optimization problem.
One strand focuses on designing auction mechanisms that are approximately incentive compatible. The constraints ensure that players do not gain more than an $\varepsilon$ in expectation by deviating from the truthful strategy, even \emph{after} all valuations are known. 
The goal is to identify a mechanism where the truthful strategy profile is an equilibrium strategy, meaning a bidder has approximately no incentive to conceal their valuation, regardless of the reported valuations of the opponents. 
\citet{duettingOptimalAuctionsDeep2019} provide a concentration bound to empirically assess the violation of incentive compatibility. 
However, this bound assumes that the \textit{ex post} violation can be precisely determined. 
\citet{curryCertifyingStrategyproofAuction2020} address this issue by linearizing the learned neural network, effectively reducing the problem to an integer program that allows for an accurate estimation of the error. \citet{curryDifferentiableEconomicsRandomized2023} use deep learning to learn auction mechanisms within randomized affine maximizer auctions, a class within which each mechanism is exactly incentive compatible.
Our work contributes to the overall agenda of differential economics, but focuses on finding equilibrium in a given game using neural networks instead of designing a game. Second, when a candidate strategy profile has been identified, we want to verify whether it is indeed an approximate equilibrium. 

\subsection{Equilibrium computation}

Equilibrium learning algorithms for single-stage Bayesian auction games have only emerged in recent years. One notable algorithm, Neural Pseudogradient Ascent (NPGA), employs neural networks and self-play techniques to compute Bayes-Nash equilibria across a wide range of auction games \citep{bichler2021learning, bichler2023learning}. Remarkably, NPGA has demonstrated the ability to find equilibria in diverse auctions and contests, including models with interdependent valuations, non-quasilinear utilities, and multiple Bayes-Nash equilibria. Follow-up work demonstrates that NPGA also successfully finds equilibrium even in two-sided markets \citep{bichlerLearningEquilibriumBilateral2022}, where an infinite number of equilibria exist, and in auctions with asymmetric bidders~\citep{bichlerLearningEquilibriaAsymmetric2023}.

Prior forms of equilibrium learning are limited to single-stage interactions. An exception is \citet{gerstgrasserOraclesFollowersStackelberg2023} who proposed a meta-learning approach for learning equilibrium in Stackelberg games.
However, the work assumes access to a follower best-response oracle and an algorithm that perfectly solves a much larger proxy single-agent partially observable Markov decision process, which are hard problems on their own. The approach is not designed for continuous actions as it would require (infinitely) many queries of the leader policy.

\subsection{Computational hardness of equilibrium computation}

Equilibrium learning is a form of equilibrium computation for which there is a large literature about computational complexity. For finite, complete-information games, computing a Nash equilibrium is PPAD-hard~\citep{daskalakisComplexityComputingNash2009}. 
For subjective priors and a uniform tie-breaking rule, it was shown that approximating the Bayes-Nash equilibrium of a first-price auction is PPAD-hard by \citet{filos2021complexity}. Note, that this contradicts the common prior assumption usually made in auction theory. \citet{caiSimultaneousBayesianAuctions2014} prove that finding a Bayesian Nash equilibrium in a simultaneous Vickrey auction game in which the bidders have combinatorial valuations, even approximating a Bayesian Nash equilibrium is NP-hard. For sequential game variants, such as stochastic games, also computing coarse correlated equilibria is  computationally hard~\citep{daskalakis2022complexity}.
We remark that the question of computational complexity is distinct from our question on whether learning algorithms in self-play converge to approximate equilibrium strategies. Importantly, in the models that we analyze agents have a low-dimensional continuous action space, and we have a small number of stages such as in an elimination contest with a few stages or a sequential auction of a few objects only that are sold sequentially. Yet, equilibrium analysis for such models is challenging, and we neither have analytical equilibrium predictions nor numerical techniques to approximate equilibrium in most of these games. 


\subsection{Verifying equilibrium}

Verifying whether a strategy profile is an equilibrium is a central challenge for continuous multi-stage games. Related problems arise in imperfect information extensive-form games such as poker. 
Game abstraction is an influential technique for solving large imperfect-information games~\citep{Shi00:Abstraction,Billings03:Approximating,Gilpin06:Competitive}. The process involves automatically abstracting the game into a smaller version, solving this smaller game for an (approximate) equilibrium, and then mapping the strategies back to the original game. However,  abstraction techniques do not guarantee equilibrium approximation in the original game~\citep{waughAbstractionPathologiesExtensive2009}.
Most methods that provide error guarantees focus on games with finitely many actions~\citep{Gilpin07:Lossless, Sandholm12:Lossy, Kroer14:Extensive, Kroer16:Imperfect, Kroer18:Unified}. \citet{Kroer15:Discretization} address abstraction in games with continuous actions, but their model cannot easily be adapted for the multi-stage games with continuous signals and actions that we consider. 

\citet{bichlerComputingBayesNash2023} focuses on single-stage continuous Bayesian games such as auctions. They use discretization in single-stage single-item auctions to verify equilibrium, assuming explicit access to game dynamics. In contrast, our approach for multi-stage games relies only on sampling for game rollouts over all stages.
\citet{timbersApproximateExploitabilityLearning2020} propose a reinforcement learning-based method to estimate a \emph{lower} bound on the maximum utility loss, providing insights into potential gains from deviation but not verifying whether a strategy profile is an approximate equilibrium.
\citet{hosoyaApproximatePurificationMixed2022} demonstrate that one can estimate a best-response arbitrarily well with only finitely many points in Bayesian games with continuous utilities and independent priors. In addition, \citet{bosshard2020computing} introduce a verification method for approximate equilibrium strategies in single-stage auctions with independent priors, approximating utility loss at a finite number of grid points. However, in dynamic games, interdependencies naturally arise due to signals revealing information in earlier stages that influence future actions, which is what we consider in our verification method.

In summary, previous results on verifying equilibrium are restricted to finite games, single-stage games, make disentanglement assumptions (e.g., independent priors), or provide only lower bounds on utility loss. Therefore, these verification techniques cannot be extended to continuous multi-stage games. We present an equilibrium verifier designed specifically for multi-stage games with continuous signal and action spaces, which does not make such restrictive assumptions.

\section{The model}

In what follows, we introduce continuous multi-stage games with continuous signal and action spaces formally and discuss relevant equilibrium solution concepts and reinforcement learning algorithms used.

\subsection{Continuous Multi-Stage Games}

Finite multi-stage games are often modeled in extensive form, leading to a game tree. We study games with multiple but finitely stages, where players choose actions from a continuous set but only need to make finitely many decisions, such as consecutive bids in sequential auctions. Also, agents receive signals that may be from uncountable sets, such as a continuous signal (or value) space in auction theory. 
In contrast to finite games, much less is known about the existence of equilibrium in this rich class of multi-stage games with continuous signal and action spaces. We draw on the rich formal setting introduced by \citet{myersonPerfectConditionalEEquilibria2020}. 
In a multi-stage game $\Gamma$, $N$ players plus nature interact simultaneously for $T$ stages. In each stage $t$, each player $i$ receives a signal $s_{it}$ and chooses its action $a_{it}$. After $T$ stages, each agent receives a utility $u_i(a)$ based on an outcome $a$. 

We assume all functions to be measurable over suitable sigma-algebras and denote with $\Delta(X)$ the set of countably additive probability measures on the measurable subsets of $X$. 
For measurable spaces $X$ and $Y$, a mapping $f: Y \rightarrow \Delta(X)$ is called a transition probability if, for every measurable subset $A \subset X$, the function $f(A\,|\,\argdot): Y \rightarrow \mathbb{R}$ is measurable.

\begin{definition}[Multi-stage game \citep{myersonPerfectConditionalEEquilibria2020}]\label{def:msg}
	A multi-stage game is specified by a tuple $\Gamma = \left(\mathcal{N}, T, S, \mathcal{A}, p, \sigma, u \right)$, where the elements are given by:
	\begin{enumerate}
		\item $\mathcal{N} = \{1, \dots, N\}$ is the set of $N \in \mathbb{N}$ players. Let $\mathcal{N}^* = \mathcal{N} \cup \{0\}$, where $0$ is nature.
		\item $T \in \mathbb{N}$ is the number of stages. Let $L=\mathcal{N} \times \{1, \dots, T\}$ denote the set of players and stages, and $L^*= \mathcal{N}^* \times \{1, \dots , T\}$. For simplicity, we write ``$it$'' for $(i, t)$ and $t\in T$ to express $t \in \{1, \dots, T\}$.  
		\item $S = \bigtimes_{it \in L} S_{it}$, where $S_{it}$ is the set of possible signals that player $i$ can receive at stage $t$. It holds $S_{i1}= \{ \emptyset\}$ for all $i \in \mathcal{N}$.
		\item $\mathcal{A}_{it}$ denotes the set of player $i$'s actions in stage $t$. $\mathcal{A}_{0t}$ is nature's set of actions in stage $t$. $\mathcal{A} = \bigtimes_{it \in L^*} \mathcal{A}_{it}$ denotes the set of possible \emph{outcomes} of the game. Each element in $\mathcal{A}$ describes a game's complete roll-out. 
		We denote the projection onto stages before $t$ by a subscript $<t$. For example, for $a \in \mathcal{A}$, $a_{<t} = (a_{ir})_{i \in \mathcal{N}^*, r < t}$ denotes the history proceeding stage $t$.
		\item $p = (p_1, \dots, p_T)$ is nature's fixed probability function, where $p_t : \mathcal{A}_{<t} \rightarrow \Delta \left( \mathcal{A}_{0t}\right)$ for all $t \in T$.
		\item $\sigma_{it} : \mathcal{A}_{<t} \rightarrow S_{it}$ denotes player $i$'s signal function for stage $t$, and $\sigma = \left(\sigma_{it} \right)_{it \in L}$.
		\item $u_i:\mathcal{A} \rightarrow  \mathbb{R}$ is player $i$'s bounded utility function, and $u = \left(u_i \right)_{i \in \mathcal{N}}$. 
	\end{enumerate}
\end{definition}

We make the common assumption \citep{kanekoBehaviorStrategiesMixed1995, bonannoMemoryPerfectRecall2004, shohamMultiagentSystemsAlgorithmic2008} of perfect recall (Definition~\ref{def:perfect-recall}), i.e., agents remember their actions and received information. 
In Section~\ref{sec:verifier-formal-description-and-analysis}, we discuss additional assumptions needed for our theoretical results such as the bid functions to be pure and Lipschitz continuous.  
Multi-stage auctions or contests are illustrative examples and the focus of this paper, but the game class is very general and allows to model single-stage Bayesian games, signaling games, and all types of finitely-repeated games, and finite-horizon stochastic games. 
By including dummy actions and signals, one can model sequential move games or, more generally, games where a subset of players acts in each stage. For example, in the settings considered in Section~\ref{sec:experimental-results}, nature moves first and draws players' types. The players receive their types as signals in the second stage and act. This can be modeled by setting $\mathcal{A}_{i1}$ to be a singleton for $i \in \mathcal{N}$. So, a two-stage sequential auction can be modeled by a three-stage game.

A strategy for player $i$ in stage $t$ is a transition probability $\beta_{it}: S_{it} \rightarrow \Delta \left(\mathcal{A}_{it} \right)$. Let $\Sigma_{it}$ denote $i$'s set of strategies at time $t$ and $\Sigma_i = \bigtimes_{t \in T} \Sigma_{it}$ player $i$'s set of strategies. Finally, $\Sigma = \bigtimes_{it \in L} \Sigma_{it}$ is the set of all strategies. We denote with $\beta_{\cdot t} = \left(\beta_{it} \right)_{i \in \mathcal{N}}$ the strategy vector of stage $t$. Likewise, the $\cdot t$ notation denotes corresponding product spaces and vectors for stage $t$, e.g., $\mathcal{A}_{\cdot t} = \bigtimes_{i \in \mathcal{N}^*} \mathcal{A}_{it}$ and $\Sigma_{\cdot t} = \bigtimes_{i \in \mathcal{N}} \Sigma_{it}$.
We use~$\tilde{\cdot}$ to refer to variables in the ex-ante state of the game. For example, $\tilde{u}: \Sigma \rightarrow \mathbb{R}$ describes the ex-ante utility with strategy profile $\beta$.


The most prominent equilibrium solution concept in non-cooperative game theory is the Nash equilibrium (NE) \citep{nash1950equilibrium}. Informally, it is a fixed point in strategy space where no player unilaterally wants to deviate from.

\begin{definition}[$\varepsilon$-Nash equilibrium] \label{def:eps-Nash-equilibrium}
	Let $\Gamma = \left(\mathcal{N}, T, S, \mathcal{A}, p, \sigma, u \right)$ be a multi-stage game, $\varepsilon \geq 0$, and $\beta^* \in \Sigma$ a strategy profile. Then $\beta^*$ is an $\varepsilon$-Nash equilibrium ($\varepsilon$-NE) if and only if for all $i \in \mathcal{N}$ and $ \beta_{i} \in \Sigma_i$
	\begin{align} \label{equ:NE-condition}
		\tilde{u}_i(\beta_{i}, \beta^{*}_{-i}) \; \leq \; \tilde{u}_i(\beta^*) + \varepsilon.
	\end{align}
	We denote $\beta^*$ simply as Nash equilibrium for $\varepsilon=0$.
\end{definition}

The NE is prevalent for normal-form games: single-stage games with complete information. However, in games with multiple stages, one might want to exclude Nash equilibria that rely on non-credible threats or non-best responses in some subgames. 
One way to exclude these unwanted equilibria is to demand the strategies to be \emph{sequentially rational}. 
\citet{krishna2009auction} defines this recursively as equilibria with the property that following an outcome of the current stage, the strategies in the next stage form an equilibrium. In the complete information case, this leads to \emph{subgame perfect (Nash) equilibrium}, eliminating unreasonable Nash equilibria. 
Extending these ideas to games with imperfect information demands the agents to have consistent beliefs over the other agents' hidden information, leading to \emph{perfect Bayesian equilibrium}~\citep{choSignalingGamesStable1987} and \emph{sequential equilibrium}~\citep{kreps1982sequential}. With continuous signal and action spaces, there is still an ongoing discussion about a suitable definition of refined solution concepts such as the perfect Bayesian equilibrium \citep{gonzalez-diazNotionPerfectBayesian2014, watsonGeneralPracticableDefinition2017a} or sequential equilibrium \citep{myersonPerfectConditionalEEquilibria2020}.
Due to the technical hardness of these refinements, we restrict ourselves with finding and verifying equilibrium from the broader class of NE in the present work. 


\subsection{Learning for Equilibrium Selection}

A game theoretic solution concept such as the NE usually follows the normative approach, telling the players how to act. This neglects the question of how the players would find and agree on an equilibrium when having only partial information about the game and the opponents' behavior \citep{ashlagiRobustLearningEquilibrium2006}. Instead, we use learning as equilibrium selection method and analyze whether the resulting strategy profile is a $\varepsilon$-Nash equilibrium, drawing on the equilibrium learning paradigm. When the algorithms converge to a specific equilibrium strategy profile, there is a clear answer to the above questions.

Consider a parametrization of the players' strategies so that for every $it \in L$, there is a set of parameters $\Theta_{it}$ and a mapping $\Theta_{it} \mapsto \Sigma_{it}$ that maps onto a strategy. We denote a parametrized strategy by $\beta_{\theta_{it}} \in \Sigma_{it}$ and $\beta_{\theta_i} = \left(\beta_{\theta_{i1}}, \dots, \beta_{\theta_{iT}} \right) \in \Sigma_{i}$. In practice, we train a single neural network for all stages and choose the signal spaces accordingly. Additionally, we either use one network for each player or share the weights in case of a symmetric task.
We focus on policy gradient learning algorithms. That is, player $i$ updates the parameters $\theta_i^r$ in iteration $r$ by
\begin{align*}
	\theta_i^{r+1} = \theta_i^{r} + \eta_i^r \cdot \nabla_{\theta_i} \tilde{u}_i\left(\beta_{\theta_i^r}, \beta_{\theta_{-i}^r}\right),
\end{align*}
where $\eta_i^r \in \mathbb{R}$ denotes the learning rate. 

Commonly, players do not have perfect information about the environment but instead, treat them as black boxes. Thus, one can only evaluate the objective function for a given policy but cannot estimate the gradient with respect to the policy's parameters. Multiple gradient estimation techniques that overcome this issue have been suggested. We focus on two common policy gradient algorithms, namely \textsc{Reinforce} \citep{mohamedMonteCarloGradient2020} and PPO \citep{schulman2017proximal}. These Policy Gradient Methods learn distributional strategies, which provides a different way to deal with the discontinuities in the ex-post utility functions that led to challenges in earlier approaches \citep{bichler2021learning}. Rather than custom training algorithms based on Evolutionary Strategies, we use standard backpropagation to train the neural networks. 
\subsection{\textsc{Reinforce}}
\textsc{Reinforce} is a policy gradient algorithm based on the idea of optimizing the policy directly without explicitly estimating the value function. The \textsc{Reinforce} algorithm computes the gradient of the expected return with respect to the policy parameters. This estimator is also known as the score-function estimator. The gradient is given as 
\begin{align*}
	\nabla_{\theta_i} \, \tilde{u}_i\left(\beta_{\theta_i}, \beta_{\theta_{-i}}\right) &= \nabla_{\theta_i} \int_{\mathcal{A}} u_i(a) \rho\left(a  \, | \, \beta_{\theta_{i}}, \beta_{\theta_{-i}}\right) d a \\
	& = \int_{\mathcal{A}} u_i(a) \nabla_{\theta_i}  \log\left( \rho\left(a  \, | \, \beta_{\theta_{i}}, \beta_{\theta_{-i}}\right) \right) d a \\
	& = \mathbb{E}_{a \sim P\left(\cdot \, | \, \beta_{\theta_{i}}, \beta_{\theta_{-i}}\right)} \left[u_i(a) \nabla_{\theta_i}  \log\left( \rho\left(a  \, | \, \beta_{\theta_{i}}, \beta_{\theta_{-i}}\right) \right) \right],
\end{align*}
where $\rho\left(\cdot  \, | \, \beta_{\theta_{i}}, \beta_{\theta_{-i}}\right)$ denotes $P\left(\cdot \, | \, \beta_{\theta_{i}}, \beta_{\theta_{-i}}\right)$'s density function. This derivation follows the policy gradient theorem \citep{sutton2018ReinforcementLearningIntroduction}. We parametrize the neural network to output the mean and standard deviation of a Gaussian distribution so that one can assume $\rho$ to exist. More details and a broader discussion on this estimator can be found in the study of \citet{mohamedMonteCarloGradient2020}.

\subsection{Proximal Policy Optimization}

Proximal Policy Optimization (PPO) has been introduced by \citet{schulman2017proximal} and can be considered an extension of the \textsc{Reinforce} algorithm.
PPO is designed to strike a balance between sample efficiency and stability during training. It addresses the challenges of policy optimization by using multiple iterations of stochastic gradient ascent, where the update in each iteration is limited to a certain range, thus avoiding large policy updates that could disrupt learning. This constraint helps maintain stability and prevents the agent from deviating too far from its previous policy. 

PPO has been widely adopted and has demonstrated strong performance across a range of complex tasks.
It falls in the broader class of actor-critic algorithms -- thereby introducing a second network that estimates the state's value --  and additionally uses a technique called trust region policy optimization, which helps to prevent the algorithm from making large, potentially harmful changes to the policy. This tends to make learning more stable compared to \textsc{Reinforce}. 
PPO is considered a state-of-the-art method, particularly in complex environments with high-dimensional state spaces, and it has been particularly successful in combinatorial games \citep{yu2022surprising}.

\section{Evaluation Metrics}\label{subsec:metrics}

The ex-ante \emph{utility loss} is a metric to measure the loss of an agent by not playing a best-response to the opponents' strategies $\beta$ \citep{srinivasan2018actor,brownDeepCounterfactualRegret2019}. It is given by
\begin{align} \label{equ:ex-ante-util-loss}
	\tilde{\ell}_i(\beta_i, \beta_{-i}) = \sup_{\beta^{\prime}_i \in \Sigma_i} \tilde{u}_i(\beta^{\prime}_i, \beta_{-i}) - \tilde{u}_i(\beta_i, \beta_{-i}).
\end{align}
By Definition~\ref{def:eps-Nash-equilibrium}, a strategy profile $\beta$ is an $\varepsilon$-NE if and only if $\tilde{\ell}_i(\beta_i, \beta_{-i}) \leq \varepsilon$ for all agents $i$. 
In a setting with a known analytical equilibrium $\beta^*=(\beta_i^*, \beta_{-i}^*)$, we compute $\tilde{\ell}_i$ with respect to the opponents' equilibrium strategies instead. That is, we estimate closeness in utility of a strategy $\beta_i$ to $\beta_i^*$ by the utility loss in equilibrium by
\begin{align} \label{equ:ex-ante-util-loss-in-equilibrium}
	\tilde{\ell}_i^{\text{equ}}(\beta_i) := \tilde{\ell}_i(\beta_i, \beta_{-i}^*) = \tilde{u}_i(\beta_i^*, \beta_{-i}^*) - \tilde{u}_i(\beta_i, \beta_{-i}^*).
\end{align}

In general, we cannot evaluate the ex-ante utility for a strategy profile $\beta \in \Sigma$. Therefore, we estimate it via Monte-Carlo approximation by
\begin{align} \label{equ:monte-carlo-estimation-for-ex-ante-utility}
	\hat{u}_i\left(\beta \right) := \frac{1}{n_{\text{batch}}} \sum_{a \sim P(\cdot \, | \, \beta)} u_i(a) \approx \tilde{u}_i\left(\beta\right),
\end{align}
where we simulate $n_{\text{batch}}$ roll-outs using strategy profile $\beta$. We use this procedure in all reported metrics based on the ex-ante utility. For example, we use this to estimate $\tilde{\ell}_i^{\text{equ}}(\beta_i)$ and denote the estimated utility loss in equilibrium by $\ell_i^{\text{equ}}$. 

Additionally, we want to evaluate convergence in strategy space because even strategies with a low utility loss could be arbitrarily far away from the equilibrium strategies. For this, we report the probability-weighted root mean squared error of $\beta_{it}$ and $\beta_{it}^*$ for every agent $i$ and stage $t$ to approximate the weighted $L_2$ distance of these functions
\begin{align} \label{equ:dated-l2-distance}
	L_2^{it}(\beta, \beta^*) = \left(\frac{1}{n_{\text{batch}}} \sum_{s_{it} \sim P_{it}\left(\cdot \,|\,\beta^* \right)}\left(\beta_{it}(s_{it}) - \beta_{it}^*(s_{it}) \right)^2 \right)^{\frac{1}{2}},
\end{align}
where $P_{it}\left(\cdot \,|\,\beta^* \right)$ denotes the probability measure over signals induced by strategy profile $\beta^*$. We also report the mean $L_2^{i, \text{avg}}$ loss over all stages, which is given by $L_2^{i, \text{avg}}(\beta, \beta^*) = \frac{1}{T} \sum_{t \in T} L_2^{it}(\beta, \beta^*)$.

\section{Verification in Settings with Unknown Equilibrium} \label{sec:verifier-formal-description-and-analysis}
In cases where no analytical solution is known, estimating the best-response utility for the utility loss becomes necessary in order to verify whether a strategy profile is indeed an approximate NE. However, deriving such guarantees ex-post is challenging in continuous games with multiple stages. Moreover, even within the realm of single-stage games, the existing theoretical guarantees for verification methods are limited and depend on strong assumptions \citep{bosshard2020computing}. Due to space constraints, the detailed technical description of the verifier is located in Appendix~\ref{sec:verification-procedure-details}.

Without imposing any additional regularity on the strategy profile $\beta = (\beta_i, \beta_{-i})$, it is very hard to give any theoretical guarantees.
Therefore, we limit each strategy $\beta_i$ to be from the set of Lipschitz continuous functions, that is denoted by $\Sigma_i^{\text{Lip}}$. While this limits the strategy space somewhat, it still allows for mixed strategies and any strategy that can be represented by a neural network \citep{hornik1991ApproximationCapabilitiesMultilayer}. Therefore, we argue that this restriction is mild, in particular, for practical applications such as auctions. However, it's essential to strike a balance; we shouldn't overly constrain our set of strategies. For instance, limiting it to pure or piece-wise constant functions would sacrifice the generality and scalability that the broader space of Lipschitz continuous strategies offer.

The search for best-responses is limited further to the space of pure Lipschitz continuous functions $\Sigma_i^{\text{Lip, p}}$. While strategies outside of $\Sigma_i^{\text{Lip, p}}$ could potentially yield higher utilities, we argue that this is a mild restriction as the best-response utility can often be attained in pure strategies under mild assumptions \citep{askouraInfiniteDimensionalPurification2019, hosoyaApproximatePurificationMixed2022}.
The utility loss with regard to pure Lipschitz continuous functions will be denoted by
\begin{align} \label{equ:utility-loss-wrt-pure-lipschitz}
	\tilde{\ell}_i^{\text{Lip, p}}\left(\beta\right) := \sup_{\beta^{\prime}_i \in \Sigma_i^{\text{Lip, p}}} \tilde{u}_i(\beta^{\prime}_i, \beta_{-i}) - \tilde{u}_i(\beta_i, \beta_{-i}).
\end{align}
We face two challenges in estimating $\tilde{\ell}_i^{\text{Lip, p}}\left(\beta\right)$. The first is that one needs to search within the infinite-dimensional function space $\Sigma_{i}^{\text{Lip,p}}$ for a best-response.
The second challenge pertains to the precise evaluation of ex-ante utility, even for a single strategy profile $\beta = (\beta_i, \beta_{-i})$. Only a precise estimate of the ex-ante utility allows for an accurate verification of equilibrium. 
Our verifier employs two core concepts to address these challenges.

First, we limit the search space by approximating the utility of a best response with a set of finite precision step functions denoted as $\Sigma_{i}^{\text{D}}$. Here, $\text{D} \in \mathbb{N}$ represents the number of steps or discretization points in both the domain and image space. This means that for a single agent, we consider strategies that only allow responses to signals with finite precision.
For a meaningful bound, it is necessary to ensure that a finite discretization adequately captures both the signal and action spaces. To facilitate our theoretical analysis and ensure the practicality of our approach, we make the following assumption.

\begin{assumption} \label{ass:bounded-signaling-and-action-spaces}
	For every $it \in L^*$, we assume there exist finitely many bounded closed intervals $\mathcal{A}_{it}^r$ and $S_{it}^r$, such that $\mathcal{A}_{it} = \bigtimes_{r=1}^{N_{\mathcal{A}_{it}}} \mathcal{A}_{it}^r$ and $S_{it} = \bigtimes_{r=1}^{N_{S_{it}}} S_{it}^r$ for dimensions $N_{\mathcal{A}_{it}}, N_{S_{it}} \in \mathbb{N}$.
\end{assumption}

Under this assumption, we can divide the signal \emph{and} action spaces into grids. For $it \in L$, we denote the grid points of the signal and action spaces as $S_{it}^\text{D}$ and $\mathcal{A}_{it}^\text{D}$, respectively. Here, $\text{D} \in \mathbb{N}$ denotes a precision parameter, where an increasing $\text{D}$ translates to an increase in grid points.

Second, we employ Monte-Carlo estimation for the ex-ante utility, as described in \autoref{equ:monte-carlo-estimation-for-ex-ante-utility}.
As a result, we only need to consider a finite set of strategies, and maximize this over an empirical estimate of the expected utilities.

However, when considering only finitely many strategies $\Sigma_{i}^{\text{D}}$ for a best-response utility estimation, one may miss large potential gains if the game and opponents' strategies are not sufficiently regular. Therefore, we require a second assumption to control errors in our verifier. One assumption is that players do not respond significantly differently to slightly different signals. Additionally, we assume that the ex-post utility functions, denoted as $u_i$, are continuous. To be more precise, we make the following assumption.

\begin{assumption} \label{ass:lipschitz-continuous-signals-and-ultimately-strategies}
	The signaling functions $\sigma_{it}: \mathcal{A}_{<t} \rightarrow S_{it}$ are Lipschitz continuous for all $it \in L$. Furthermore, there exists $K_{0t} > 0$ such that nature's probability distribution $p_{0t}$ is $K_{0t}$-Lipschitz continuous with respect to the Wasserstein distance (Definition \ref{def:wasserstein-distance}) for every $t \in T$. More specifically, $d_W \left(p_{0t}\left(\argdot \, | \, a_{<t} \right), p_{0t}\left(\argdot \, | \, a_{<t}^{\prime} \right) \right) \leq K_{0t} \left|\left| a_{<t} - a_{<t}^{\prime} \right| \right|$ for all $a_{<t}, a_{<t}^{\prime} \in \mathcal{A}_{<t}$ and some norm $\left|\left|\, \cdot \, \right|  \right|$.
\end{assumption}

With this, we can prove that the estimated utility loss by our verifier $\ell_i^{\text{ver}}(\beta)$ serves as an upper bound for $\tilde{\ell}_i^{\text{Lip, p}}(\beta)$ with sufficient resources. This allows us to prove the following convergence guarantees for our verifier.
\begin{theorem}[informal]
	\label{thm:main-approximation-result-for-lipschitz-utilities-informal}
	Let $\Gamma = \left(\mathcal{N}, T, S, \mathcal{A}, p, \sigma, u \right)$ be a multi-stage game, where Assumptions \ref{ass:bounded-signaling-and-action-spaces} and \ref{ass:lipschitz-continuous-signals-and-ultimately-strategies} hold. For a strategy profile $\beta = (\beta_i, \beta_{-i})$ with $\beta_i \in \Sigma_{i} $ and $\beta_{-i} \in \Sigma^{\text{Lip}}_{-i}$, and a continuous utility function $u_i$, we have that with a sufficiently high precision $\text{D}$ and initial simulation count $M_{\text{IS}}$ the estimated utility loss is an upper bound for the utility loss over pure Lipschitz continuous functions, that is
	\begin{align*}
		\lim_{D \rightarrow \infty} \lim_{M_{\text{IS}} \rightarrow \infty} \ell_i^{\text{ver}}(\beta) \geq  \tilde{\ell}_i^{\text{Lip, p}} (\beta).
	\end{align*}
\end{theorem}

We refer to Appendix~\ref{sec:proof-of-main-theorem} for the precise formulation of above's theorem and its proof.
One may wonder if one can substantially strengthen the above results by weakening the made assumptions. It is observed that a higher discretization and a greater number of samples lead to a lower error. While this result is intuitive, its generality is far from obvious. For example, one might question the effectiveness of our verifier in estimating the best-response utility against opponents employing strategies from the broader space of measurable strategies. Consider a simplified Stackelberg-type game with a leader and a follower, where the follower receives the leader’s action as a signal and adopts the strategy $\beta_F (a_{L1} )=1$ for $a_{L1}=\sqrt{2}$ and $-1$ otherwise. The leader's utility is given by $u_L(a)=\beta_F (a_{L1})$ and maximizes at $a_{L1}=\sqrt{2}$. However, with a finite discretization of $1/2^D$, the verifier fails to identify this best response, resulting in an estimated best-response utility of $-1$ for all $D \in \mathbb{N}$. This example illustrates that for any chosen discretization, a scenario can be constructed where the verifier's estimated utility does not converge to the true best-response utility.

Nevertheless, it is important to note that while some of the assumptions made in our analysis are relatively mild, others are not fully satisfied by some interesting settings. Let's discuss these assumptions in more detail.
Firstly, assuming the signal and action spaces to be bounded can be considered a mild restriction. Most settings already satisfy this assumption, and imposing bounds on unbounded variables usually has little practical relevance. For instance, in auctions, while bids may not have an upper bound in general, capping them to a sufficiently high value has no impact on known strategic considerations in most cases.
The second assumption deals with nature being Lipschitz continuous in the players' actions, which can also be considered a weak assumption. In many settings, nature is treated as a fixed probability distribution where events are drawn independently of the players' actions.
The assumption that the opponent strategies $\beta_{-i}$ are Lipschitz continuous usually is satisfied in our use cases, as many learning algorithms' parameterizations naturally satisfy this condition. However, for example, considering step functions for the opponents does not satisfy this assumption.
Lastly, the most stringent restrictions involve the signaling functions being Lipschitz continuous and the ex-post utilities being continuous. The assumption is common in theoretical work on (Bayesian) Nash equilibria \citep{glicksbergFurtherGeneralizationKakutani1952, reny1999existence, ui2016BayesianNashEquilibrium}. While many relevant settings such as dynamic oligopolies \citep{bylkaDiscreteTimeDynamic2000} satisfy this assumption, others do not. For instance, in markets with indivisible goods, the allocation function is discontinuous, violating both Lipschitz continuity of the signaling functions and continuity of the utilities. In Section~\ref{app:discretization}, we show empirically that the verifier also reliably estimates the utility loss in settings where the assumption is not fully satisfied.  
In addition, there are approaches implementing a smoothed version (market) games with indivisibilities \citep{kohringEnablingFirstOrderGradientBased2023} that restores Lipschitz continuity. 


\subsection{Discretization and Number of Samples}\label{app:discretization}

The utility loss estimates calculated by the verifier depend on the level of discretization and the number of simulated games. Former ensures that game state, signal, and action space closely resemble the original continuous game, whereas the latter ensures that the approximated utilities are close to their expectations (with respect to the distribution of the valuations and actions).
We illustrate which precision can be reached for multiple configurations of these parameters. For this, we deploy the verifier in the two-stage sequential sales with three participants and a first-price payment rule. 
The opponents play their equilibrium strategy, ensuring that the exact utility loss is zero.
\autoref{fig:utility_loss_analysis_2_stages} shows how the discretization size and the number of simulations influence the utility loss. The total error can be decomposed into two parts.
The first comes from restricting the search to only finitely many strategies $\Sigma_i^{\text{D}}$ and is denoted by the discretization error $\varepsilon_{\text{D}}$ (Equation~\ref{equ:discretization-error}). The second error term comes from evaluating the expectations by empirical means over simulated data and is denoted by the simulation error $\varepsilon_{M_{\text{IS}}}$ (Equation~\ref{equ:simulation-error}). 

For a low number of initial simulations, the simulation error $\varepsilon_{M_{\text{IS}}}$ dominates and strongly overestimates the utility loss. That is because our procedure chooses the maximum attainable utility over the simulated data. For a sufficient amount of simulations, the estimated utility loss becomes negative, showing the discretization error's $\varepsilon_{\text{D}}$ effect. For a finer discretization, $\varepsilon_{\text{D}}$ tends towards zero. We use a discretization of $\text{D} = 64$ and an initial number of simulations $M_{\text{IS}} = 2^{21}$ if not stated otherwise.

\begin{figure*}[h]
	\centering
	\includegraphics[width=.9\textwidth]{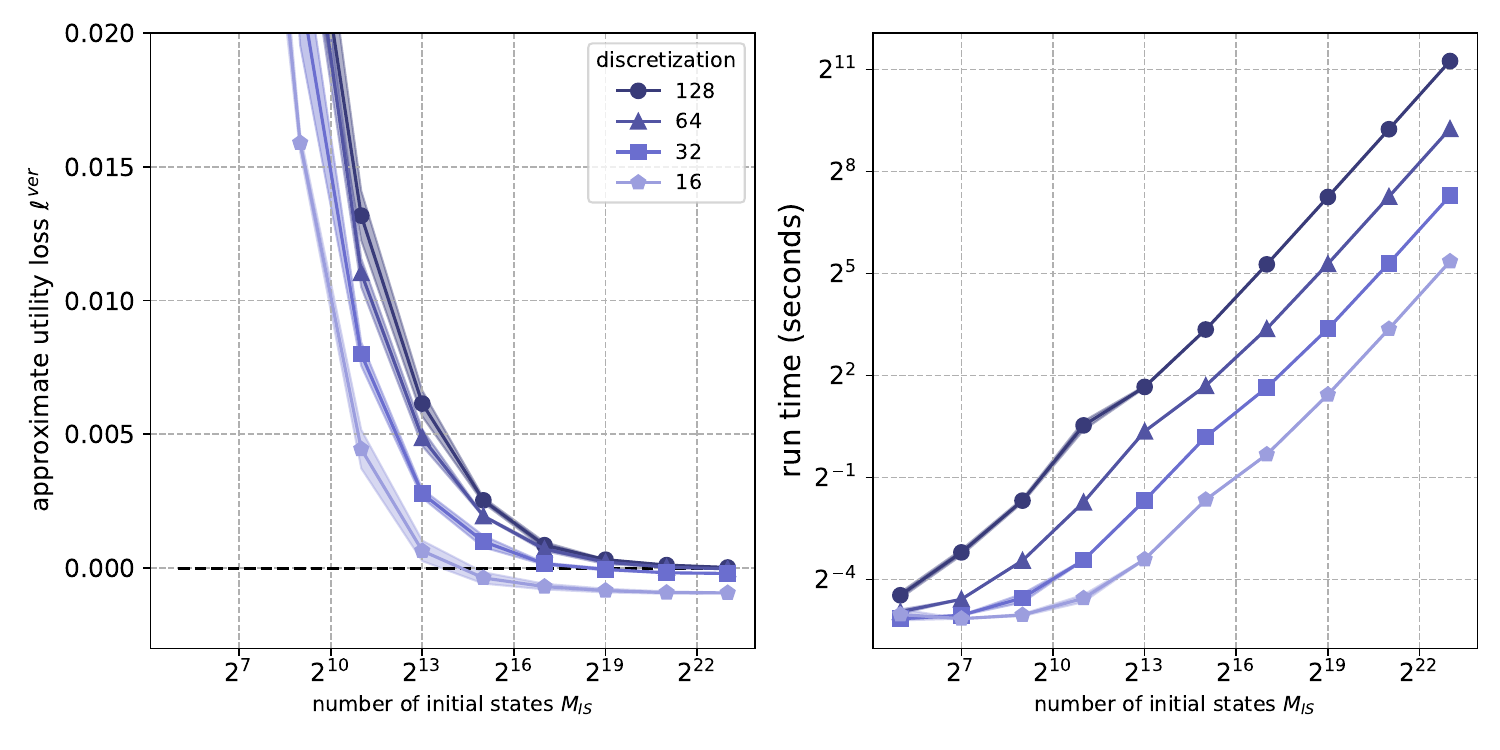}
	\caption{Approximate utility loss for different configurations of the discretization size and the number of simulations \emph{(left)} and their corresponding run times \emph{(right)}.}
	\label{fig:utility_loss_analysis_2_stages}
\end{figure*}


The verifier makes use of a vectorized implementation, which allows parallel evaluation of thousands of games. The run times eventually increase linearly as parallel batches are split into sequentially evaluated mini-batches.

It is important to acknowledge that the computational demand of the procedure is high, primarily due to the exponential growth of the tree as the number of stages and dimensions of signal and action spaces increase. Consequently, the applicability of the procedure is limited to games with a few stages only.
This restriction is not surprising given the known computational hardness of verifying equilibrium in continuous Bayesian games~\citep{caiSimultaneousBayesianAuctions2014} and of verifying equilibrium refinements in sequential complete-information games~\citep{gattiNewResultsVerification2012, hansenComputationalComplexityTrembling2010}.

Nevertheless, through parallelization and various techniques that optimize the utilization of precomputed results~\citep{johansonAcceleratingBestResponse2011}, we can achieve a high precision using a single GPU for games. For example, the runtime for two-stage games is about two minutes. A similar game with four stages takes about five hours.
Note that this level of precision is sufficient for studying a wide range of relevant continuous multi-stage games. Examples of such games include sequential auctions~\citep{krishna2009auction}, multi-stage contests~\citep{yildirim2005contests}, or sequential Colonel Blotto games~\citep{powellAllocatingDefensiveResources2007}. These games are of significant interest and can be effectively analyzed within the computational capabilities of our procedure.

\section{Experimental Results} \label{sec:experimental-results}

This section explores multi-stage games characterized by continuous signals and actions, which are fundamental in both economics and management science. To illustrate the versatility and efficacy of the approach we report results on three models: sequential auctions, elimination contests, and Stackelberg-Bertrand competitions. We demonstrate that our method effectively recovers known analytical equilibrium strategies. Furthermore, we investigate model variants involving asymmetries, interdependent priors, and risk-aversion, revealing equilibrium strategies for such unexplored models. 

\subsection{Experimental Design} \label{sec:experimental-design}

For each model considered in our analysis, there exists a known analytically derived equilibrium strategy. This provides an unambigous baseline to compare our numerical results against. We refer to these scenarios as the \emph{standard setting}. We commence our empirical analysis with the standard settings. Subsequently, we relax simplifying assumptions, leading to situations where analytical equilibrium strategies are unknown. Then, we utilize our verifier to estimate the proximity of our learned strategies to equilibrium.
We relax the analytical assumptions along three economically significant dimensions that typically complicate the derivation of analytical equilibrium. Specifically, we introduce asymmetries among agents, interdependent priors, and deviate from pure quasilinear utility functions to consider risk-averse bidders.

\subsubsection{Asymmetries}

The vast majority of economic models assume symmetric bidders and equilibrium strategies~\citep{krishna2009auction}. A model is considered symmetric if each agent faces the same decision problem, meaning the marginal signal distributions and utilities are identical for all agents. The analytical search for a symmetric equilibrium is substantially easier because it requires solving for a single strategy in the resulting differential equations rather than one for each agent. This also simplifies learning equilibrium, as a single neural network suffices instead of one for each agent.

Although asymmetric equilibria can exist in a symmetric setting, theory almost exclusively focuses on symmetric equilibria ~\citep{krishna2009auction}. For instance, in a second-price single-item auction, one player might bid the upper bound of the distribution while all others bid zero, regardless of their valuations, which also constitutes an equilibrium. While the existence of a symmetric equilibrium is often guaranteed in economic games~\citep{renyExistenceMonotonePureStrategy2011}, there is no guarantee that agents will agree on this one. We test this in the sequential second-price auction in Section~\ref{sec:sequential-auctions-asymmetry}, finding a novel asymmetric approximate equilibrium that is attracting under our learning dynamics.

While symmetric models cover some important economic games in theoretical literature, many interesting environments involve asymmetries in the prior distributions or utility functions. For example, asymmetric priors are relevant in auctions with strong and weak bidders drawn from different distributions. Additionally, asymmetries are more likely in multi-stage games when one agent moves first and its choices are revealed to opponents before they decide. We study the cases of weak and strong bidders in the elimination contest in Section~\ref{sec:elimination-contest-asymmetries} and the first- and second-mover scenario in the Stackelberg-Bertrand competition in Section~\ref{sec:stackelberg-bertrand-standard-setting}. 


\subsubsection{Interdependent prior distributions}

The independent private values/costs model is the predominant assumption in auction and contest theory. In this model, each agent knows the value or cost of an object precisely and only to itself, and the prior distribution is a product distribution. This greatly simplifies the derivation of equilibrium strategies, as an agent’s knowledge of its own type does not provide additional information about other agents' types.

First, we drop the assumption of private values/costs and assume that an agent receives noisy information in the form of its signal. Thus, the agent needs to infer the distribution of values/costs given its signal. Second, we allow for non-product prior distributions, meaning the distribution of opponents' types changes based on the agent's type. This requires reasoning over a different conditional prior distribution for every type. In general, dealing with conditional probability distributions is mathematically challenging~\citep{ackermanNoncomputableConditionalDistributions2011}.

For our experiments, we extend two different interdependent prior distributions from the common values model introduced by \citet{wilsonCompetitiveBiddingDisparate1969} and \citet{milgromTheoryAuctionsCompetitive1982} to multiple stages. In both settings, all bidders share the same valuation $v$ but receive noisy information about $v$ as an observation $x_i$. We use the following prior distributions to model the valuations for winning an item in the sequential auction and elimination contest, as well as the costs in the Bertrand competition.

The first model is an instance of the so-called \emph{mineral rights} model~\citep{krishna2009auction}. \citet{reeceCompetitiveBiddingOffshore1978} uses this model to describe competitive bidding for offshore oil drilling rights. The overall value of the resource is the same for all competitors but remains uncertain until the \emph{ex post} stage. The true common value $v$ is randomly drawn but unobserved. Each competitor independently estimates the true valuation from a distribution conditioned on $v$, given the observation $x_i$. Specifically, the valuation or cost $v$ is uniformly drawn from the interval $[0, 1]$, while the bidders' observations $x_i$ are uniformly drawn from $[0, 2v]$.

The second model concerns \emph{affiliated} type distributions~\citep{krishna2009auction}, meaning a high value of one bidder's estimate makes high values of others' estimates more likely. Consider the auction of a piece of art: a buyer who perceives the item as particularly fine may conclude that other bidders also value it highly~\citep{milgromTheoryAuctionsCompetitive1982}. In this case, the observations are positively affiliated. For our experiments, we uniformly draw a latent variable $s$ and estimation variables $z_i$ from $[0, 1]$. The observation for each agent is given by $x_i = z_i + s$. The common valuation or cost $v$ is determined by the mean of all observations $x_i$, that is, $v = \frac{1}{N} \sum_{i \in \mathcal{N}} x_i$.

\subsubsection{Risk aversion}

Basic models of auctions and contests assume a quasi-linear utility function ~\citep{krishna2009auction, vojnovicContestTheoryIncentive2016}, where bidders maximize their payoff and utility is linear in price. 
Risk aversion captures the tendency of individuals to trade off a potentially lower payoff for increased chances of winning~\citep{bernoulliExpositionNewTheory1954}. It influences a wide range of economic decisions, from investment and savings to insurance and consumption, and has become a central behavioral effect studied in economic theory~\citep{arrowEssaysTheoryRiskbearing1971, prattRiskAversionSmall1964}.
Incorporating risk aversion into utility models increases their complexity because it introduces non-linearities and requires consideration of individuals' subjective preferences. Consequently, analytical solutions often become intractable.

In our experiments, we use models with \emph{constant absolute risk aversion} (CARA)~\citep{prattRiskAversionSmall1964}. The CARA utility function is characterized by a constant level of absolute risk aversion, meaning that the individual's attitude towards risk does not change with the level of wealth. It is commonly used to model preferences in insurance purchases and option pricing~\citep{gerberUtilityFunctionsRisk1998}. Mathematically, it can be modeled by a function $h_{\rho}: \mathbb{R} \rightarrow \mathbb{R}$ with
$h_{\rho}(x) = \frac{1 - \exp(-\rho x)}{\rho}$, where $\rho \in [0, \infty)$ describes the level of risk aversion. We apply $h_{\rho}$ to the quasi-linear utilities from the standard settings to model risk-averse agents, that is, $u_{i, \text{risk}}(a) = h_{\rho}(u_i(a))$. For $\rho = 0.0$, the function $h_{\rho}$ is the identity, which corresponds to risk-neutral payoff maximization as in a standard quasi-linear model.


\subsection{Sequential Auctions}

The game-theoretical analysis of sequential auctions goes back more than 40 years and describes a central problem in auction theory \citep{weber1981multiple, hausch1988model, menezes2003synergies, cai2007non}. Here we revisit the model introduced by \citet{milgrom2000theoryb} (also described in \citet{krishna2009auction}) who consider first- and second-price auctions. Yet, in order to find equilibrium strategies analytically, strong assumptions needed to be made. After we discuss the standard setting, we relax the assumptions in the standard setting and arrive at models, for which we are not aware of previously known equilibrium strategies.

\subsubsection{Standard setting}\label{sec:sequential-auction-standard-setting}

Let $T$ be the number of homogeneous units for sale and let there be $N > T$ bidders. In each stage $t$, there is exactly one unit for sale. Bidders are only interested in winning a single item -- thus, a bidder drops out of the auction after having acquired an item -- and they are privately informed of their valuation $v_i$ before the beginning of the first stage. Based on the submitted sealed-bids $a_{\cdot t}$ in each stage, an auction mechanism calculates the allocation of the good and the price $p_{it}(a_{\cdot t})$ for the winner.

Assuming risk-neutrality with a utility of $v_i - p_{it}(a_t)$ for the winner and zero for the losers, an analytical solution that only depends on a bidder's valuation can be derived:
\begin{proposition}[\citet{krishna2009auction}]\label{prop:sa_bne}
	Suppose bidders have unit-demand drawn uniformly from $[0,1]$ and $T$ units are sold by means of sequential auctions. Then, the following strategies constitute symmetric equilibria in the $t$-th stage:
	\begin{enumerate}
		\item First-price:
		\begin{equation*}
			\beta_{it}(v_i) = \frac{N-T}{N-t+1} v_i,
		\end{equation*}
		\item Second-price:
		\begin{equation*}
			\beta_{it}(v_i) = \frac{N-T}{N-t} v_i.
		\end{equation*}
	\end{enumerate}
\end{proposition}
Because the ratio of supply to demand is decreasing over time, bidders are forced to increase their bids in both mechanisms.
Commonly, one assumes that the prices from previous stages are revealed. 

We evaluate the learning algorithms for the first- and second-price mechanism and for different numbers of stages $T$. For simplicity, we set the number of bidders $N$ to $T+1$ such that there remains competition in the final stage.   
\autoref{fig:sequential_sales} depicts exemplary strategies in a two-stage auction. 
The full results for PPO and \textsc{Reinforce}, first- and second-price sequential auctions, and different numbers of objects $T$ can be found in~\autoref{tab:table_sequential_auction}. We train each setting for $10,000$ iterations and report the mean and standard deviation over ten different runs. We use the same set of hyperparameters for almost all settings (Section~\ref{sec:hyperparameters}). The utility loss is very low for all settings considered.

\begin{figure*}[]
	\centering
	\includegraphics[width=.9\textwidth]{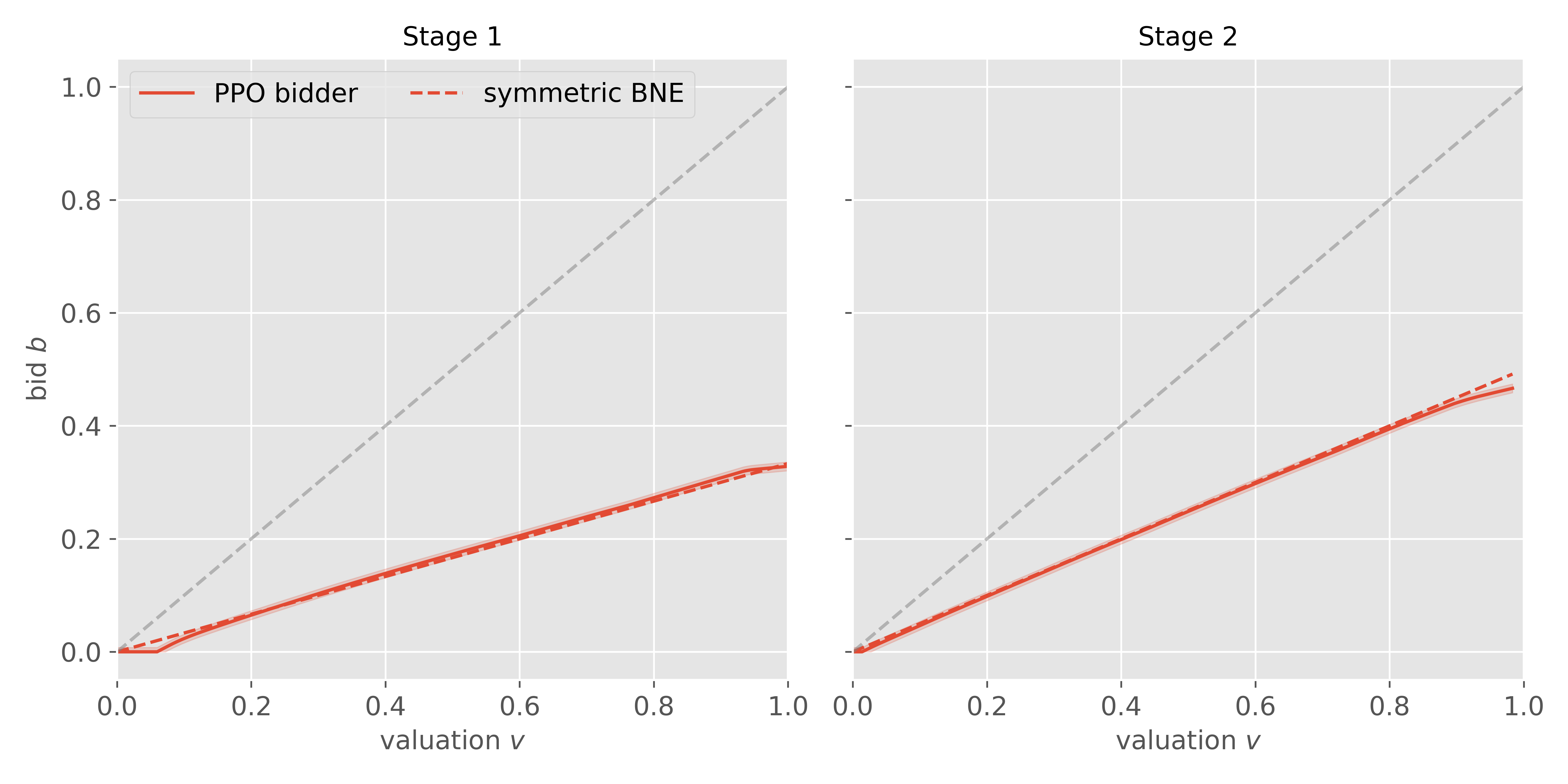}
	\caption{Equilibrium and PPO-based learned strategies in sequential sales with a first-price mechanism, two-stages, and three bidders.}
	\label{fig:sequential_sales}
\end{figure*}

\begin{table}[]
\centering
\caption{Learning results for sequential sales with various numbers of stages. We report the mean $L_2^\text{avg}$ (\autoref{equ:dated-l2-distance}) and utility loss $\ell^{\text{equ}}$ with respect to the analytical symmetric equilibirum (\autoref{equ:ex-ante-util-loss-in-equilibrium}), as well as the estimated utility loss $\ell^{\text{ver}}$ (\autoref{equ:final-utility-loss-estimation-from-verifcation}) across ten runs together with the standard deviations.}
\label{tab:table_sequential_auction}
\begin{tabular}{lllrr}
\toprule
mechanism & $T$ & metric & \textsc{Reinforce} &             PPO \\
\midrule
first-price  & 1 & $L_2^\text{avg}$ &  0.0060 (0.0010)  &   0.0047 (0.0009)    \\ \cdashline{3-5}
       &   & $\ell^{\text{equ}}$ &    0.0001 (0.0002) &   0.0001 (0.0001) \\ \cdashline{3-5}
       &   & $\ell^\text{ver}$ &    0.0003 (0.0001) &   0.0003 (0.0001) \\ \cdashline{2-5} \Tstrut
       & 2 & $L_2^\text{avg}$ &  0.0098 (0.0046)  &   0.0054 (0.0020)    \\ \cdashline{3-5}
       &   & $\ell^{\text{equ}}$ &    0.0002 (0.0002) &     0.0000 (0.0002) \\ \cdashline{3-5}
       &   & $\ell^\text{ver}$ &    0.0013 (0.0002) &   0.0003 (0.0003) \\ \cdashline{2-5} \Tstrut
       & 4 & $L_2^\text{avg}$ &  0.0110 (0.0044)  &   0.0056 (0.0034)    \\ \cdashline{3-5}
       &   & $\ell^{\text{equ}}$ &    0.0003 (0.0003) &      0.0000 (0.0003) \\ \cdashline{3-5}
       &   & $\ell^\text{ver}$ &   -0.0005 (0.0013) &   -0.0016 (0.0010) \\ \hline \Tstrut
second-price & 1 & $L_2^\text{avg}$ &  0.0091 (0.0007)  &    0.0063 (0.0022) \\ \cdashline{3-5}
	   &   & $\ell^{\text{equ}}$  &  0.0001 (0.0002) &   0.0002 (0.0001) \\               \cdashline{3-5}
       &   & $\ell^\text{ver}$ &          0.0000 (0.0000) &        0.0000 (0.0000) \\ \cdashline{2-5} \Tstrut
       & 2 & $L_2^\text{avg}$ &  0.0075 (0.0028)  &    0.0068 (0.0028) \\ \cdashline{3-5}
       &   &$\ell^{\text{equ}}$ &   0.0001 (0.0002) &  -0.0001 (0.0003) \\ \cdashline{3-5}
       &   & $\ell^\text{ver}$ &    0.0033 (0.0005) &    0.0020 (0.0006) \\ \cdashline{2-5} \Tstrut
       & 4 & $L_2^\text{avg}$ &  0.0140 (0.0036)  &    0.0072 (0.0031) \\ \cdashline{3-5}
       &   & $\ell^{\text{equ}}$  &  0.0002 (0.0002) &     0.0000 (0.0004) \\ \cdashline{3-5}
       &   & $\ell^\text{ver}$ &     0.0050 (0.0003) &   0.0039 (0.0008) \\
\bottomrule
\end{tabular}
\end{table}

\subsubsection{Asymmetries} \label{sec:sequential-auctions-asymmetry}
The existing literature on (sequential) auctions predominantly assumes symmetric priors and symmetric equilibrium strategies \citep{hausch1986multi, bikhchandani1988reputation, milgrom2000theory, mezzetti2011sequential, trifunovic2014sequential, rosato2023loss}. While assuming symmetric strategies is convenient, it's important to note that it may not be always justified. While it is known that the symmetric equilibrium is the unique equilibrium in the single-stage first-price auction \citep{chawlaAuctionsUniqueEquilibria2018}, there are several equilibria when allowing asymmetric strategies in the single-stage second-price auction \citep[p. 118]{krishna2009auction}. The work by \citet{katzman1999two} is one of the exceptions, that obtains an asymmetric equilibrium in a model with symmetric priors and multi-unit demand. 

We analyze a two-stage second-price sequential auction with uniform prior and risk-neutral bidders, where each agent trains its own neural network.
Our empirical findings show that both the \textsc{Reinforce} and PPO algorithms consistently converge to an approximate asymmetric equilibrium, as depicted in Figure~\ref{fig:asymmetric_second_price_sequential_auction}. We could not find such equilibria reported in the existing literature.

The estimated utility loss is very low for all agents (see Table \ref{tab:table_sequential_auction_asymmetric_second_price}).
This equilibrium exhibits two bidders with higher equilibrium bid functions compared to the symmetric equilibrium, while another bidder has a lower equilibrium bid function. In the second stage, all agents employ a truthful strategy. The latter bidder wins less frequently but achieves a higher payoff if he or she wins. The estimated utility for each agent is nearly identical, with a mean of $0.2476$ for the lower bidding agent and $0.2466$ for the two higher bidding ones.

Furthermore, the expected utility closely matches the expected utility in the symmetric equilibrium, which is $0.25$ analytically and $0.2465$ in the approximated equilibrium reported in Section~\ref{sec:sequential-auction-standard-setting}. We did not observe a similar effect in sequential auctions using a first-price rule, where all bidders converge approximately to the symmetric equilibrium.
This empirical observation is particularly intriguing as it shows the existence of an approximate asymmetric equilibrium, even when the overall decision problem is symmetric. Furthermore, the learning dynamics consistently converge to the approximate asymmetric equilibrium in this case.


\begin{table}
    \centering
    \caption{Approximated utility losses for all agents $i$ in a second-price two-stage auction with three bidders that end up in an asymmetric equilibrium.}
    \label{tab:table_sequential_auction_asymmetric_second_price}
    \begin{tabular}{llrr}
    \toprule
    $i$ & metric          & \textsc{Reinforce} & PPO\\
    \midrule
    1 & $\ell^\text{ver}$ & 0.0009 (0.0003)    & 0.0006 (0.0003) \\ \cdashline{1-4}
    2 & $\ell^\text{ver}$ & 0.0008 (0.0003)    & 0.0007 (0.0001) \\ \cdashline{1-4}
    3 & $\ell^\text{ver}$ & 0.0009 (0.0002)    & 0.0005 (0.0002) \\
    \bottomrule
    \end{tabular}
\end{table}

\begin{figure*}[ht]
	\centering
	\includegraphics[width=.9\textwidth]{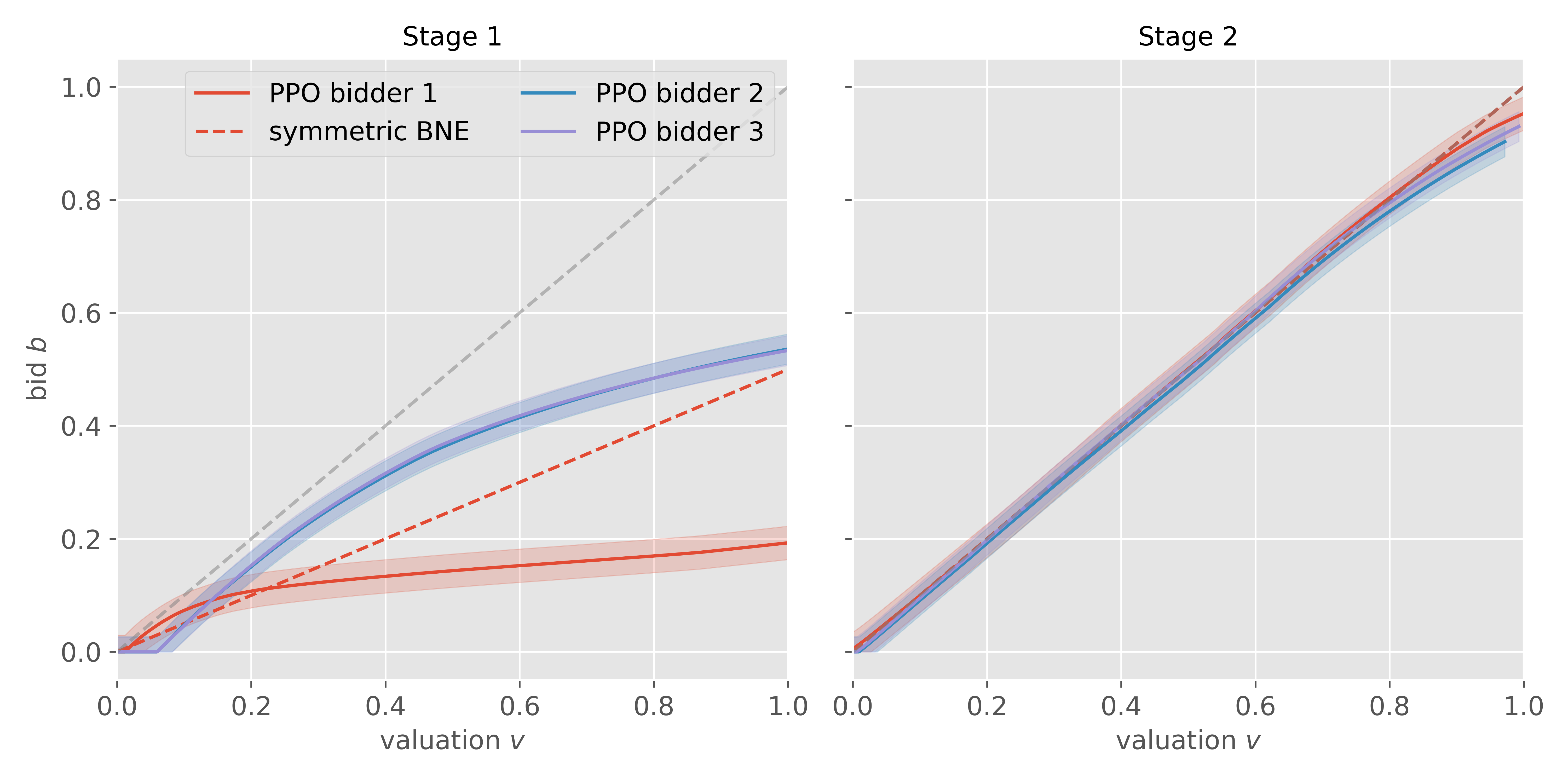}
	\caption{Asymmetric PPO-based learned strategies in sequential sales with a second-price mechanism, two stages, and three bidders.}
	\label{fig:asymmetric_second_price_sequential_auction}
\end{figure*}

\subsubsection{Risk aversion and interdependent priors} \label{sec:sequential_sales-risk}

In this section, we study sequential auctions with a combination of interdependent prior distribution and risk-averse bidders. 
Due to the complexity arising from combinations of these model assumptions, analytical derivations become challenging.
In his paper on risk aversion in sequential auctions, \citet{mezzetti2011sequential} writes that affiliation `is not very tractable'. He also focuses on a specific form of aversion to price risk.
We report results on sequential auctions with risk aversion and affiliated or mineral rights priors. We are not aware of any existing equilibrium strategies for the reported settings in the literature.

Table \ref{tab:table_sequential_auction_interdependence_plus_risk} shows estimated utility losses in the mineral rights and affiliated values settings, considering different levels of risk aversion. Additionally, Figure \ref{fig:common-values-plus-risk} presents an approximate equilibrium strategy in the mineral rights setting with a second-price rule and a risk parameter of $\rho = 2.0$.
One observes a significant increase in bids from the first to the second stage of the auction. Presumably, this can be attributed to two factors. Firstly, bidders have a second chance to win the item in the second stage, leading to lower bids in the first stage. Secondly, reaching the second stage means that the first stage's winner had a higher estimation of the true value, leading bidders to perceive their initial estimation as an underestimation of the true valuation. Consequently, they bid more aggressively in the second stage.

\begin{table}
\centering
\caption{Approximated utility losses of the experiments with valuation interdependencies and risk-averse bidders. Here, no analytical equilibrium is available for comparison.}
\label{tab:table_sequential_auction_interdependence_plus_risk}
\begin{tabular}{llllrr}
	\toprule
	mechanism & prior & risk $\rho$ & metric &          \textsc{Reinforce} &    PPO\\
	\midrule
	second-price & mineral rights  & 0.5 & $\ell^\text{ver}$ & 0.0015 (0.0012) &   0.0018 (0.0013) \\ \cdashline{3-6}
	&                             & 1.0 & $\ell^\text{ver}$ &  0.0015 (0.0012) &   0.0019 (0.0012) \\ \cdashline{3-6}
	&                             & 2.0 & $\ell^\text{ver}$ &  0.0016 (0.0012) &   0.0017 (0.0011) \\ \hline \Tstrut
	first-price & affiliated           & 0.5 & $\ell^\text{ver}$ &  0.0003 (0.0001) &  -0.0001 (0.0000) \\ \cdashline{3-6}
	&                             & 1.0 & $\ell^\text{ver}$ &  0.0003 (0.0001) &  -0.0001 (0.0000) \\ \cdashline{3-6}
	&                             & 2.0 & $\ell^\text{ver}$ &  0.0003 (0.0001) &  -0.0001 (0.0000) \\ 
	\bottomrule
\end{tabular}

\end{table}

\begin{figure*}[ht]
	\centering
	\includegraphics[width=.9\textwidth]{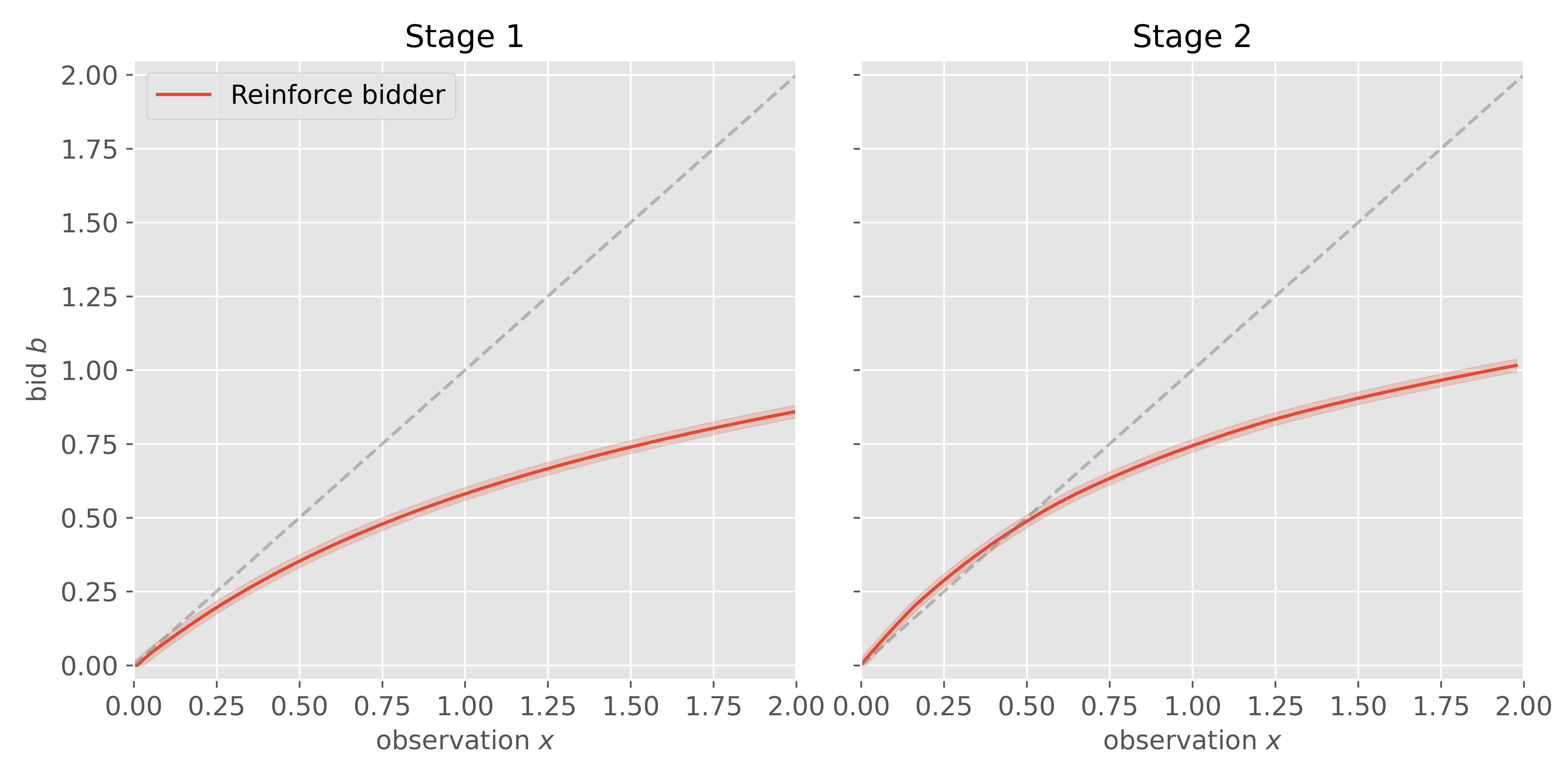}
	\caption{\textsc{Reinforce}-based learned strategies in sequential sales with a second-price mechanism, a mineral rights prior, and three risk-averse bidders with $\rho = 2.0$ in two stages.}
	\label{fig:common-values-plus-risk}
\end{figure*}

\subsection{Elimination Contest} \label{sec:elimination-contest}

Contests are used to model lobbying, political campaigns, and R\&D competitions, among others \citep{konrad2009strategy, corchon2018contest}. Equilibrium analysis has been a central approach to analyzing competition in contests. 
Almost the entire literature focuses on Nash equilibria and models contests as complete-information games \citep{corchon2018contest}. However, similar to auction theory, much of the strategic complexity in contests is due to the fact that contestants only have incomplete information about their competitors, and they aim to find a bid function for a continuous set of signals. 

Importantly, here we analyze multi-stage contests. In many real-life contests, players are initially divided into a few groups, they first compete within their subgroups, and then winners from each group compete again in later stages.
The central question is whether pre-commitments in sequential or elimination contests incentivize players to invest more effort than in a single-stage contest. Furthermore, it is important to understand the effects of information revelation between stages.

We study the signaling contest as introduced by \citet{zhang2008simultaneous} as it has some interesting properties.
He analyzes the effect of publishing private information (valuations) or submitted bids (signaling effect) on the equilibrium in a two-round elimination contest. The work is inspired by \citet{moldovanu2006contest} who leave open to analyze the role of information in contests with multiple rounds. While the paper studies a fairly general case, the equilibrium strategies are provided as abstract integral solutions. Therefore, we consider a special instance for the standard setting where we can solve the related integrals and provide an explicit equilibrium strategy. Consecutively, we study its extension in each of the three mentioned dimensions: asymmetries, interdependent priors, and risk-averse contestants. We are not aware of known equilibrium strategies for these adaptations. 

\subsubsection{Standard setting}
Consider $N = 4$ risk-neutral bidders that privately learn their valuations $v=(v_1, v_2, v_3, v_4)$ for the prize, which are independently and uniformly distributed on the interval $[1.0, 1.5]$. 
In the first stage, they compete within two equally sized groups of two bidders by simultaneously submitting bids. Here, all bidders pay for their bids regardless of success (all-pay auction). Then, the two winners compete in the final round.
Before the final round, either their true valuations or their bids are revealed to the others. 
Thus, the finalists can now base their decisions on their private information about the prize and the public information about their opponent. 
In the second round, the players' winning probabilities are equal to the ratio of their own bid to the cumulative bids of the finalists (Tullock contest).

To reduce sample variance, we directly model their utility as the expected utility for given valuations and efforts by weighting their valuation for the prize by the probability of winning. For a finalist $i$, we have $u_i(v_i, a_i, a_{-i}) = \frac{a_{i2}}{a_{i2} + a_{j2}} v_i - a_{i2} - a_{i1}$, where $j$ denotes the other finalist, and for a non-finalist $k$, we have $u_k(v_k, a_k, a_{-k}) = -a_{k1}$.
\citet{zhang2008simultaneous} derived the following equilibrium:
\begin{proposition}[\citet{zhang2008simultaneous}]\label{prop:sc_bne}
	Consider a four-bidder two stage-contest as described above. Let $i$ be some bidder, and denote with $j$ the first round's winner of the other group. Then there exists a separating equilibrium for both information cases, which is given by the following.
	\begin{enumerate}
		\item Assuming the true valuations are revealed after the first stage, i.e., $\sigma_{i2}(a_{\cdot 1}) = v_j$, we have the following symmetric equilibrium:
		\begin{align*}
			\beta_{i1}(v_i) &= \text{WE}(v_i) \\
			\beta_{i2}(v_i, v_j)   &= \frac{v_i^2 v_j}{\left(v_i + v_j\right)^2}
		\end{align*}
		\item Assuming the winning bids of the other group are revealed after the first stage, i.e., $\sigma_{i2}(a_{\cdot 1}) = a_{1j}$, we have the following equilibrium:
		\begin{align*}
			\beta_{i1}(v_i) &= \text{WE}(v_i) + \text{SE}(v_i) \\
			\beta_{i2}(v_i, a_{j1}) &= \frac{v_i^2 \beta_{i1}^{-1}(a_{j1})}{\left(v_i + \beta_{i1}^{-1}(a_{j1})\right)^2}
		\end{align*}
	\end{enumerate}
	where the functions $\text{WE}$ and $\text{SE}$ are defined as follows:
	\begin{align*}
		\text{WE}(v_i) &= 27\,\log \left(v_i+\frac{3}{2}\right)-\frac{17\,v_i}{2}-\frac{43\,\log \left(\frac{5}{2}\right)}{4}+\frac{7\,v_i^2 }{2}-2\,v_i^3 -4\,\log \left(v_i+1\right)\,{\left(v_i^4 -1\right)} \\
		& +4\,\log \left(v_i+\frac{3}{2}\right)\,{\left(v_i^4 -\frac{81}{16}\right)}+7 \\
		\text{SE}(1) &= 0 \\
		\text{SE}(v_i) &= 17\,\log \left(5\right)-8\,\log \left(v_i+1\right)-9\,\log \left(v_i+\frac{3}{2}\right) -17\,\log \left(2\right)-16\,v_i+8\,v_i^2 \,\log \left(v_i+1\right) \\
		&+16\,v_i^3 \,\log \left(v_i+1\right)-16\,v_i^4 \,\log \left(v_i+1\right)-8\,v_i^2 \,\log \left(v_i+\frac{3}{2}\right) -16\,v_i^3 \, \log \left(v_i+\frac{3}{2}\right) \\
		&+16\,v_i^4 \,\log \left(v_i+\frac{3}{2}\right) -\frac{135}{2\,v_i+3}+18\,v_i^2 -8\,v_i^3 +33 \text{ for } v_i \in (1, 1.5].
	\end{align*}
\end{proposition}


The results are presented in \autoref{tab:table_signaling_contest} and again show a very small utility loss. If the estimated utility loss is negative, this means the learned strategy is better than the best finite precision step function strategy.

\begin{table}
\centering
\caption{Learning results in the signaling contest. We again report the mean $L_2$ loss for both stages and the utility losses $\ell^\text{equ}$ and $\ell^\text{ver}$ and their standard deviations over ten runs.}
\label{tab:table_signaling_contest}
\begin{tabular}{llrr}
\toprule
public information & metric &   \textsc{Reinforce}  &    PPO     \\
\midrule
valuations & $\ell^\text{equ}$ &    0.0001 (0.0002) &     0.0000 (0.0001) \\ \cdashline{2-4} \Tstrut
             & $\ell^\text{ver}$ &    0.0013 (0.0004) &  -0.0003 (0.0003) \\ \cdashline{2-4} \Tstrut
             &  $L_2^{S1}$ &    0.0059 (0.0015) &   0.0029 (0.0011) \\ \cdashline{2-4} \Tstrut
             &  $L_2^{S2}$ &     0.0060 (0.0010) &   0.0013 (0.0004) \\ \hline \Tstrut
bids & $\ell^\text{equ}$ &    0.0002 (0.0001) &      0.0000 (0.0001) \\ \cdashline{2-4} \Tstrut
             & $\ell^\text{ver}$ &   -0.0008 (0.0004) &     0.0000 (0.0004) \\ \cdashline{2-4} \Tstrut
             &  $L_2^{S1}$ &    0.0072 (0.0012) &   0.0029 (0.0008) \\ \cdashline{2-4} \Tstrut
             &  $L_2^{S2}$ &     0.0066 (0.0010) &   0.0014 (0.0003) \\
\bottomrule
\end{tabular}
\end{table}




\subsubsection{Asymmetries} \label{sec:elimination-contest-asymmetries}
Asymmetries in contests can manifest as weak and strong agents~\citep{baikEffortLevelsContests1994} or through information asymmetry~\citep{smithLogicAsymmetricContests1976}. We examine the elimination contest with asymmetric contestants, specifically focusing on a scenario with weak and strong participants. In this setup, two weak agents have valuations drawn uniformly from $[1.0, 1.5]$, and two strong agents have valuations drawn uniformly from $[1.0, 2.0]$. The two groups in the initial round each consist of one weak and one strong agent. \citet{baikEffortLevelsContests1994} highlights that the combination of weak and strong agents with incomplete information is particularly challenging to solve analytically. We are unaware of a known equilibrium strategy for this setting.

\begin{table}
\centering
\caption{Learning results in the elimination contest with asymmetric agents (two strong and two weak competitors). We report the estimated utility loss $\ell^\text{ver}$ with its standard deviations over ten runs.}
\label{tab:table_signaling_contest_asymmetric}
\begin{tabular}{llrr}
\toprule
agent & metric &    \textsc{Reinforce}  &          PPO        \\
\midrule
1 & $\ell^\text{est}$ &    0.0064 (0.0035) &   0.0080 (0.0035) \\ \cdashline{1-4} \Tstrut
2 & $\ell^\text{est}$ &    0.0028 (0.0025) &   0.0040 (0.0026) \\ \cdashline{1-4} \Tstrut
3 & $\ell^\text{est}$ &    0.0057 (0.0030) &   0.0081 (0.0046) \\ \cdashline{1-4} \Tstrut
4 & $\ell^\text{est}$ &    0.0026 (0.0028) &   0.0040 (0.0024) \\
\bottomrule
\end{tabular}
\end{table}

Table~\ref{tab:table_signaling_contest_asymmetric} shows the estimated utility loss, which is low for all agents. Figure~\ref{fig:contest-asymmetries-strategies} shows the learned strategies, where all agents use a \textsc{Reinforce} learner.
It is observable that the two weak (strong) agents employ nearly identical strategies in both stages. In the first round, the weak agents bid zero for low valuations but bid more aggressively than the strong agents for higher ones.

The bidding behavior can be explained as follows. For low valuations, weak agents have a very low chance of winning. Since they first compete in an all-pay auction, they maximize their expected utility by not participating. However, as their valuation increases, winning becomes more attractive, prompting weak bidders to start bidding above zero. To offset their strategic disadvantage due to lower valuation distributions, they bid more aggressively than the strong bidders, resulting in higher bids. The second-round strategies are very similar for all contestants, increasing approximately linearly with their valuation and the opponent's bid.

While asymmetric settings are often known to have multiple equilibria~\citep{bichler2023ijoc}, we observe that the algorithms consistently converge to approximately the same equilibrium.

\begin{figure*}
	\centering
	\subfigure{
		\label{fig:contest-asymmetries-round-1-strategies}
		\includegraphics[width=0.44\columnwidth]{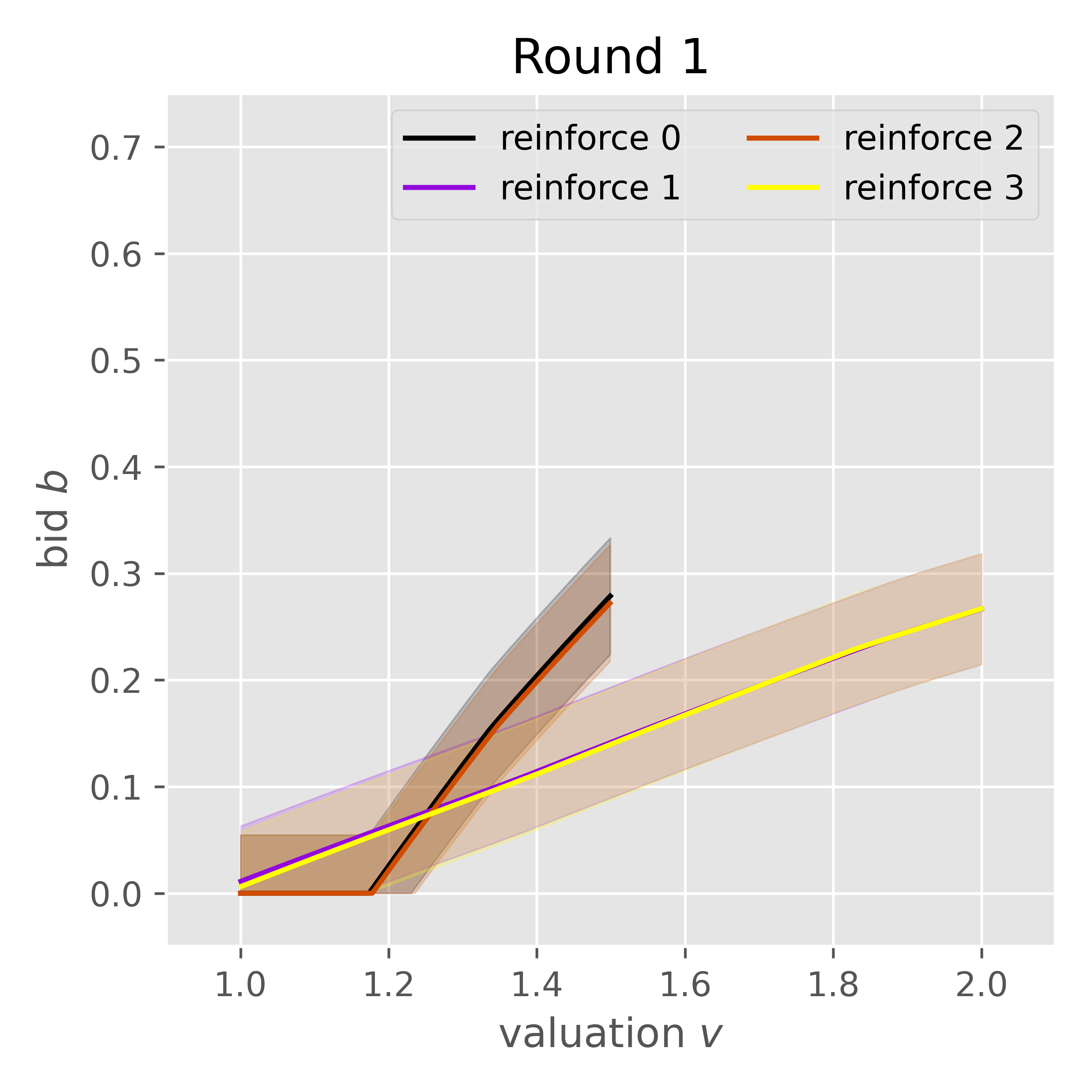}
	}
	\hfill
	\subfigure{
		\label{fig:contest-asymmetries-round-2-strategies}
		\includegraphics[width=0.44\columnwidth]{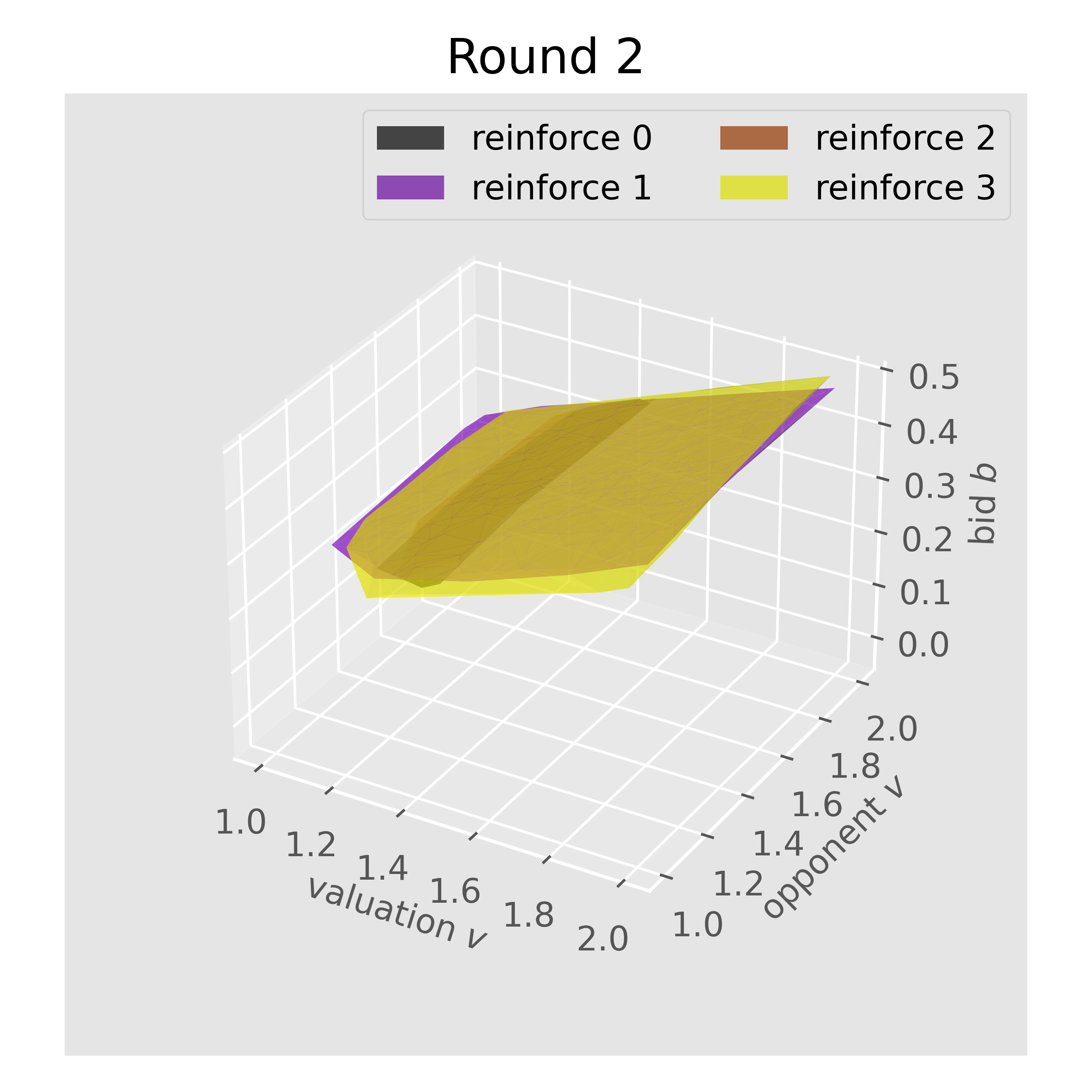}
	}
	\caption{Asymmetric \textsc{Reinforce}-based learned strategies in the elimination contest with two weak and two strong contestants.}
	\label{fig:contest-asymmetries-strategies}
\end{figure*}

\subsubsection{Interdependent priors}

Several studies have examined contests with incomplete information. Notably, \citet{amannAsymmetricAllPayAuctions1996} and \citet{noussairBehaviorAllpayAuctions2006} specified explicit equilibrium strategies with independent prior distributions. Furthermore, \citet{ewerhartUniqueEquilibriumContests2020} proves the existence and uniqueness of a pure strategy equilibrium for various additional prior distributions, including some interdependent priors. However, to the best of our knowledge, an explicit equilibrium strategy for contests with interdependent prior distributions has not yet been derived. Importantly, we consider a two-stage contest, which further increases the complexity.

\begin{table}
\centering
\caption{Learning results in the elimination contest with different interdependent prior distributions. We report the estimated utility loss $\ell^\text{ver}$ with its standard deviations over ten runs.}
\label{tab:table_signaling_contest_interdependencies}
\begin{tabular}{lllrr}
\toprule
public information & prior & metric &       \textsc{Reinforce}             &      PPO             \\
\midrule
observations & mineral rights & $\ell^\text{est}$ &   -0.0012 (0.0002) &  -0.0017 (0.0002) \\ \cdashline{2-5} \Tstrut
             & affiliated & $\ell^\text{est}$ &        0.0330 (0.0020) &   0.0167 (0.0012) \\ \hline \Tstrut
bids & 		   mineral rights & $\ell^\text{est}$ &   -0.0046 (0.0006) &  -0.0057 (0.0003) \\ \cdashline{2-5} \Tstrut
             & affiliated & $\ell^\text{est}$ &        0.0039 (0.0017) &   0.0074 (0.0014) \\
\bottomrule
\end{tabular}
\end{table}

We study the elimination contest where agents have interdependent prior distributions. 
Table~\ref{tab:table_signaling_contest_interdependencies} shows the results for the mineral rights and affiliated prior with different forms of public information.
The utility loss is small in all considered settings. However, under the mineral rights prior, the utility loss is lower than with affiliated valuations.

Figure~\ref{fig:contest-interdependencies-affiliated-winning-bids-strategies} shows the first-round strategies after $10,000$ iterations for a PPO learner with interdependent priors, where the winner's bids are published after the first round. In both cases, the strategies are strictly increasing with the observation, which was expected since a higher observation increases the chances of a higher valuation. Interestingly, the bids are much higher with an affiliated prior, and the bidding strategy is non-linear, showing an accelerated rate of increase for higher observations. The first-round strategies under the mineral rights prior are approximately linear. We observed that the second-round strategies are approximately linear for both cases; however, the bids remain higher with the affiliated prior.

\begin{figure*}
	\centering
	\subfigure[mineral rights prior]{
		\label{fig:contest-interdependencies-affiliated-winning-bids-round-2-strategies}
		\includegraphics[width=0.44\columnwidth]{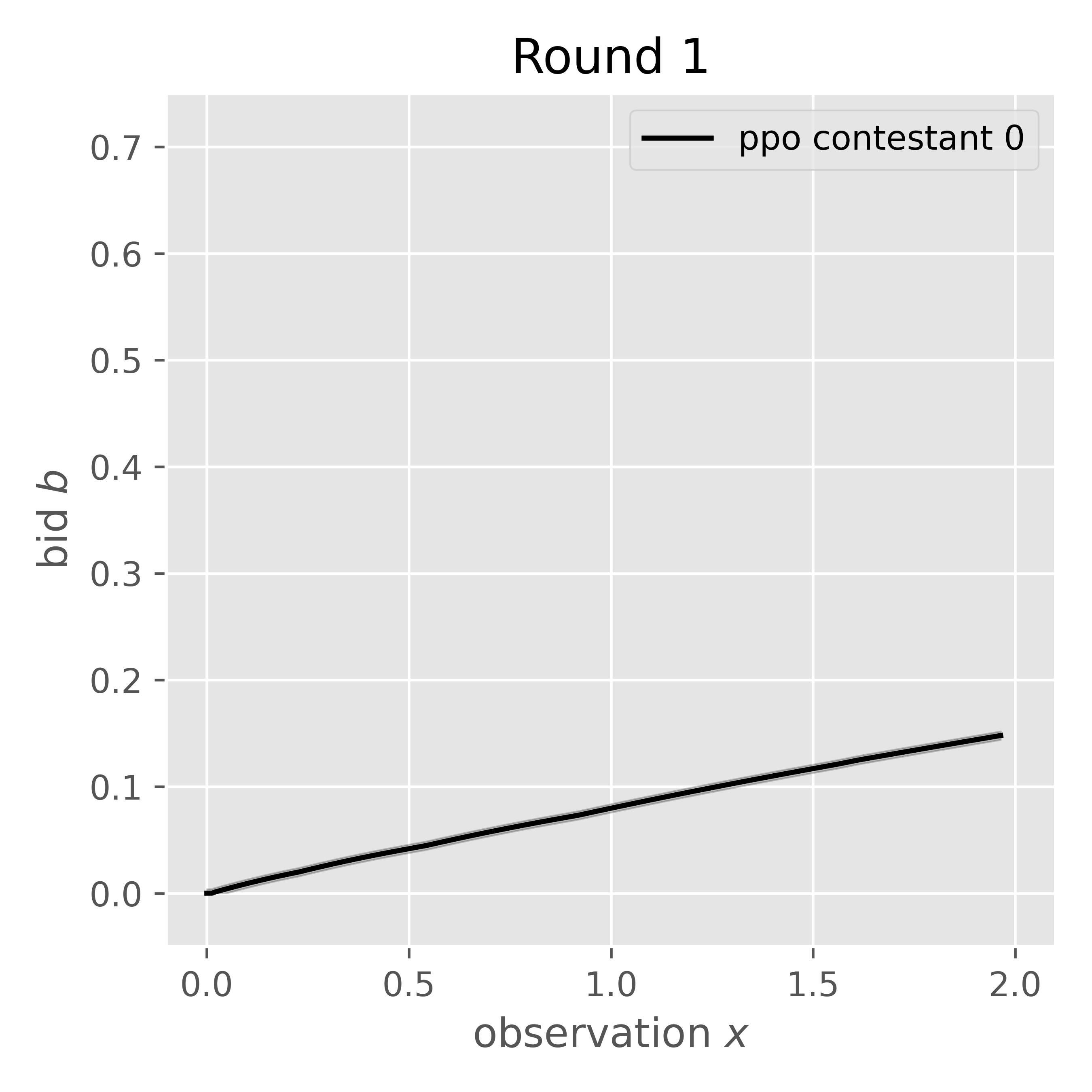}
	}
	\hfill
	\subfigure[affilited prior]{
		\label{fig:contest-interdependencies-affiliated-winning-bids-round-1-strategies}
		\includegraphics[width=0.44\columnwidth]{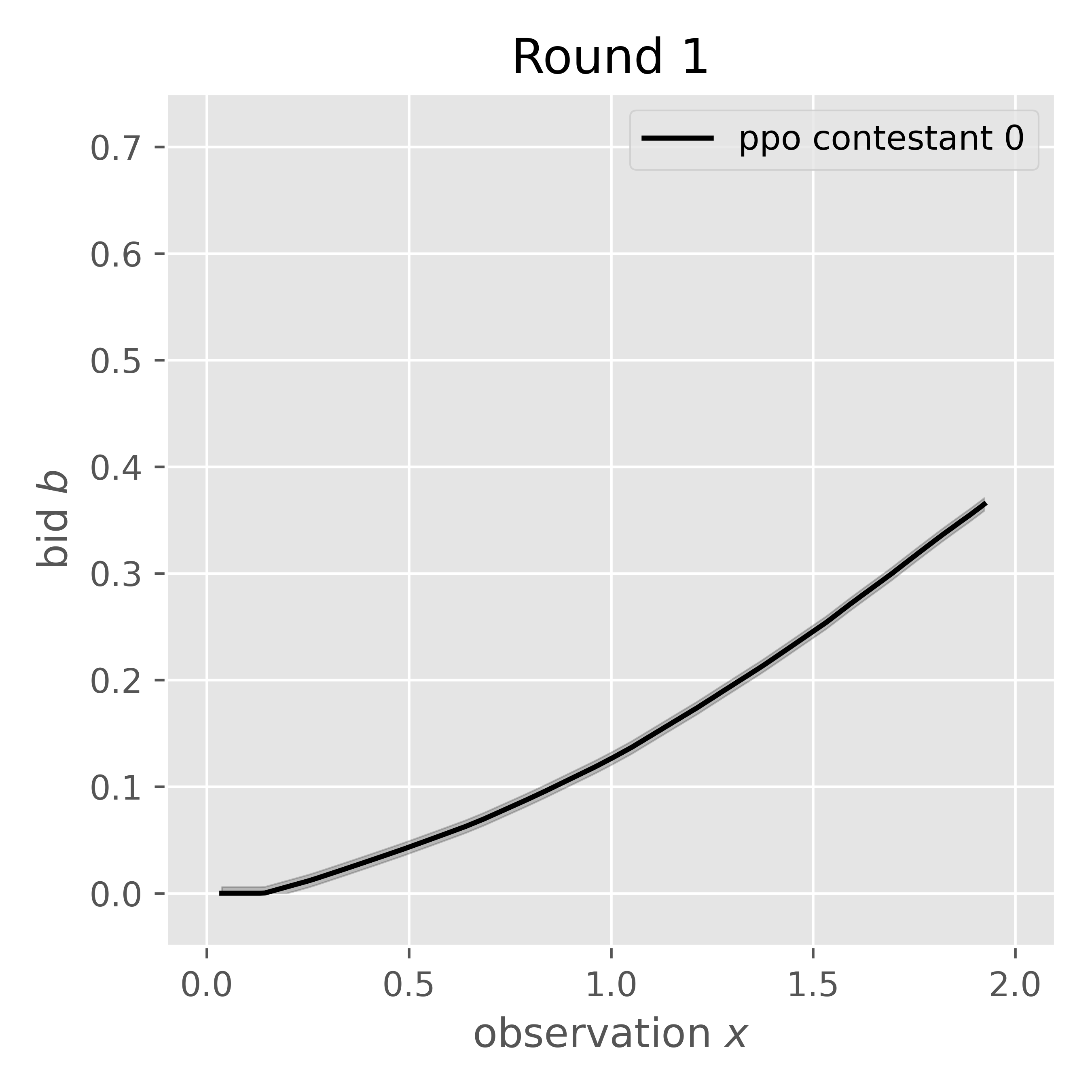}
	}
	\caption{PPO-based learned strategies in the first round of the elimination contest with interdependent priors, where the winner's bids are published after the first round.}
	\label{fig:contest-interdependencies-affiliated-winning-bids-strategies}
\end{figure*}

\subsubsection{Risk aversion}

Previous studies have shown that risk aversion can significantly impact strategies in contests. For example, \citet{skaperdasRiskAversionContests1995} demonstrate that the chance of winning is typically higher for risk-averse contestants, particularly under conditions of limited liability, where there is a non-trivial lower bound on utility. However, they also find that less risk-averse agents achieve higher utility overall, even though the probability of winning may be higher for risk-averse agents.
\citet{strackeOptimalPrizesDynamic2014} conduct an experimental study in a two-stage complete-information elimination contest. Their theoretical predictions suggest that risk-averse participants invest less effort in a single prize structure than under risk neutrality. This is due to their aversion to the uncertainty of winning the sole prize, leading to lower effort levels, which is also observed in their empirical experiments.
While these studies are relevant to our setting, they either do not involve multiple rounds or do not consider imperfect information. Again, we are not aware of a known equilibrium strategy in the elimination contest with risk-averse contestants.

\begin{table}
\centering
\caption{Learning results in the elimination contest with different levels of risk-averse bidders. We report the estimated utility loss $\ell^\text{ver}$ with its standard deviations over ten runs.}
\label{tab:table_signaling_contest_cara_risk}
\begin{tabular}{lllrr}
\toprule
public information & risk $\rho$ & metric &    \textsc{Reinforce}             &         PPO   \\
\midrule
valuations &   0.5 & $\ell^\text{est}$ &  0.0018 (0.0005)  &     0.0017 (0.0014)\\ \cdashline{2-5} \Tstrut
             & 1.0 & $\ell^\text{est}$ & -0.0007 (0.0005) &     -0.0012 (0.0004)\\ \cdashline{2-5} \Tstrut
             & 2.0 & $\ell^\text{est}$ & -0.0019 (0.0007) &     -0.0020 (0.0005)\\ \hline \Tstrut
bids &         0.5 & $\ell^\text{est}$ &  0.0008 (0.0009) &      0.0017 (0.0008)\\ \cdashline{2-5} \Tstrut
             & 1.0 & $\ell^\text{est}$ &  0.0007 (0.0004) &     -0.0011 (0.0003)\\ \cdashline{2-5} \Tstrut
             & 2.0 & $\ell^\text{est}$ & -0.0012 (0.0005)  &    -0.0018 (0.0003)\\
\bottomrule
\end{tabular}
\end{table}

We consider the two-stage elimination contest described in Section~\ref{sec:elimination-contest} with risk-averse contestants. Table~\ref{tab:table_signaling_contest_cara_risk} shows the estimated utility loss for different levels of risk aversion and types of public information. The estimated utility loss is low for all settings.
Figure~\ref{fig:contest-risk-winning-bids-strategies} shows the approximate equilibrium strategies from PPO agents for low ($\rho=0.5$) and high ($\rho=2.0$) risk aversion parameters. It is observed that the invested effort is lower for higher risk aversion, consistent with the observations of \citet{strackeOptimalPrizesDynamic2014}.

\begin{figure*}
	\centering
	\subfigure[risk $\rho = 0.5$]{
		\label{fig:contest-risk-winning-bids-round-1-strategies}
		\includegraphics[width=0.44\columnwidth]{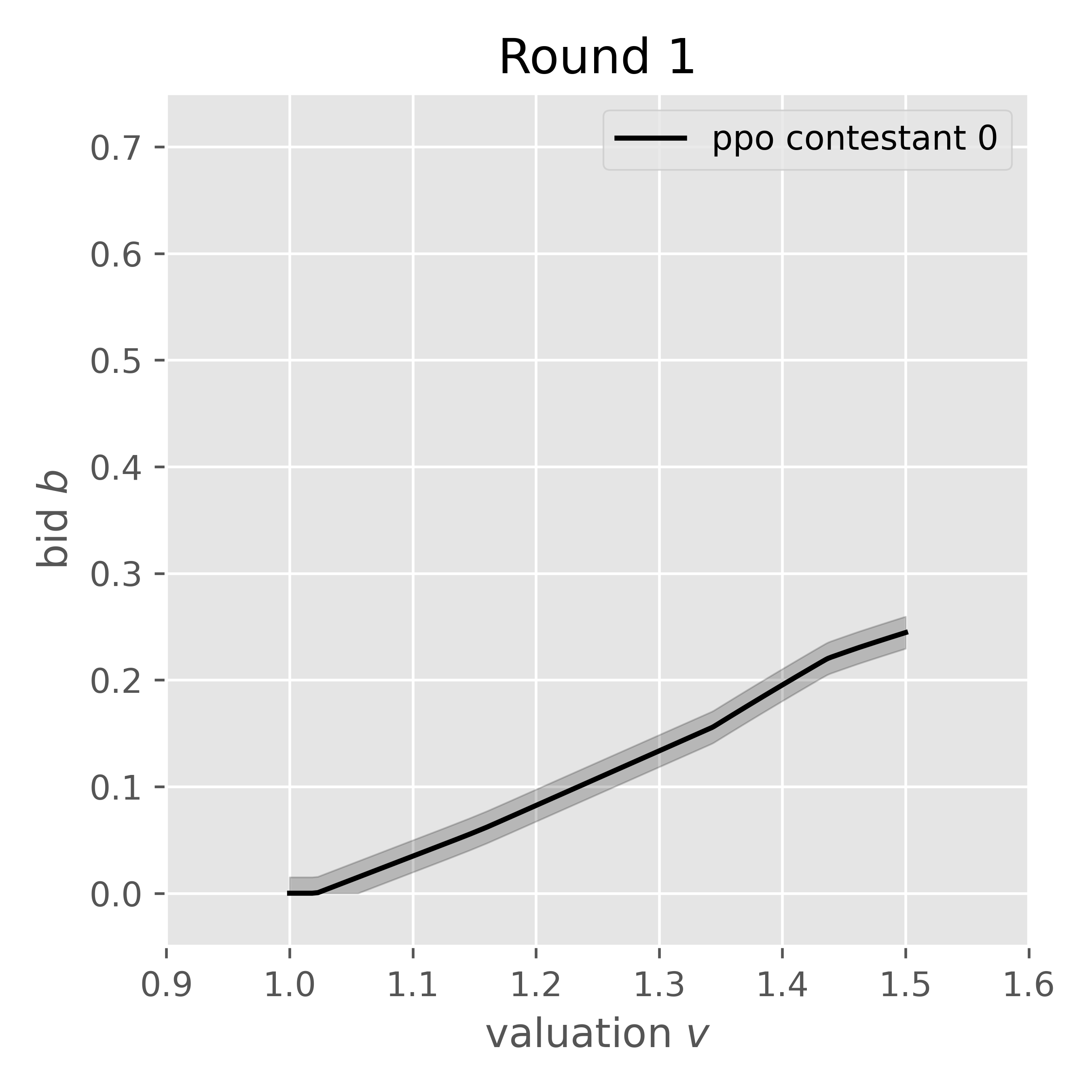}
	}
	\hfill
	\subfigure[risk $\rho = 2.0$]{
		\label{fig:contest-risk-winning-bids-round-2-strategies}
		\includegraphics[width=0.44\columnwidth]{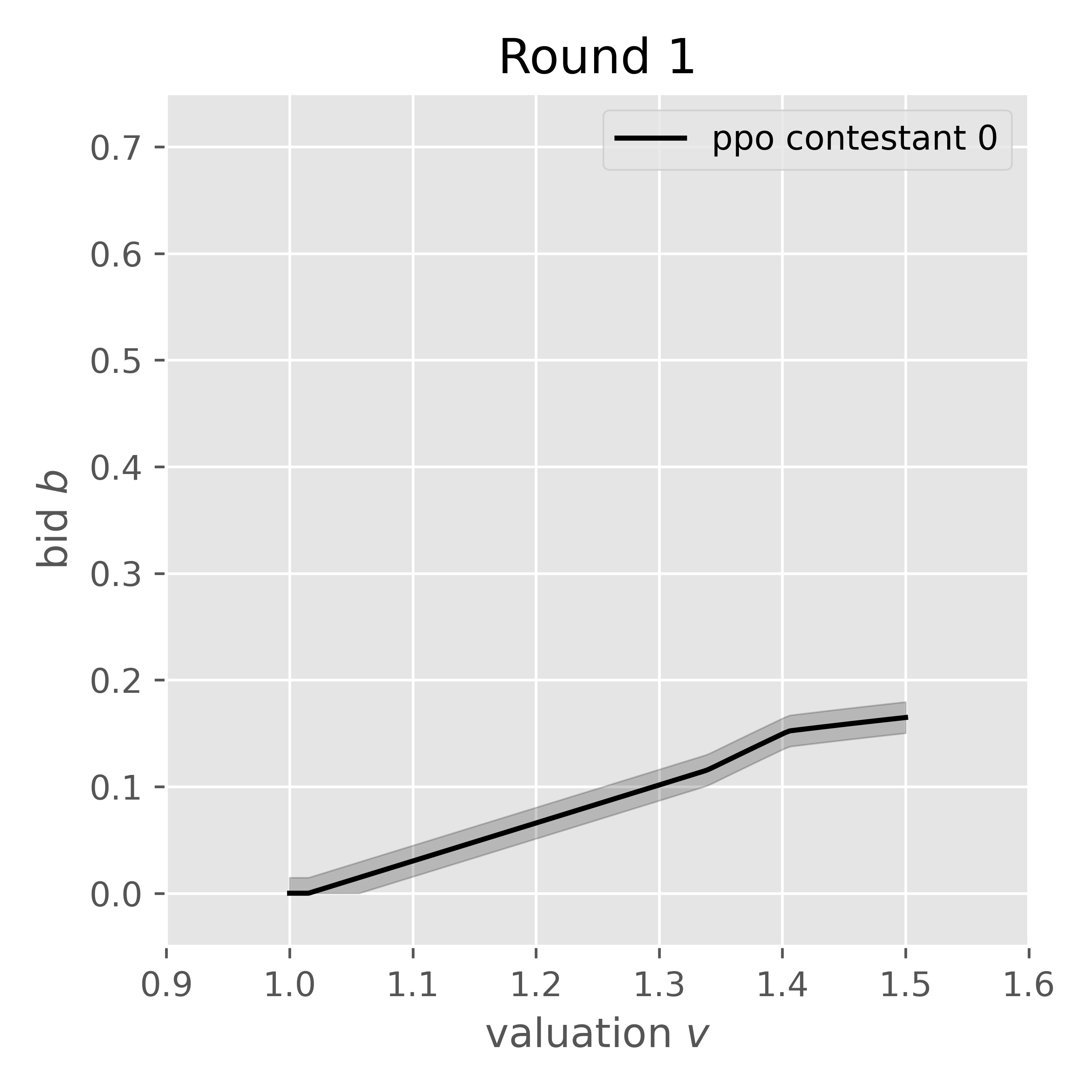}
	}
	\caption{PPO-based learned strategies in the first round of the elimination contest with risk-averse contestants, where the winner's bids are published after the first round.}
	\label{fig:contest-risk-winning-bids-strategies}
\end{figure*}

\subsection{Stackelberg Bertrand Competition}  \label{sec:bertrand-competition}
Stackelberg games, originally introduced by \citeauthor{vonstackelbergMarketStructureEquilibrium2011} in 1934, hold significant importance in economic theory and find various applications today \citep{liCournotOligopolyInformation1985, dowrickStackelbergCournotDuopoly1986, powellAllocatingDefensiveResources2007}. These game models involves a two-step process, where a leader makes the initial move, which is then observed by a follower who subsequently decides on its action. One crucial question that arises is the relative advantage of being the leader or the follower in such scenarios.

In our study, we focus on a specific variation known as the Stackelberg duopoly with incomplete information, as introduced by \citet{arozamenaSimultaneousVsSequential2009}. Their work explores equilibrium strategies in both simultaneous and Stackelberg-Bertrand competitions, comparing the behavior of the leader firm in these settings under different assumptions. Notably, they establish the existence of a second-mover advantage in the Stackelberg setting.
While their paper provides a general analysis and implicitly presents the equilibrium solution by providing its inverse, we delve into a special case within this framework for the standard setting. As a Stackelberg game is inherently asymmetric, we do not study an extension along this dimension. However, we consider interdependent prior distributions for the firms' costs and risk aversion. Again, we are not aware of known equilibrium strategies for these two adaptations.

\subsubsection{Standard setting} \label{sec:stackelberg-bertrand-standard-setting}
Consider a Stackelberg-Bertrand competition with two firms competing in a homogeneous-good market. Each firm sets its price, and the goal is to investigate the strategic interactions between the leader (Firm 1) and the follower (Firm 2).
Let $c_1$ and $c_2$ represent the unit costs of firms 1 and 2, respectively, which are drawn independently and identically distributed from the cumulative distribution function $F(c) = \frac{1}{2}c + \frac{c^2}{2}$. These costs are considered private information for each firm.

The game unfolds in the following manner: Firm 1 observes its private cost $c_1$ and subsequently sets its price $p_1$. Firm 2 observes its private cost $c_2$ and also the leader's posted price $p_1$. Based on this information, Firm 2 then sets its own price $p_2$.
Firm 1 wins the competition if $p_1 < p_2$, otherwise Firm 2 wins. The loser gets a utility of zero, whereas the winner $i$ receives a utility of $u_i(c_i, p_1, p_2) = (p_i - c_i) \cdot Q(p_i)$, where $Q(p) = 10 - p$ denotes the demand and $p_i = \min\{p_1, p_2\}$.
\citet{arozamenaSimultaneousVsSequential2009} derived the following class of equilibria.

\begin{proposition}[\citet{arozamenaSimultaneousVsSequential2009}]\label{prop:equilibrium-bertrand-competition}
	Consider a two-firm Stackelberg-Bertrand competition as described above. Then for every measurable function $f: \mathbb{R} \rightarrow \mathbb{R}$ such that $f(x) > x$, an equilibrium is given as follows:
	\begin{align*}
		\beta_{11}^{-1}(p_1) &= p_1 - \frac{Q(p_1)(1-F(p_1))}{Q(p_1)F^{\prime}(p_1) - Q^{\prime}(p_1)(1-F(p_1))} = \frac{4p_1^3 - 27 p_1^2 - 24p_1 + 20}{3p_1^2 - 18p_1 - 12}, \\ 
		\beta_{22}(c_2, p_1) &= \begin{cases}
			\min\{p_1, p^M(c_2)\}, &\text{ for } p_1 \geq c_2, \\
			f(b_1), &\text{ else,}
		\end{cases}
	\end{align*}
	where $p^M(c_2) = \max_{p_2}Q(p_2)(p_2 - c_2) = 5 + \frac{c_2}{2}$ denotes the monopoly price, and $\beta_{11}^{-1}$ is the leader's inverse equilibrium strategy.
\end{proposition}

The leader's equilibrium strategy $\beta_{11}$ is guaranteed to be invertible in the above setting, so that we recover it by numerically inverting $\beta_{11}^{-1}$ from above.
The results of the algorithms under consideration in this study are presented in Table~\ref{tab:table_bertrand_competition}. They show a small $L_2$-loss and (estimated) utility loss. If the estimated utility loss is negative, this means the learned strategy is better than the best finite precision step function strategy. Note that the utility losses' specific values $\ell^{\text{equ}}$ and $\ell^{\text{ver}}$ are not directly comparable to other experiments as the equilibrium utilities are about ten to fifteen times larger in this experiment.

\begin{table}
\centering
\caption{Learning results in the Bertrand competition. We again report the $L_2$ loss for each stage and agent, and the utility losses $\ell^\text{equ}$ and $\ell^\text{ver}$ with their standard deviations over ten runs.}
\label{tab:table_bertrand_competition}
\begin{tabular}{llrr}
\toprule
agent & metric &   \textsc{Reinforce}  &    PPO     \\
\midrule
leader & $\ell^\text{equ}$ &    0.0006 (0.0008) &  0.0001 (0.0006) \\ \cdashline{2-4} \Tstrut
             & $\ell^\text{ver}$ &    0.0004 (0.0002) &  0.0008 (0.0003) \\ \cdashline{2-4} \Tstrut
             &  $L_2^{S1}$ &     0.0064 (0.0036) &  0.0053 (0.0025) \\ \hline \Tstrut
follower & $\ell^\text{equ}$ &    0.0337 (0.0031) &  0.0435 (0.0078) \\ \cdashline{2-4} \Tstrut
             & $\ell^\text{ver}$ &   -0.0129 (0.0032) &  -0.0031 (0.0081) \\ \cdashline{2-4} \Tstrut
             &  $L_2^{S2}$ &     0.0046 (0.0005) &   0.0059 (0.0010) \\
\bottomrule
\end{tabular}
\end{table}

\subsubsection{Interdependent priors}
In this section, we study the Stackelberg Bertrand competition with interdependent priors, exploring how firms engage in strategic price-setting when their private information about costs or valuations is correlated. This scenario mirrors real-world markets where competitors' costs are interlinked, such as in industries with shared cost drivers.

A central insight of a Stackelberg interaction is the notion that a commitment to a course of action can confer a strategic advantage~\citep{vonstackelbergMarketStructureEquilibrium2011}, which was first presented by von Stackelberg in 1934 in his extension of the Cournot duopoly. In the analysis of \citet{arozamenaSimultaneousVsSequential2009}, which we discussed in Section~\ref{sec:bertrand-competition} there was a clear advantage for the follower, which achieved an approximate expected utility of $2.27$ compared to $0.79$ for the leader.
However, it is established that the information structure in a Stackelberg interaction significantly influences the strategic outcome~\citep{bagwellCommitmentObservabilityGames1995, maggiValueCommitmentImperfect1999}.
For example, \citet{gal-orFirstMoverDisadvantages1987} considers a Stackelberg Cournot model with incomplete information and allows for interdependent priors. He establishes that the follower gains a clear advantage over the leader the better she can estimate the action and private information of the leader. While there has been significant work on Stackelberg interactions with incomplete information, we are not aware of a known equilibrium strategy for our setting.

\begin{table}
\centering
\caption{Learning results in the Bertrand competition with interdependent prior distributions. There is no analytical equilibrium available for comparison. We report the estimated utility loss $\ell^\text{ver}$ with its standard deviations over ten runs.}
\label{tab:table_bertrand_competition_interdependencies}
\begin{tabular}{lllrr}
\toprule
prior & agent & metric &       \textsc{Reinforce}         &    PPO  \\
\midrule
mineral rights & leader & $\ell^\text{est}$ &   0.0017 (0.0029)  &     0.0072 (0.0064)\\  \cdashline{3-5} \Tstrut
			& follower & $\ell^\text{est}$ &   -0.0805 (0.0182)  &    -0.0863 (0.0175)\\ \hline \Tstrut
affiliated & leader & $\ell^\text{est}$ &       0.0056 (0.0046)   &    0.0189 (0.0399)\\ \cdashline{3-5} \Tstrut
			& follower & $\ell^\text{est}$ &   -0.0740 (0.0098)  &    -0.0395 (0.0144)\\ 
\bottomrule
\end{tabular}
\end{table}

We study the Stackelberg Bertrand competition with incomplete information under the mineral rights and affiliated prior for the firms' costs. The results are displayed in Table~\ref{tab:table_bertrand_competition_interdependencies}.
The interdependent priors lead to a scenario where the common cost remains unknown to both firms. However, the revealed price to the follower provides additional information about the leader's observation, allowing the follower to improve its estimate of the true cost. Nevertheless, the collected information is insufficient to perfectly determine the cost, so there remains a risk if the leader's price is low. However, the follower can at least estimate the cost as well as the leader, so a follower advantage remains.
Figure~\ref{fig:bertrand-interdependencies-strategies} shows the leader strategies of a PPO-based learner for the mineral rights and affiliated prior settings. In the case of the mineral rights prior, the follower bids close to $1.0$ for all observations. This contrasts with the leader's strategy under the affiliated prior, where the prices are substantially higher and further increase with the observation. For both prior distributions, the follower ignores its observation and simply posts a price slightly below the leader's price. Therefore, the results align with the insights gained by \citet{gal-orFirstMoverDisadvantages1987}. We observe a follower advantage in the experiments, where the follower makes a high revenue, whereas the leader's revenue is close to zero ($<10^{-3}$).

Under the affiliated prior, we observed that the approximated equilibrium strategies can vary depending on the seeds used. While the qualitative form remains consistent for all found approximate equilibria, as depicted in Figure~\ref{fig:bertrand-interdependencies-affiliated-strategies-leader}, the price ranges differ, typically between $2.2$ and $3.7$. This variation indicates the presence of several equilibrium strategies with minimal utility loss. Consequently, our method proves effective in learning different approximate equilibria as well.

\begin{figure*}
	\centering
	\subfigure[Mineral rights prior]{
		\label{fig:bertrand-interdependencies-mineral-rights-strategies-leader}
		\includegraphics[width=0.44\columnwidth]{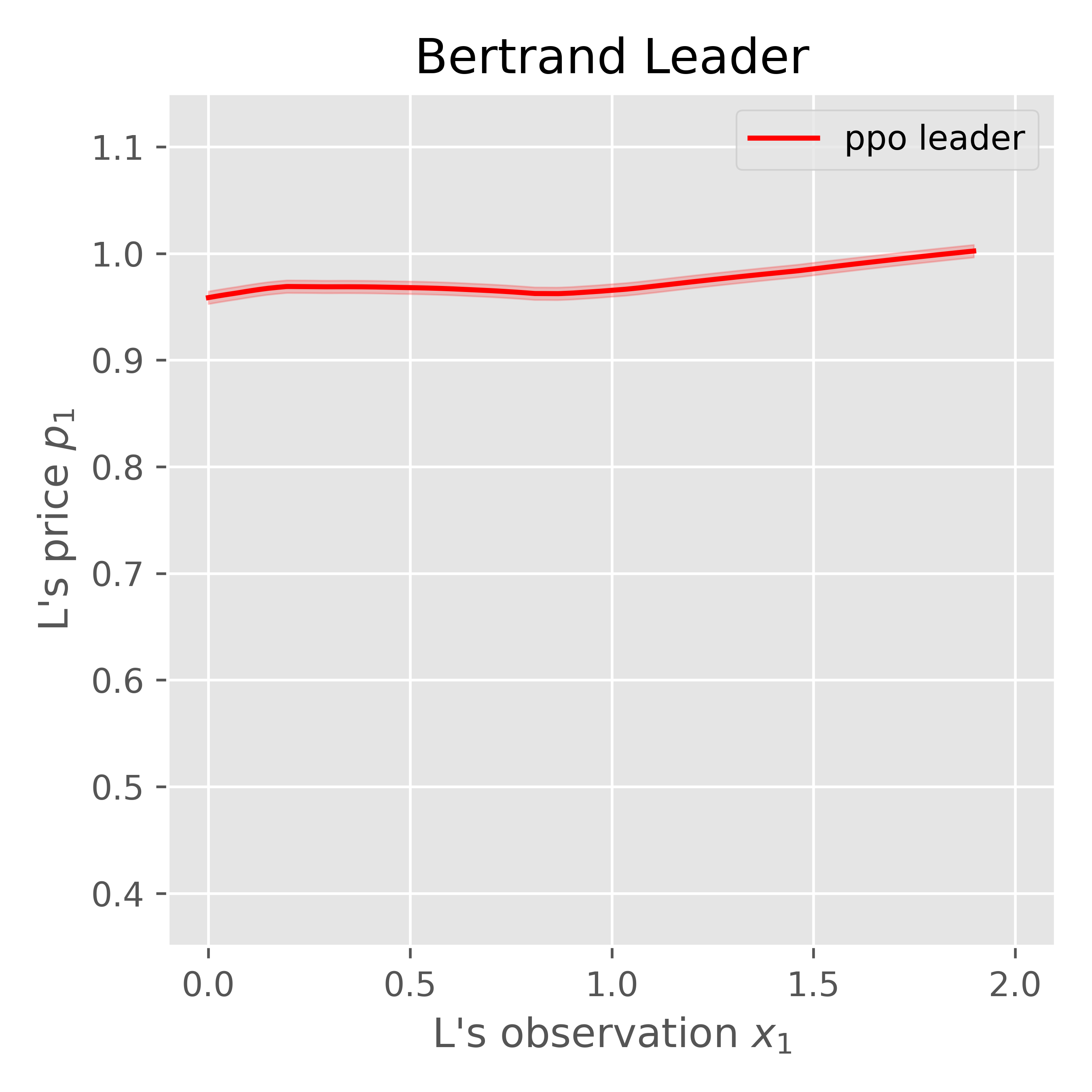}
	}
	\hfill
	\subfigure[Affiliated prior]{
		\label{fig:bertrand-interdependencies-affiliated-strategies-leader}
		\includegraphics[width=0.44\columnwidth]{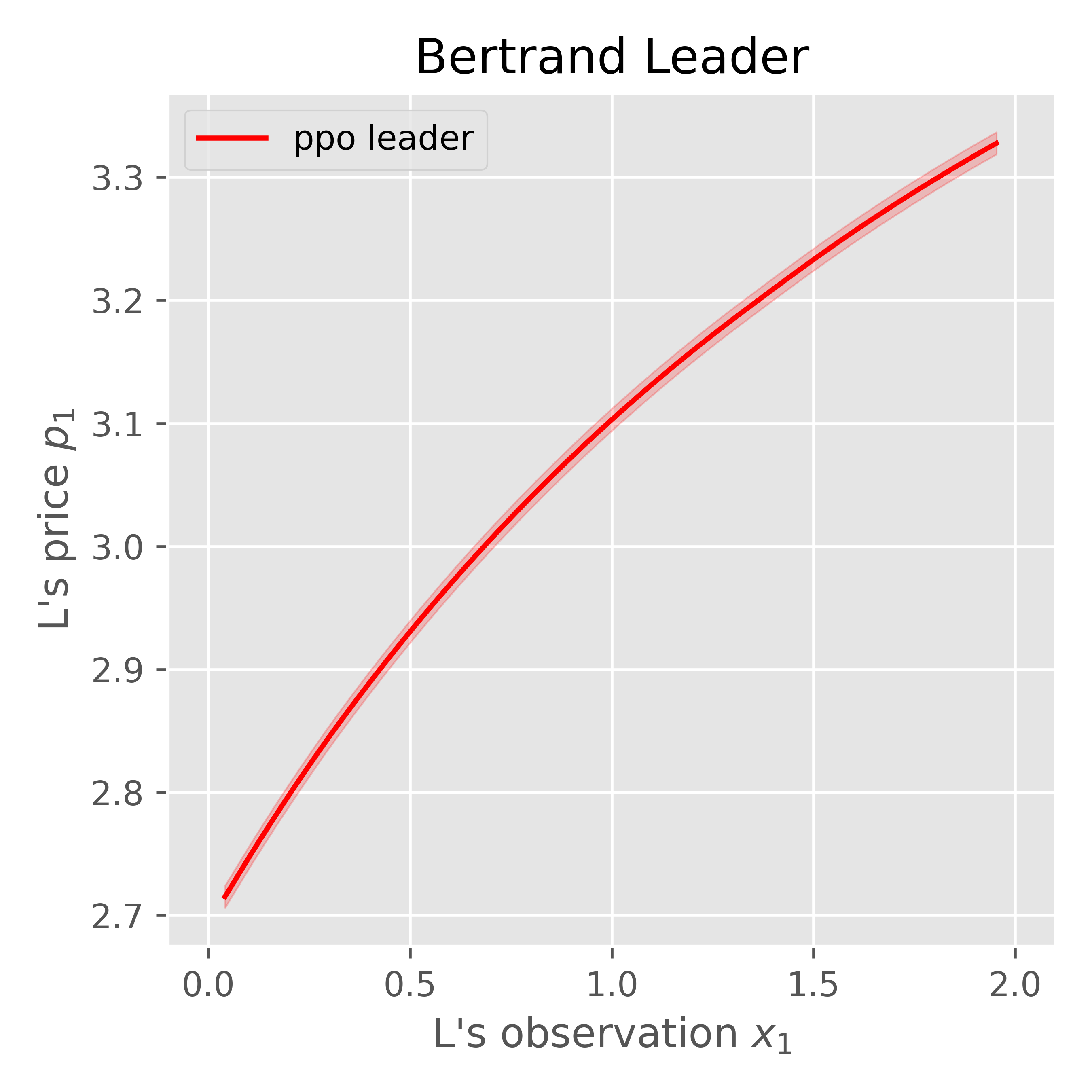}
	}
	\caption{PPO-based learned strategies of the leader in the Stackelberg Bertrand competition with interdependent prior distributions.}
	\label{fig:bertrand-interdependencies-strategies}
\end{figure*}

\subsubsection{Risk aversion}
Next, we study the Stackelberg Bertrand model with risk-averse firms. \citet{wambachBertrandCompetitionCost1999} examines how incomplete information and risk aversion affect price-setting behavior in a simultaneous Bertrand competition model. He demonstrates that, contrary to the Bertrand paradox, risk-averse firms can sustain prices above the competitive level even as the number of firms increases due to the potential losses associated with cost uncertainty.
\citet{maComplexDynamicsBertrandStackelberg2014} examine a supply chain consisting of two manufacturers and a common retailer. The manufacturers first interact in a Bertrand competition, acting as leaders, whereas the retailer reacts to their prices as a follower. They study gradient dynamics of price setting under uncertain demand and risk aversion, finding that the level of risk aversion greatly influences the system's stability.
None of the above settings incorporate a Stackelberg interaction and incomplete information under risk aversion. While risk-averse behavior has been extensively studied in the literature, we are not aware of a known equilibrium strategy for our setting.

\begin{table}
\centering
\caption{Approximated utility losses of the Bertrand competition with risk-averse bidders. There is no analytical equilibrium available for comparison. We report the estimated utility loss $\ell^\text{ver}$ with its standard deviations over ten runs. The reward was normalized for the runs indicated with $(*)$.}
\label{tab:table_bertrand_competition_risk}
\begin{tabular}{lllrr}
\toprule
risk $\rho$ & agent & metric &      \textsc{Reinforce}          &           PPO   \\
\midrule
0.5 & leader & $\ell^\text{est}$ &    0.0001 (0.0001)  &      0.0003 (0.0002)\\ \cdashline{3-5} \Tstrut
    & follower & $\ell^\text{est}$ & -0.0058 (0.0007)   &    -0.0048 (0.0021)\\ \hline \Tstrut
1.0 & leader & $\ell^\text{est}$ &    0.0001 (0.0001)   &     0.0001 (0.0001)\\ \cdashline{3-5} \Tstrut
    & follower & $\ell^\text{est}$ &  0.0007 (0.0013) &      -0.0013 (0.0014)\\ \hline \Tstrut
2.0 & leader & $\ell^\text{est}$ &    0.0000 (0.0000)   &     0.0001 (0.0001)$(*)$ \\ \cdashline{3-5} \Tstrut
    & follower & $\ell^\text{est}$ &  0.0052 (0.0011)  &     -0.0014 (0.0008)$(*)$ \\
\bottomrule
\end{tabular}
\end{table}

The results are summarized in Table~\ref{tab:table_bertrand_competition_risk} and show a very small estimated utility loss for all settings.
Figure~\ref{fig:bertrand-risk-strategies} shows the leader's strategy for two different levels of risk aversion, where the left plot displays the strategy for risk parameter $\rho=0.5$ and the right for $\rho=2.0$. Contrary to the prediction made by \citet{wambachBertrandCompetitionCost1999}, higher levels of risk aversion lead to lower prices and, therefore, higher competition. The follower's strategy aligns with expectations by slightly underbidding the leader's price whenever her cost is sufficiently low.

For a risk parameter of $\rho = 2.0$, the exponential nature of CARA risk aversion results in utilities on the order of $-10^{14}$ for the initialized bidding strategies. This scale of rewards is known to cause issues for learning algorithms. Therefore, we performed an additional normalization step of the rewards for the PPO algorithm. Our \textsc{Reinforce} implementation normalizes rewards by default.

\begin{figure*}
	\centering
	\subfigure[risk $\rho=0.5$]{
		\label{fig:bertrand-risk-low-strategies-leader}
		\includegraphics[width=0.44\columnwidth]{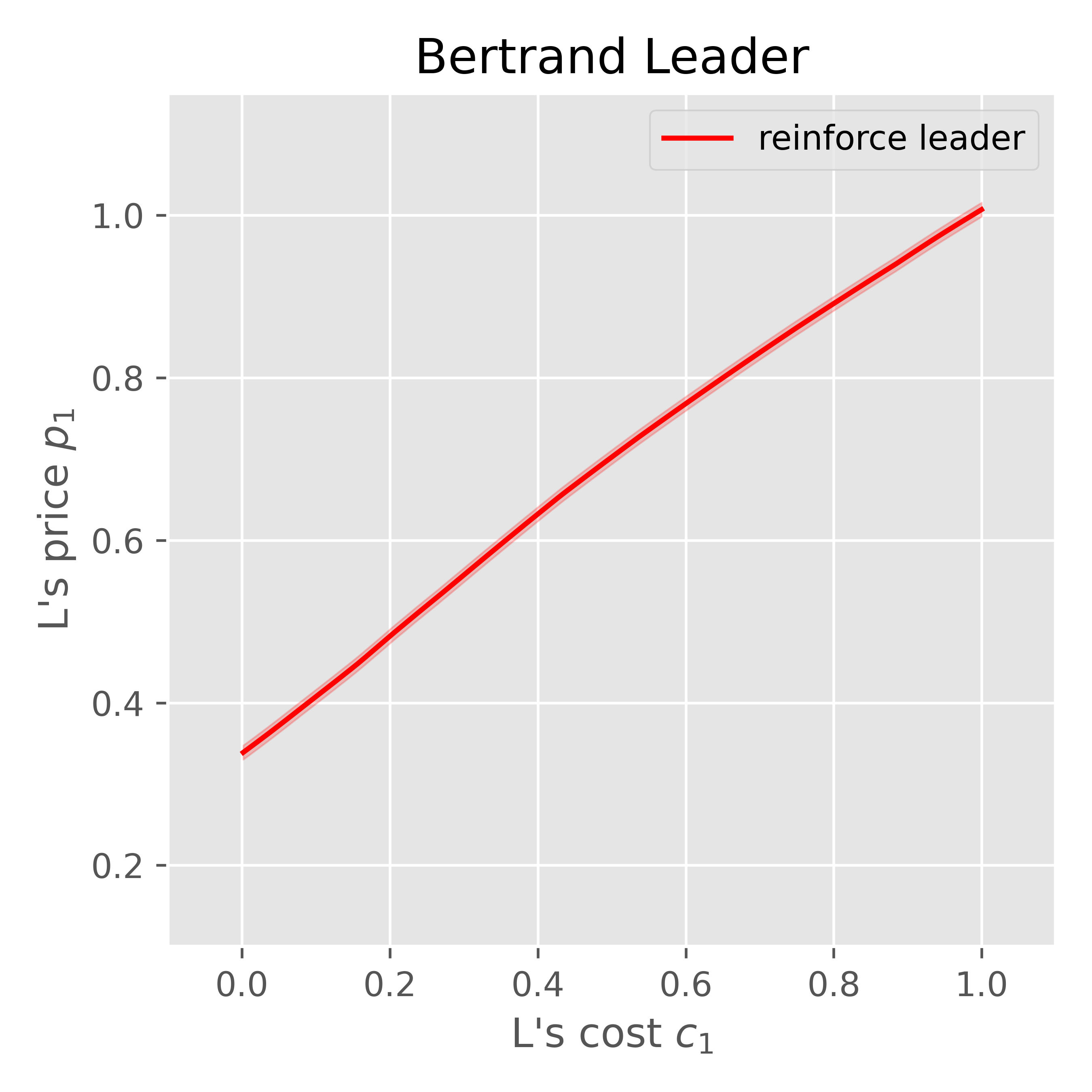}
	}
	\hfill
	\subfigure[risk $\rho=2.0$]{
		\label{fig:bertrand-risk-high-strategies-leader}
		\includegraphics[width=0.44\columnwidth]{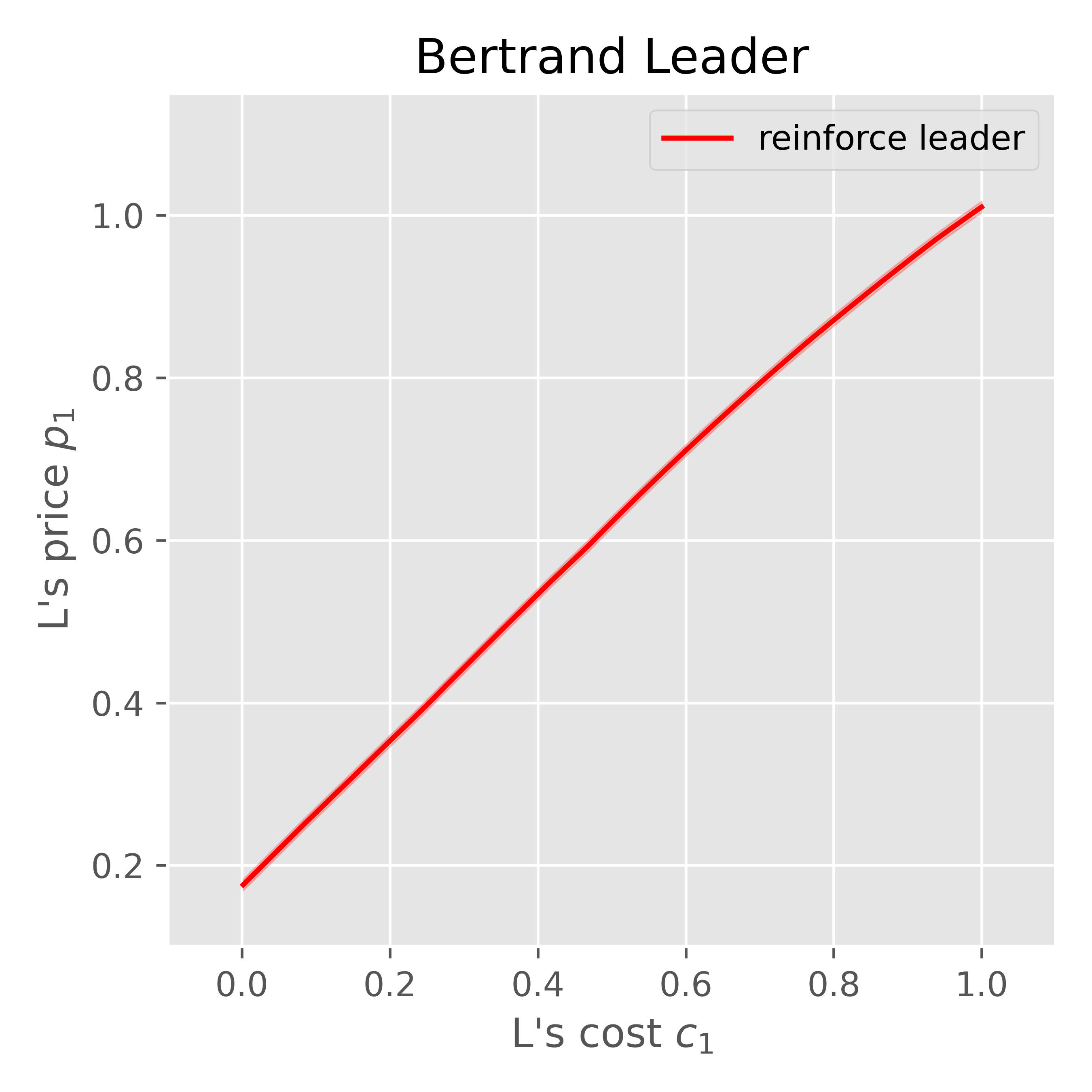}
	}
	\caption{\textsc{Reinforce}-based learned strategies of the leader in the Stackelberg Bertrand competition with risk-averse firms.}
	\label{fig:bertrand-risk-strategies}
\end{figure*}

\section{Conclusions}

We investigate the application of DRL for continuous multi-stage games. Although such games are central to theory and in numerous real-world applications, we know equilibrium strategies for only a few restricted models.
Infinite type and action spaces make the analysis of such games challenging and different from combinatorial games such as Chess and Go, as they have been the focus of multi-agent reinforcement learning so far. Besides, scholars are interested in equilibrium strategies of respective games. Computing such equilibrium strategies has long been considered intractable. 

We employed deep reinforcement learning to approximate the equilibrium bidding strategies in these continuous multi-stage games. Interestingly, our experiments showed that these methods find equilibrium strategies in sequential auctions and contests, but also in Stackelberg pricing competition under very different model assumptions. Importantly, the analysis can be performed quickly. We can explore new environments but also unravel new equilibria that were not known before. While symmetric game-theoretical models sometimes have a unique symmetric equilibrium, there might be multiple asymmetric equilibria. Equilibrium learning can be seen as a form of equilibrium selection that identifies one equilibrium to which learning and repeated play converges. However, there may be several equilibria that are attracting under the learning dynamics. Therefore, equilibrium learning can also be used to find a variety of different equilibria. 

A precondition is a verifier, which is able to certify an approximate equilibrium in the broad class of continuous multi-stage games, even when no analytical equilibrium is known and proves that we receive an upper bound on the utility loss as the number of samples (games played) grows large and the level of discretization increases. 
Overall, the flexibility that we get in modeling strategic situations and the speed with which different models are solved allows us to push the boundaries of equilibrium analysis and equilibrum learners are a valuable tool for academics and practitioners alike.

\bibliographystyle{ormsv080}
\bibliography{References, dairefs} 

\newpage
\begin{APPENDICES}

	\section{Verification Procedure}\label{sec:verification-procedure-details}
Let us outline our methodology for verifying approximate equilibria in continuous multi-stage games. We start with some technical preliminaries and proceed with presenting the verification procedure formally.

\subsection{Preliminaries}\label{sec:multi-stage}
In Definition~\ref{def:msg}, we defined a multi-stage game with continuous signals and actions. In what follows, we summarize some additional assumptions required for bounding the error of our verifier. 

Throughout, we assume the game to have \emph{perfect recall}.
This means that players remember all information they received, and in particular, their own actions.
This assumption greatly simplifies theory and is usually made in literature \citep{myersonPerfectConditionalEEquilibria2020}.

\begin{definition}[Perfect recall] \label{def:perfect-recall}
	Let $\Gamma = \left(\mathcal{N}, T, S, \mathcal{A}, p, \sigma, u \right)$ be a multi-stage game. It is said to have \emph{perfect recall} if for all $it \in L$ and $r>t$, there are measurable functions $\Psi:S_{ir} \rightarrow S_{it}$ and $\psi:S_{ir} \rightarrow A_{it}$ such that $\Psi(\sigma_{ir}(a_{<r})) = \sigma_{it}(a_{<t})$ and $\psi(\sigma_{ir}(a_{<r})) = a_{it}$, for all $a \in \mathcal{A}$.
\end{definition}

Under perfect recall, one can extract all of $i$'s actions and received signals up to that point from a signal $s_{ir} \in S_{ir}$. That is, there exist functions $\psi_{irt}: S_{ir} \rightarrow \mathcal{A}_{i t}$ and $\Psi_{irt}: S_{ir} \rightarrow S_{it}$ such that $\psi_{irt}(\sigma_{ir}(a_{< r})) = a_{it}$ and $\Psi_{irt}(\sigma_{ir}(a_{< r}))=\sigma_{it}(a_{<t})$ for all $a_{<r}=(a_{\cdot 1}, \cdots, a_{\cdot r-1}) \in \mathcal{A}_{<r}$. Denote with $\psi_{ir} = (\psi_{ir1}, \psi_{ir2}, \dots, \psi_{i r r-1})$ and $\Psi_{ir} = (\Psi_{ir1}, \Psi_{ir2}, \dots, \Psi_{i r r-1})$ the corresponding mappings from $S_{ir}$ into $\mathcal{A}_{i <r}$ and from $S_{ir}$ into $S_{i<r}$ respectively. 


Later on, we are interested in two restrictions of the strategy spaces and their intersection. The first restriction allows only pure strategies, i.e., strategies that map onto Dirac measures for a given signal, which we denote by
\begin{align} \label{equ:set-of-pure-strategies-over-dirac-measures}
	\Sigma_{it}^{\text{p}} := \left\{\beta \in \Sigma_{it}\;|\;\beta(s_{it}) = \delta_a \text{ for } s_{it} \in S_{it} \text{ and } a \in \mathcal{A}_{it} \right\}.
\end{align}
We can identify this with the set of measurable functions from signals to actions. With slight abuse of notation, we denote both sets by $\Sigma_{it}^{\text{p}}$ and note where it does not become clear from the context. We set $\Sigma_{i}^{\text{p}} = \bigtimes_{t \in T}\Sigma_{it}^{\text{p}}$.

For verification, we are interested in whether one can achieve the same best-response utility by restricting the space of a single-agent to pure strategies. This is often satisfied under mild assumptions. For example, it is satisfied in most Bayesian games \citep{Milgrom1985, hosoyaApproximatePurificationMixed2022}. Furthermore, it can often be guaranteed to be viable whenever a pure strategy equilibrium exists \citep{horstStationaryEquilibriaDiscounted2005, reny2011EXISTENCEMONOTONEPURESTRATEGY}.


The second restriction is to consider the set of Lipschitz continuous functions, which we denote by
\begin{align}\label{equ:lipschitz-continuous-strategies}
	\Sigma_{it}^{\text{Lip}} = \left\{\beta_{it} \in \Sigma_{it}\;|\;\beta_{it}: S_{it} \rightarrow \Delta \left(\mathcal{A}_{it} \right) \text{ is Lipschitz continuous in } d_W \right\},
\end{align}
where $d_W$ denotes the Wasserstein distance (Definition~\ref{def:wasserstein-distance}). We set $\Sigma_{i}^{\text{Lip}} = \bigtimes_{t\in T} \Sigma_{it}^{\text{Lip}}$. Note that common function approximators for distributional strategies, such as neural networks, fall into the space of Lipschitz continuous strategies.

Finally, we consider the intersection of pure and Lipschitz continuous strategies, which we denote by $\Sigma_{it}^{\text{Lip, p}} := \Sigma_{it}^{\text{Lip}} \cap \Sigma_{it}^{\text{p}}$, and $\Sigma_{i}^{\text{Lip, p}} := \bigtimes_{t \in T} \Sigma_{it}^{\text{Lip, p}}$.

For a given $\beta_{\cdot t}$, one can define a probability distribution from stage $t$ to $t+1$. Let $B \subset \mathcal{A}_{\cdot t}$ be measurable and $a_{<t} \in \mathcal{A}_{<t}$, then the mapping $P_t$ defines a transition probability from $\mathcal{A}_{<t}$ into the set of measurable subsets of $\mathcal{A}_{\cdot t}$ by
\begin{align} \label{equ:transition-prob-t-to-tplus1}
	P_t(B\,|\,a_{<t}, \beta_{\cdot t}) = p_t(B_{0t}\,|\,a_{<t}) \prod_{i \in \mathcal{N}} \beta_{it}(B_{it}\,|\,\sigma_{it}(a_{<t})),
\end{align}
where $\beta_{it}(B_{it}\,|\,\sigma_{it}(a_{<t}))$ is the probability that player $i$ takes actions from $B_{it}$ when receiving the signal $\sigma_{it}(a_{<t})$.
The players need to reason about several stages. Therefore, we inductively define probabilities that describe events from the beginning up to a certain stage.

Let $P_{<1}(\{\emptyset \}| \beta) = 1$, and for all $t \in T$ and measurable $B \subset \mathcal{A}_{<t+1}$, define the rollout measure up to stage $t$ under strategy profile $\beta$ as
\begin{align} \label{equ:transition-prob-beginning-to-t}
	P_{<t+1}(B|\beta) = \int_{\mathcal{A}_{<t}} P_t \left(\{a_{\cdot t}: \left(a_{<t}, a_{\cdot t} \right) \in B \}\,|\,a_{<t}, \beta_{\cdot t} \right) d P_{<t}\left(a_{<t}\,|\,\beta \right).
\end{align}
Intuitively speaking, this defines the probability of an intermediate history $B$ up to stage $t$ to occur when all players act according to $\beta$.
The probability measure $P(\argdot\,|\,\beta) := P_{<T+1}(\argdot\,|\,\beta)$ denotes the probability measure over outcomes induced by the strategy profile $\beta$. 

Finally, player $i$'s ex-ante utility is defined as the expected utility over all possible outcomes of the game and is given by
\begin{align} \label{equ:ex-ante-utility}
	\tilde{u}_i\left(\beta\right) = \int_{\mathcal{A}} u_i(a) d P(a\,|\,\beta).
\end{align}

In the game's interim stages, the individual agent reasons about what action to take after receiving a signal. To describe this optimization problem, we introduce the conditional probabilities and utilities to describe the expected utility given a certain signal and strategies. Let $\beta \in \Sigma$ be a strategy profile, $it \in L$, and $Z \subset S_{it}$ measurable, then define
\begin{align} \label{equ:date-t-probabilities-given-signal}
	P_{it}(Z\,|\,\beta) = P_{< t}(\sigma_{it}^{-1}(Z)\,|\,b) = P_{< t}(\{a_{< t}: \sigma_{it}(a_{< t}) \in Z \}\,|\,\beta),
\end{align}
to be the probability that player $i$'s stage $t$ signal is in $Z$ under strategy profile $\beta$. Consequently, for any $it \in L$ and measurable $Z \subset S_{it}$ such that $P_{it}(Z|\beta) > 0$, we define, with a slight abuse of notation, the conditional probabilities as
\begin{align} \label{equ:up-to-t-conditional-probabilities-given-signal}
	P_{< t}(B\,|\,Z, \beta) = P_{< t} (B \cap \sigma_{it}^{-1}(Z)\,|\,\beta) / P_{it}(Z |\beta) \quad \forall B \subset \mathcal{A}_{< t} \text{ measurable},
\end{align}
and
\begin{align} \label{equ:conditional-probabilities-given-signal}
	P(B\,|\,Z, \beta) = P (\{a \in B : \sigma_{it}(a_{< t} ) \in Z \}\,|\,\beta) / P_{it}(Z |\beta) \quad \forall B \subset \mathcal{A} \text{ measurable}.
\end{align}
With all of this, the conditional expected utilities for a measurable $Z \subset S_{it}$ are defined as
\begin{align} \label{equ:conditional-utilities-given-signal}
	\tilde{u}_i(\beta\,|\,Z) = \int_{\mathcal{A}} u_i(a) d P (a\,|\,Z, \beta).
\end{align}

\subsection{Construction} \label{sec:verifier-formal-description}
The verification procedure of an agent's learned strategy consists of two main parts. 
First, the strategy space must be discretized such that the search space is reduced to finite size, and the number of simulations must be set such that the expected utilities (across nature and the opponents' action probabilities) can be approximated via sampling.
Second, the deployment of all possible strategies is simulated, and the best-performing strategy with the highest utility is compared to the utility of the actual strategy. 

\subsubsection{Finite precision step-functions} 

Furthermore, for every $it \in L$ and $\text{D} \in \mathbb{N}$ there exists a finite number of disjoint grid cells that consist of a product of half-open intervals, partitioning $S_{it}$. We denote these grid cells by $C_{it}^{k}$, where $1 \leq k \leq G_{S_{it}}^\text{D}$ and $G_{S_{it}}^\text{D} \in \mathbb{N}$ is the number of grid cells. Finally, let the tuple $\mathcal{G}_{it} = \left(S_{it}^{\text{D}}, \mathcal{A}_{it}^{\text{D}}, C_{it} , \text{D} \right)$ denote the signal and action space discretization for $it \in L$.  

We define the set of step functions with precision $\text{D}$ for $it \in L$ by
\begin{align} \label{equ:set-of-finite-precision-step-functions}
	\Sigma_{it}^{\text{D}}\left(\mathcal{G}_{it}\right) := \left\{s_{it} \mapsto \sum_{k=1}^{G_{S_{it}}^\text{D}} \chi_{C_{it}^{k}}(s_{it}) a_{it}^{k}\,\Bigg\rvert\,a_{it}^k \in \mathcal{A}_{it}^\text{D} \right\},
\end{align}
and $\Sigma_i^{\text{D}} \left(\mathcal{G}_{i} \right) = \bigtimes_{t \in T} \Sigma_{it}^{\text{D}} \left(\mathcal{G}_{it}\right)$. Any grid so that $\Sigma_{it}^{\text{D}}$ can approximate any pure Lipschitz continuous function well for sufficiently high $\text{D}$ (\autoref{thm:finite-step-functions-arbitrarily-close-to-continuous-functions}) can be used. In the following, we restrict ourselves to the regular grid and show that it satisfies this property. Therefore, we drop the discretization and write $\Sigma_i^{\text{D}}$ instead of $\Sigma_i^{\text{D}} \left(\mathcal{G}_{it} \right)$.

For any given finite precision $\text{D} \in \mathbb{N}$, we make a \emph{discretization error}, which we denote by 
\begin{align} \label{equ:discretization-error}
	\varepsilon_{\text{D}} := \sup_{\beta^{\prime}_i \in \Sigma_i^{\text{Lip, p}}} \tilde{u}_i(\beta^{\prime}_i, \beta_{-i}) - \sup_{\beta^{\prime}_i \in \Sigma_{i}^{\text{D}}} \tilde{u}_i(\beta^{\prime}_i, \beta_{-i}).
\end{align}

\subsubsection{Backward induction over finite precision step functions}

For every finite precision $\text{D}$, there are finitely many elements in $\Sigma_{i}^{\text{D}}$, which can be translated into finitely many decision points for player $i$. That is, we can build a finite game or decision tree representing all possible step functions from $\Sigma_{i}^{\text{D}}$. We perform a backward induction scheme on this finite decision tree to get the maximal ex-ante utility from any step function.

To achieve this, we define player $i$'s \emph{counterfactual conditional utility} as the conditional utility for taking a specific action given a certain signal, excluding player $i$'s influence of reaching this signal. This is similar to formulations of counterfactual reach probabilities and utilities from literature in finite games \citep{bibid}.
Before the counterfactual conditional utilities can be formally defined, we need to introduce some other objects first.

We exclude player $i$'s influence of reaching a certain signal grid cell $C_{it}^k \subset S_{it}$ by considering a strategy that deterministically plays to reach $C_{it}^k$. That is possible because, without loss of generality, one can assume that there exists a unique sequence of grid actions $a_{i<t}^{C_{it}^k} \in \mathcal{A}_{i <t}^{\text{D}}$ that need to be taken for every grid cell $C_{it}^k$. To see this, note that due to perfect recall, agent $i$'s actions taken prior to stage $t$ can be extracted from every signal $s_{it}$. Due to our construction, this is a unique sequence for every grid cell $C_{it}^k$ if, for example, each action space $\mathcal{A}_{ir}$ gets appended to the signaling space in the following stage $S_{ir+1}$. Therefore, for every $C_{it}^k$, there exist $s_{it} \in C_{it}^k$ and $a_{i<t} \in \mathcal{A}_{i<t}^{\text{D}}$ such that $\psi_{it}(s_{it}) = a_{i<t}$. At the same time, there exists no $s_{it}^{\prime} \in C_{it}^k$ such that $\psi_{it}(s_{it}^{\prime}) \neq a_{i<t}$ and $\psi_{it}(s_{it}^{\prime}) \in \mathcal{A}_{i<t}^{\text{D}}$. Therefore, we define functions $\psi_{it}^{\text{D}}(C_{it}^k) = a_{i<t}^{C_{it}^k}$ that map a grid cell to its unique history of grid actions. 

Given a finite precision step function strategy $\beta_{i} \in \Sigma_{i}^{\text{D}}$, we can now construct a strategy for player $i$ that plays to reach $C_{it}^k$ (adapting stages $1, \dots, t-1$), takes a certain action in stage $t$, and remains the same for stages $t+1, \dots, T$. More specifically, let $it \in L$, $\beta = (\beta_i, \beta_{-i})$ be a strategy profile with $\beta_i \in \Sigma_{i}^{\text{D}}$ and $\beta_{-i} \in \Sigma_{-i}$, and $a_{it}^k \in \mathcal{A}_{it}^{\text{D}}$. Then we define a function $\left(\beta_{i}\right)^{C_{it}^k, a_{it}^k} = \left(\left(\beta_{i1}\right)^{C_{it}^k, a_{it}^k}, \dots, \left(\beta_{iT}\right)^{C_{it}^k, a_{it}^k} \right)$ that is playing to reach $C_{it}^k$ in the following way:

\begin{align*} \label{equ:counterfactual-discrete-reach-strategy}
	\left(\beta_{ir}\right)^{C_{it}^k, a_{it}^k} &= \beta_{ir} \quad \text{for } r> t, \\
	\left(\beta_{it}\right)^{C_{it}^k, a_{it}^k}(s_{it}) &=
	\begin{cases*}
		a_{it}^k \quad \text{for } s_{it} \in C_{it}^k,\\
		\beta_{it}(s_{it}) \quad \text{ for } s_{it} \in S_{it} \setminus C_{it}^k,
	\end{cases*}\\
	\left(\beta_{i<t}\right)^{C_{it}^k, a_{it}^k}\left(\Psi_{it}(s_{it})\right) &=
	\begin{cases*}
		\psi_{it}^{\text{D}}(C_{it}^k) \quad \text{for } s_{it} \in C_{it}^k,\\
		\beta_{i<t}\left(\Psi_{it}(s_{it})\right) \quad \text{for } s_{it} \in S_{it} \setminus C_{it}^k,
	\end{cases*}
\end{align*}

The counterfactual conditional utilities for precision $\text{D}$ are then defined as
\begin{align}
	\tilde{u}_i^{\text{c, D}}\left(\beta\,|\,C_{it}^k, a_{it}^k \right) = \tilde{u}_i\left(\left(\beta_{i}\right)^{C_{it}^k, a_{it}^k}, \beta_{-i}\,|\,C_{it}^k \right),
\end{align}
where $\tilde{u}_i(\argdot\,|\,\argdot)$ is the conditional utility defined in \autoref{equ:conditional-utilities-given-signal}. Note that $\tilde{u}_i^{\text{c, D}}$ is independent of $\beta_{i<t} \in \Sigma_{i<t}^{\text{D}}$, as only histories conditioned on observing signals from $C_{it}^k$ are considered. All actions that may be taken off paths that lead to this set of signals do not matter. Therefore, we write $\tilde{u}_i^{\text{c, D}}\left(\beta_{i > t}, \beta_{-i}\,|\,C_{it}^k, a_{it}^k \right)$ instead of $\tilde{u}_i^{\text{c, D}}\left(\beta\,|\,C_{it}^k, a_{it}^k \right)$, where $\beta_{i > t} = (\beta_{it+1}, \dots, \beta_{iT})$, to emphasize this independence.

We are now ready to define a best response over the finite strategy set $\Sigma_{i}^{\text{D}}$ via backward induction.
For a given opponent strategy profile $\beta_{-i}$, we inductively define a step function $\beta_{i}^{\text{D}, *} \in \Sigma_i^{\text{D}}$. For the last stage $s_{iT} \in S_{iT}$, define
\begin{align} \label{equ:backward-induction-stage-T-analytically}
	\beta_{iT}^{\text{D}, *}(s_{iT}) = \argmax_{a_{iT} \in \mathcal{A}_{iT}^\text{D}} \tilde{u}_i^{\text{c, D}}\left(\beta_{-i}\,|\,C_{iT}^k, a_{iT} \right),
\end{align}
where $C_{iT}^k$ is the unique set such that $s_{iT} \in C_{iT}^k$. For preceding stages $t<T$, we define
\begin{align} \label{equ:backward-induction-intermediate-stages-t-analytically}
	\beta_{it}^{\text{D}, *}(s_{it}) = \argmax_{a_{it} \in \mathcal{A}_{it}^\text{D}} \tilde{u}_i^{\text{c, D}}\left(\beta_{i> t}^{\text{D}, *}, \beta_{-i}\,|\,C_{it}^k, a_{it} \right),
\end{align}
again with $s_{it} \in C_{it}^k$.
Note that the $\argmax_{a_{it} \in \mathcal{A}_{it}^\text{D}}$ is non-empty, as the utility functions are bounded, and there are only finitely many values to consider. If there is more than one element in the $\argmax_{a_{it} \in \mathcal{A}_{it}^\text{D}}$, then we simply choose one discrete action for a whole grid cell $C_{it}^k$. This backward induction procedure gives us a best response over the set $\Sigma_{i}^{\text{D}}$ (\autoref{thm:backward-induction-gives-sup-over-step-functions}).

\subsubsection{Monte-Carlo integration for conditional utilities} \label{sec:monte-carlo-integratoin-for-conditional-utilities}
The backward induction procedure above assumes that the conditional expected utilities from Equations \ref{equ:backward-induction-stage-T-analytically} and \ref{equ:backward-induction-intermediate-stages-t-analytically} can be evaluated, which is, in general, impossible. It would require having access to the expectations of the conditional utilities (\autoref{equ:conditional-utilities-given-signal}) for which there may be no closed-form solutions.
Therefore, we employ Monte-Carlo approximation to estimate $\beta_{i}^{\text{D}, *}$ and its ex-ante utility. We separate the approximation into a simulation and an aggregation phase.

We start with the simulation phase by sampling a single initial game state, and the players receive their respective signals $s_{\argdot 1}$. We collect the opponent actions $a_{-i1}$ according to $\beta_{-i1}$. For player $i$, we register into which grid cell $C_{i1}^k$ the signal $s_{i1}$ belongs and increase the cell's counter (aka.\ visitation count), which we denote by $M(C_{i1}^k)$. Then, we simulate the transition to the next stage for every possible action $a_{i1} \in \mathcal{A}_{i1}^{\text{D}}$, multiplying the number of simulated games by a factor of $| \mathcal{A}_{i1}^{\text{D}}|$. We proceed in this pattern; for each simulation, collect the opponent actions according to $\beta_{-it}$, register the corresponding grid cell $C_{it}^k$ for player $i$'s signal $s_{it}$, increase a respective counter $M(C_{it}^k)$, and simulate the state transition for every possible action $a_{it} \in \mathcal{A}_{it}^{\text{D}}$ multiplying the number of simulated games by a factor of $|\mathcal{A}_{it}^{\text{D}}|$.\footnote{One can decrease this branching factor sometimes using game-specific knowledge. For example, if an agent loses in the first stage of the signaling contest, he or she may no longer bid in the second stage.} After $T$ stages, there are $\prod_{t\in T}\,|\,\mathcal{A}_{it}^{\text{D}} |$ complete histories $a$, for which the utility $u_i(a)$ is evaluated. This procedure is performed for $M_{\text{IS}} \in \mathbb{N}$ initial states, resulting in a total of $M_{\text{Tot}} = M_{\text{IS}} \cdot \prod_{t\in T}\,|\,\mathcal{A}_{it}^{\text{D}} |$ simulated histories and evaluated utilities, concluding the simulation phase. We denote the set of all simulated histories by $A_{M_{\text{IS}}}$.

After performing $M_{\text{Tot}}$ simulations, the aggregation phase starts. Depending on which subsets of $A_{M_{\text{IS}}}$ are chosen, we get samples from different distributions. For example, let $\beta_i \in \Sigma_i^{\text{D}}$ arbitrary. Consider the rollout procedure above with a single initial state. Then, as we explore every discrete action, there exists a simulated history $a^{l}$ that is consistent with $\beta_i$. That is, $a_{it}^l = \beta_i(\sigma_{it}(a^l_{<t}))$ for $1 \leq t \leq T$. Due to construction, we have that $a^l \sim P(\argdot \, | \, \beta_{i}, \beta_{-i})$. Therefore, for every initial state, we get at least one sample for every possible $\beta_i \in \Sigma_{i}^{\text{D}}$. 

To perform the backward induction procedure as described above, we want to sample from conditional measures as well. So, for a grid cell $C_{it}^k \subset S_{it}$ and discretized action $a_{it}^k \in \mathcal{A}_{it}^{\text{D}}$, let $\beta_i \in \Sigma_i^{\text{D}}$ such that $\beta_i = \left(\beta_i \right)^{C_{it}^k, a_{it}^k}$. That is, $\beta_i$ is playing to reach $C_{it}^k$ and then plays $a_{it}^k$. The above procedure allows us to sample $a^l \sim P(\argdot \, | \, \beta_i, \beta_{-i})$. Suppose $P_{it}(C_{it}^k \, | \, \beta_i, \beta_{-i}) > 0$ and we only consider those $\tilde{a}^l$ with $\sigma_{it}(\tilde{a}^l_{<t}) \in C_{it}^k$, then $\tilde{a}^l \sim P(\argdot \, | \, C_{it}^k, (\beta_i, \beta_{-i}))$. The set of simulated histories $\tilde{a}^l$ for grid cell $C_{it}^k$, discrete action $a_{it} \in \mathcal{A}_{it}^{\text{D}}$, and step functions $\beta_{i>t}^{\text{D}} \in \Sigma_{>t}^{\text{D}}$ is given by
\begin{align*}
	A\left(\beta_{i>t}^{\text{D}}, C_{it}^k, a_{it}; M_{\text{IS}} \right) :=& \left\{a^{l} \in A_{M_{\text{IS}}} \,\big|\, \sigma_{it}(a^{l}_{<t}) \in C_{it}^k, a_{it}^{l} = a_{it}, \right. \nonumber\\
	&\;\left. \beta_{it + m}^{\text{D}}\left(\sigma_{it+m-1}\left(a_{<t+m}^{l} \right) \right) = a_{it+m}^{l} \text{ for } 1 \leq m \leq T-t\right\}.
\end{align*}
It holds that $\left|A\left(\beta_{i>t}^{\text{D}}, C_{it}^k, a_{it}; M_{\text{IS}}\right)\right| = M(C_{it}^k)$ for any $a_{it} \in \mathcal{A}_{it}^{\text{D}}$ and $\beta_{i>t}^{\text{D}} \in \Sigma_{>t}^{\text{D}}$. That is, we get a valid sample from $P(\argdot \, | \, C_{it}^k, ((\beta_{i>t})^{C_{it}^k, a_{it}^k}, \beta_{-i}))$ whenever a simulation falls into $C_{it}^k$ for every discrete action $a_{it}^k$.
We define the estimated counterfactual conditional utility by
\begin{align}\label{equ:estimated-counterfactual-conditional-utility}
	\hat{u}_i^{\text{c, D}}\left(\beta_{i>t}^{\text{D}}, \beta_{-i} \,\big|\, A_{t}\left(\beta_{i>t}^{\text{D}}, C_{it}^k, a_{it}; M_{\text{IS}}\right)\right) := \frac{1}{M(C_{it}^k)} \sum_{l=1}^{M(C_{it}^k)} u_i\left(a^{l} \right),
\end{align}
for $\beta_{i>t}^{\text{D}} \in \Sigma_{i>t}^{\text{D}}$, $a_{it} \in \mathcal{A}_{it}^{\text{D}}$, grid cell $C_{it}^k$, and $a^{l} \in A\left(\beta_{i>t}^{\text{D}}, C_{it}^k, a_{it}; M_{\text{IS}}\right)$. If $M(C_{it}^k) = 0$, we set the value to zero.

This approximates the counterfactual conditional utility from Equation \ref{equ:backward-induction-stage-T-analytically}. We use these to construct a step function $\beta_{i}^{\text{D}, M_{\text{IS}}} \in \Sigma_{i}^{\text{D}}$ according to the backward induction procedure from Equations \ref{equ:backward-induction-stage-T-analytically} and \ref{equ:backward-induction-intermediate-stages-t-analytically}. From this, using the relation between the counterfactual conditional and ex-ante utilities (see Lemma \ref{thm:relation-t+1-to-t-conditional-counterfactual-utilities}), we get an estimated best response utility over the simulations $A_{M_{\text{IS}}}$ which we define by
\begin{align} \label{equ:estimated-best-response-utility-over-monte-carlo-simulation}
	\hat{u}_i^{\text{ver, D}}(\beta_{-i}\,|\,A_{M_{\text{IS}}}) := \sum_{k=1}^{G_{S_{i1}}^\text{D}} \frac{M(C_{i1}^k)}{\sum_j M(C_{i1}^j)} \hat{u}_i^{\text{c, D}}\left(\beta_{i>1}^{\text{D}, M_{\text{IS}}}, \beta_{-i}\,\big|\,A_1 \left(\beta_{i>1}^{\text{D}, M_{\text{IS}}}, C_{i1}^k, \beta_{i1}^{\text{D}, M_{\text{IS}}}(C_{i1}^k); M_{\text{IS}} \right) \right).
\end{align}

In the limit, the approximation recovers the maximum utility and best response over the set of step function $\Sigma_{i}^{\text{D}}$ (see Lemma \ref{thm:estimated-counterfactual-conditional-utility-converges-to-analytical-value}).
We denote the \emph{simulation error} by 
\begin{align} \label{equ:simulation-error}
	\varepsilon_{M_{\text{IS}}} := \sup_{\beta^{\prime}_i \in \Sigma_{i}^{\text{D}}} \tilde{u}_i(\beta^{\prime}_i, \beta_{-i}) - \hat{u}_i^{\text{ver, D}}(\beta_{-i}\,|\,A_{M_{\text{IS}}}).
\end{align}

Finally, let $\beta \in \Sigma$ be a strategy profile. Then, we simulate $M_{\text{IS}}$ complete histories from $P(\cdot \, | \, \beta)$, and collect them into a data set $B_{M_{\text{IS}}}$. Using Equation \ref{equ:monte-carlo-estimation-for-ex-ante-utility}, we obtain an estimation of the expected utility, which we denote by $\hat{u}_i(\beta | B_{M_{\text{IS}}})$.\\
The final estimation of our verification procedure for the utility loss over pure Lipschitz continuous strategies is then given by
\begin{align}\label{equ:final-utility-loss-estimation-from-verifcation}
	\ell^{\text{ver}}(\beta) := \hat{u}_i^{\text{ver, D}}(\beta_{-i}\,|\,A_{M_{\text{IS}}}) - \hat{u}_i(\beta | B_{M_{\text{IS}}}).
\end{align}

\section{Hyperparameters} \label{sec:hyperparameters}

We employ common hyperparameters for our experiments, utilizing fully connected neural networks with two hidden layers, each consisting of $64$ nodes, and employing SeLU activations on the inner nodes \citep{klambauerSelfnormalizingNeuralNetworks2017}. The weights and biases of these networks determine the parameters $\theta_i$.
All experiments were conducted on a single Nvidia GeForce 2080Ti GPU with 11 gigabytes of RAM, accommodating a parallel simulation of $20,000$ environments. We employed the ADAM optimizer. We tuned the learning rate and the initial log-standard deviation for the individual settings. We used a learning rate of $8\times10^{-6}$ for all experiments in the sequential auction and a learning rate of $5 \times 10^{-5}$ for all experiments of the Bertrand competition. We set the initial log-standard deviation for the sequential auction and Bertrand competition to $-3.0$ for all experiments. For the elimination contest, we used the same learning rate and initial log-standard deviation as in the sequential auction environment as default. However, we made the following adaptations for the individual settings. 
We set the initial log-standard deviation in the standard setting where the winning bids are published to $-2.0$ for \textsc{Reinforce}. We set the learning rate with risk-averse contestants for PPO to $6\times 10^{-6}$. Finally, we set the learning rate to $1 \times 10^{-5}$ and the initial log-standard deviation to $-2.0$ for the \textsc{Reinforce} algorithm in the setting with asymmetric contestants.
All remaining parameters are set to the default values used in the framework by \citet{raffinStableBaselines3ReliableReinforcement2021}.

For our verification procedure, we employ a discretization parameter of $\text{D} = 64$ and an initial number of simulations $M_{\text{IS}} = 2^{23}$ as default. We increase the discretization to $\text{D} = 128$ in the Stackelberg Bertrand competition, and reduce it to $\text{D} = 16$ for the four-stage Sequential Auction, so that the information set tree fits onto a single GPU.

The code will be available at \href{https://github.com/404}{blinded for review}.

\section{Proof of Verifier Convergence Theorem} \label{sec:proof-of-main-theorem}

Before stating the proof of the main theorem, we show several auxiliary results.

\begin{lemma} \label{thm:relation-t+1-to-t-conditional-counterfactual-utilities}
	Let $\Gamma = \left(\mathcal{N}, T, S, \mathcal{A}, p, \sigma, u \right)$ be a multi-stage game under Assumption~\ref{ass:bounded-signaling-and-action-spaces}, $\beta_{-i} \in \Sigma_{-i}$, and $\beta_{i} \in \Sigma_{i}^{\text{D}}$ for $\text{D} \in \mathbb{N}$. For a grid cell $C_{it}^k \subset S_{it}$ that $C_{it}^k$ is reachable under $\beta_{-i}$ and $a_{it}^k \in \mathcal{A}_{it}^{\text{D}}$, consider $J \subset \{1, \dots, G_{S_{it+1}} \}$ such that $C_{it+1}^j$ is reachable from $C_{it}^k$ by taking $a_{it}^k$ under $\beta_{-i}$, i.e., $P_{it+1}\left(C_{it+1}^j | \left(\beta_i \right)^{C_{it}^k, a_{it}^k}, \beta_{-i} \right) > 0$, then there is the following relationship between the conditional probabilities from stage $t+1$ to stage $t$:
	\begin{align*}
		&P_{it}\left(C_{it}^k | \left(\beta_i \right)^{C_{it}^k, a_{it}^k}, \beta_{-i} \right) \cdot \tilde{u}_i^{\text{c, D}} \left(\beta_i, \beta_{-i} | C_{it}^k, a_{it}^k \right) \\
		&= \sum_{j \in J} P_{it+1} \left(C_{it+1}^j | \left(\beta_i \right)^{C_{it+1}^j, \beta_i(C_{it+1}^j)}, \beta_{-i} \right) \cdot \tilde{u}_i^\text{c, D}\left(\beta_i, \beta_{-i} | C_{it+1}^j, \beta_i(C_{it+1}^j) \right)
	\end{align*}
	
	In particular, it holds that
	\begin{align*} \label{equ:analytical-relation-counterfactual-conditional-and-ex-ante-utilities}
		\tilde{u}_i(\beta_{i}, \beta_{-i}) = \sum_{k=1}^{G_{S_{i1}}^\text{D}} P_{i1}\left(C_{i1}^k |\beta_{i} , \beta_{-i}\right) \cdot \tilde{u}_i^{\text{c, D}} \left(\beta_{i}, \beta_{-i}\,|\,C_{i1}^k, \beta_i(C_{i1}^k) \right).
	\end{align*}
	So, choosing $a_{it}^k = \beta_i(C_{it}^k)$ for every $t$, we can calculate player $i$'s ex-ante utility $\tilde{u}_i(\beta_i, \beta_{-i})$ by iteratively summing up the conditional probabilities.
\end{lemma}

\begin{proof}
	The second statement follows directly from the first by seeing that $G_{S_{i1}^\text{D}} = 1$, as there is only a single signal that can be received in stage $t=1$.
	
	For the first statement, note that due to construction and perfect recall, it holds that $\left(\beta_i \right)^{C_{it}^k, a_{it}^k} = \left(\beta_i \right)^{C_{it+1}^j, \beta_i(C_{it+1}^j)}$ for all $j \in J$. We then have
	\begin{align}
		&\sum_{j \in J} P_{it+1} \left(C_{it+1}^j | \left(\beta_i \right)^{C_{it+1}^j, \beta_i(C_{it+1}^j)}, \beta_{-i} \right) \cdot \tilde{u}_i^\text{c, D}\left(\beta_i, \beta_{-i} | C_{it+1}^j, \beta_i(C_{it+1}^j) \right) \\
		&=\sum_{j \in J} P_{it+1} \left(C_{it+1}^j | \left(\beta_i \right)^{C_{it+1}^j, \beta_i(C_{it+1}^j)}, \beta_{-i} \right) \cdot \int_{\mathcal{A}} u_i(a) d P\left(a | C_{it+1}^j, \left(\beta_i \right)^{C_{it+1}^j, \beta_i(C_{it+1}^j)}, \beta_{-i}  \right) \label{equ:relation-t+1-to-t-conditional-counterfactual-utilities-use-def-counterfactual-utilities}\\
		&=\sum_{j \in J} \int_{\{a \in \mathcal{A} | \sigma_{it+1}(a_{<t+1} \in C_{it+1}^j)  \}} u_i(a) d P \left(a | \left(\beta_i \right)^{C_{it+1}^j, \beta_i(C_{it+1}^j)}, \beta_{-i} \right) \label{equ:relation-t+1-to-t-conditional-counterfactual-utilities-zse-def-conditional-measure}\\
		&= \sum_{j \in J} \int_{\{a \in \mathcal{A} | \sigma_{it+1}(a_{<t+1} \in C_{it+1}^j)  \}} u_i(a) d P \left(a | \left(\beta_i \right)^{C_{it}^k, a_{it}^k}, \beta_{-i} \right) \\
		&= \int_{\{a \in \mathcal{A} | \sigma_{it+1}(a_{<t+1} \in \bigcup_{j \in J} C_{it+1}^j)  \}} u_i(a) d P \left(a | \left(\beta_i \right)^{C_{it}^k, a_{it}^k}, \beta_{-i} \right) \label{equ:relation-t+1-to-t-conditional-counterfactual-utilities-pull-sum-into-set-of-int} \\
		&= P_{it}\left(C_{it}^k | \left(\beta_i \right)^{C_{it}^k, a_{it}^k}, \beta_{-i} \right) \cdot \tilde{u}_i\left(\left(\beta_i \right)^{C_{it}^k, a_{it}^k}, \beta_{-i} | C_{it}^k \right) \label{equ:relation-t+1-to-t-conditional-counterfactual-utilities-sum-over-js-is-k} \\
		&= P_{it}\left(C_{it}^k | \left(\beta_i \right)^{C_{it}^k, a_{it}^k}, \beta_{-i} \right) \cdot \tilde{u}_i^{\text{c, D}} \left(\beta_i, \beta_{-i} | C_{it}^k, a_{it}^k \right),
	\end{align}
	where we used the definitions of the counterfactual conditional utilities and the conditional measures in Equations \ref{equ:relation-t+1-to-t-conditional-counterfactual-utilities-use-def-counterfactual-utilities} and \ref{equ:relation-t+1-to-t-conditional-counterfactual-utilities-zse-def-conditional-measure}. Finally, we used that the $C_{it+1}^j$'s are disjoint and $P_{it}\left(C_{it}^k | \left(\beta_i \right)^{C_{it}^k, a_{it}^k}, \beta_{-i} \right)=\sum_{j \in J} P_{it+1} \left(C_{it+1}^j | \left(\beta_i \right)^{C_{it}^k, a_{it}^k}, \beta_{-i} \right)$ in Equations \ref{equ:relation-t+1-to-t-conditional-counterfactual-utilities-pull-sum-into-set-of-int} and \ref{equ:relation-t+1-to-t-conditional-counterfactual-utilities-sum-over-js-is-k} respectively.
\end{proof}

\begin{lemma} \label{thm:backward-induction-gives-sup-over-step-functions}
	For a given multi-stage game $\Gamma = \left(\mathcal{N}, T, S, \mathcal{A}, p, \sigma, u \right)$, under Assumption~\ref{ass:bounded-signaling-and-action-spaces}, opponent strategies $\beta_{-i} \in \Sigma_{-i}$, and precision parameter $D \in \mathbb{N}$, it holds that
	\begin{align*}
		\tilde{u}_i\left(\beta_{i}^{\text{D}, *}, \beta_{-i} \right) = \sup_{\beta_i^{\prime} \in \Sigma_i^{\text{D}}} \tilde{u}_i\left(\beta_i^{\prime}, \beta_{-i} \right).
	\end{align*}
\end{lemma}

\begin{proof}
	First, note the following property for conditional probabilities. For $\beta_i, \beta_i^{\prime} \in \Sigma_i^{\text{D}}$ and any grid cell $C_{it}^k \subset S_{it}^k$, it holds that
	\begin{align} \label{equ:backward-induction-gives-sup-over-step-functions-independence-conditional-prob-of-step-function}
		P_{it}\left(C_{it}^k | \left(\beta_i \right)^{C_{it}^k, \beta_i(C_{it}^k)}, \beta_{-i} \right) = P_{it}\left(C_{it}^k | \left(\beta_i^{\prime} \right)^{C_{it}^k, \beta_i^{\prime}(C_{it}^k)}, \beta_{-i} \right).
	\end{align}
	That is due to two reasons. First, the conditional probabilities of stage $t$ only depend on the strategies prior to stage $t$. Second, a discrete strategy $\left(\beta_i \right)^{C_{it}^k, a_{it}}$ conditioned on grid cell $C_{it}^k$ is independent of $\beta_{i<t}$.
	
	Next, we show that for all $it \in L$, $\beta_i \in \Sigma_i^{\text{D}}$, and $1 \leq k \leq G_{S_{it}}$ the following holds
	\begin{align} \label{equ:backward-induction-gives-sup-over-step-functions-d-star-upper-bound-for-conditionals}
		\tilde{u}_i^{\text{c, D}}\left(\beta_{i}, \beta_{-i} | C_{it}^k, \beta_i(C_{it}^k) \right) \leq \tilde{u}_i^{\text{c, D}}\left(\beta_{i}^{\text{D}, *}, \beta_{-i} | C_{it}^k, \beta_{i}^{\text{D}, *}(C_{it}^k) \right).
	\end{align}	
	We perform a proof by induction. Let $k \in \left\{1, \dots G_{S_{iT}}\right\}$, then it holds that
	\begin{align*}
		\tilde{u}_i^{\text{c, D}}\left(\beta_{i}, \beta_{-i} | C_{iT}^k, \beta_i(C_{iT}^k) \right) &\leq \max_{a_{iT} \in \mathcal{A}_{iT}^\text{D}} \tilde{u}_i^{\text{c, D}}\left(\beta_{i}, \beta_{-i} | C_{iT}^k, a_{iT} \right) \\
		&= \tilde{u}_i^{\text{c, D}}\left(\beta_{i}^{\text{D}, *}, \beta_{-i} | C_{iT}^k, \beta_{i}^{\text{D}, *}(C_{iT}^k) \right).
	\end{align*}
	This can be seen directly as for any $C_{it}^k$, the counterfactual conditional utility is independent of $\beta_{i<t}$. Suppose \autoref{equ:backward-induction-gives-sup-over-step-functions-d-star-upper-bound-for-conditionals} holds for $t+1, \dots, T$. Let $k \in \left\{1, \dots, G_{S_{it}}\right\}$ and $a_{it} \in \mathcal{A}_{it}^{\text{D}}$. Denote with $J^{a_{it}} \subset \left\{1, \dots, G_{S_{it+1}}\right\}$ the subset of reachable grid cells from cell $C_{it}^k$ by taking action $a_{it}$. Then we get by \autoref{thm:relation-t+1-to-t-conditional-counterfactual-utilities}
	\begin{align*}
		&P_{it}\left(C_{it}^k | \left(\beta_i \right)^{C_{it}^k, \beta_i(C_{it}^k)}, \beta_{-i} \right) \cdot \tilde{u}_i^{\text{c, D}}\left(\beta_{i}, \beta_{-i} | C_{it}^k, \beta_i(C_{it}^k) \right)\\
		&= \sum_{j \in J^{\beta_i(C_{it}^k)}} P_{it+1} \left(C_{it+1}^j | \left(\beta_i \right)^{C_{it+1}^j, \beta_i(C_{it+1}^j)}, \beta_{-i} \right) \cdot \tilde{u}_i^\text{c, D}\left(\beta_i, \beta_{-i} | C_{it+1}^j, \beta_i(C_{it+1}^j) \right) \\
		&\leq \max_{a_{it} \in \mathcal{A}_{it}^\text{D}} \sum_{j \in J^{a_{it}}} P_{it+1} \left(C_{it+1}^j | \left(\beta_i \right)^{C_{it+1}^j, \beta_i(C_{it+1}^j)}, \beta_{-i} \right) \cdot \tilde{u}_i^\text{c, D}\left(\beta_i, \beta_{-i} | C_{it+1}^j, \beta_i(C_{it+1}^j) \right) \\
		&\overset{\text{(IS)}}{\leq}  \max_{a_{it} \in \mathcal{A}_{it}^\text{D}} \sum_{j \in J^{a_{it}}} P_{it+1} \left(C_{it+1}^j | \left(\beta_i \right)^{C_{it+1}^j, \beta_i(C_{it+1}^j)}, \beta_{-i} \right) \cdot \tilde{u}_i^\text{c, D}\left(\beta_{i}^{\text{D}, *}, \beta_{-i} | C_{it+1}^j, \beta_{i}^{\text{D}, *}(C_{it+1}^j) \right) \\
		&\overset{\text{(*)}}{=} \sum_{j \in J^{\beta_i^{\text{D}, *}(C_{it}^k)}} P_{it+1} \left(C_{it+1}^j | \left(\beta_i^{\text{D}, *} \right)^{C_{it+1}^j, \beta_i^{\text{D}, *}(C_{it+1}^j)}, \beta_{-i} \right) \cdot \tilde{u}_i^\text{c, D}\left(\beta_{i}^{\text{D}, *}, \beta_{-i} | C_{it+1}^j, \beta_{i}^{\text{D}, *}(C_{it+1}^j) \right) \\
		&=P_{it}\left(C_{it}^k | \left(\beta_{i}^{\text{D}, *} \right)^{C_{it}^k, \beta_{i}^{\text{D}, *}(C_{it}^k)}, \beta_{-i} \right) \cdot \tilde{u}_i^{\text{c, D}}\left(\beta_{i}^{\text{D}, *}, \beta_{-i} | C_{it}^k, \beta_{i}^{\text{D}, *}(C_{it}^k) \right),
	\end{align*}
	where $\text{(IS)}$ denotes the induction step. Furthermore, we used \autoref{equ:backward-induction-gives-sup-over-step-functions-independence-conditional-prob-of-step-function} and the definition of $\beta_{it}^{\text{D}, *}$ from \autoref{equ:backward-induction-intermediate-stages-t-analytically} in step $\text{(*)}$. By applying \autoref{equ:backward-induction-gives-sup-over-step-functions-independence-conditional-prob-of-step-function} again, we get the statement from \autoref{equ:backward-induction-gives-sup-over-step-functions-d-star-upper-bound-for-conditionals}.
	Finally, by applying \autoref{thm:relation-t+1-to-t-conditional-counterfactual-utilities}, we get the statement. 	
\end{proof}

\begin{lemma} \label{thm:estimated-counterfactual-conditional-utility-converges-to-analytical-value}
	Let $\Gamma = \left(\mathcal{N}, T, S, \mathcal{A}, p, \sigma, u \right)$ be a multi-stage game, $\beta_{-i} \in \Sigma_{-i}$, assuming Assumptions \ref{ass:bounded-signaling-and-action-spaces} and \ref{ass:lipschitz-continuous-signals-and-ultimately-strategies} hold, $\text{D} \in \mathbb{N}$, and $A_{M_{\text{IS}}}$ with a number of $M_{\text{IS}} \in \mathbb{N}$ initial simulations. 
	It holds that
	\begin{align*}
		\lim_{M_{\text{IS}}  \rightarrow \infty} \varepsilon_{M_{\text{IS}}} = 0 \text{ almost surely.}
	\end{align*}
\end{lemma}
\begin{proof}
	Let $C_{it}^k \subset S_{it}$ be reachable, i.e., $P_{it}(C_{it}^k \, |\, \beta_i^{\prime}, \beta_{-i}) > 0$ for some $\beta_i^{\prime} \in \Sigma_i^{\text{D}}$, and $a_{it}^k  \in \mathcal{A}_{it}^{\text{D}}$ and $\beta_i \in \Sigma_i^{\text{D}}$ be arbitrary.
	Then it follows that $M(C_{it}^k) \rightarrow \infty$ for $M_{\text{IS}} \rightarrow \infty$. That is, for any reachable $C_{it}^k$, we sample infinitely often from the conditional counterfactual measure $P(\argdot \, | \, C_{it}^k, (\beta_i)^{C_{it}^k, a_{it}^k}, \beta_{-i})$.	
	Furthermore, it holds that
	\begin{align*}
		\tilde{u}_i^{\text{c, D}}(\beta_i, \beta_{-i} \, | \, C_{it}^k, a_{it}^k) \leq \int_{\mathcal{A}}\left| u_i(a) \right| d P(a \, | \, C_{it}^k, (\beta_i)^{C_{it}^k, a_{it}^k}, \beta_{-i}) \leq \left| \left|u_i\right|\right|_{\infty} < \infty.
	\end{align*}
	Finally, as $a^l \in A\left(\beta_{i>t}, C_{it}^k, a_{it}^k; M_{\text{IS}} \right)$ is distributed according to $P(\argdot \, | \, C_{it}^k, (\beta_i)^{C_{it}^k, a_{it}^k}, \beta_{-i})$, Kolmogorov's law of large numbers holds and
	\begin{align*}
		\lim_{M_{\text{IS}} \rightarrow \infty} \hat{u}_i^{\text{c, D}}\left(\beta_{i>t}, \beta_{-i} | A\left(\beta_{i>t}, C_{it}^k, a_{it}^k; M_{\text{IS}} \right) \right) = \tilde{u}_i^{\text{c, D}}\left(\beta_i, \beta_{-i} | C_{it}^k, a_{it}^k \right) \text{almost surely.} 
	\end{align*}
	In particular, it holds that
	\begin{align*}
		\lim_{M_{\text{IS}} \rightarrow \infty} \hat{u}_i^{\text{c, D}}\left(\beta_{-i} | A\left(C_{iT}^k, \beta_i^{\text{D}, M_{\text{IS}}}(C_{iT}^k); M_{\text{IS}} \right) \right) = \tilde{u}_i^{\text{c, D}}\left(\beta_i, \beta_{-i} | C_{iT}^k, \beta_i^{\text{D}, *}(C_{iT}^k) \right) \text{almost surely,}
	\end{align*}
	as the utility is independent from the choice of the $\argmax$ in \autoref{equ:backward-induction-stage-T-analytically}. Therefore, following the backward induction procedure from \autoref{equ:backward-induction-intermediate-stages-t-analytically}, we get
	\begin{align*}
		\lim_{M_{\text{IS}} \rightarrow \infty} \hat{u}_i^{\text{c, D}}\left(\beta_{i>t}^{\text{D}, M_{\text{IS}}}, \beta_{-i} | A\left(\beta_{i>t}^{\text{D}, M_{\text{IS}}}, C_{it}^k, \beta_{i>t}^{\text{D}, M_{\text{IS}}}(C_{it}^k); M_{\text{IS}} \right) \right) = \tilde{u}_i^{\text{c, D}}\left(\beta_i^{\text{D}, *}, \beta_{-i} | C_{it}^k, \beta_i^{\text{D}, *}(C_{it}^k) \right) \text{ a. s.} 
	\end{align*}
	Finally, we get
	\begin{align*}
		\lim_{M_{\text{IS}} \rightarrow \infty} \hat{u}_i^{\text{ver, D}}(\beta_{-i}\,|\,A_{M_{\text{IS}}}) = \tilde{u}_i \left(\beta_i^{\text{D}, *}, \beta_{-i} \right) = \sup_{\beta^{\prime}_i \in \Sigma_{i}^{\text{D}}} \tilde{u}_i \left(\beta^{\prime}_i, \beta_{-i} \right) \text{ almost surely} ,
	\end{align*}
	where we used \autoref{thm:relation-t+1-to-t-conditional-counterfactual-utilities} for the first, and \autoref{thm:backward-induction-gives-sup-over-step-functions} for the second equation. This gives us the desired statement.
\end{proof}

\begin{lemma} \label{thm:finite-step-functions-arbitrarily-close-to-continuous-functions}
	Let $\Gamma = \left(\mathcal{N}, T, S, \mathcal{A}, p, \sigma, u \right)$ be a multi-stage game, where Assumption~\ref{ass:bounded-signaling-and-action-spaces} holds. Let $\mathcal{G}_{it} = (S_{it}^{\text{D}}, \mathcal{A}_{it}^{\text{D}}, C_{it}, \text{D})$ be the discretization that results from a regular grid of $2^{\text{D}}$ points along each dimension of $S_{it}$ and $\mathcal{A}_{it}$. Then for every pure Lipschitz continuous strategy $\beta_{it} \in \Sigma_{it}^{\text{Lip, p}}$ there exists a sequence $\left\{\beta_{it}^{\text{D}} \right\}_{\text{D} \in \mathbb{N}}$ with $\beta_{it}^{\text{D}} \in \Sigma_{it}^{\text{D}}\left(\mathcal{G}_{it}\right)$ such that
	\begin{align*}
		\lim_{\text{D} \rightarrow \infty} \left|\left|\beta_{it} - \beta_{it}^{\text{D}} \right| \right|_{\infty} = 0.
	\end{align*}
\end{lemma}
\begin{proof}
	Let $\beta_{it} \in \Sigma_{it}^{\text{Lip, p}}$. Then there exists an $L > 0$ such that for $s, s^{\prime} \in S_{it}$ it holds that
	\begin{align*}
		d_W\left(\beta_{it}(s), \beta_{it}(s^{\prime}) \right) = \left|\left|\beta_{it}(s) - \beta_{it}(s^{\prime}) \right| \right|_{\infty} \leq L \cdot \left|\left| s - s^{\prime} \right| \right|.
	\end{align*}
	As $\mathcal{A}_{it}$ and $S_{it}$ are compact, there exists a $K> 0$ such that $\left|\left|a - a^{\prime} \right| \right|_{\infty} \leq K$ and $\left|\left|s - s^{\prime} \right| \right|_{\infty} \leq K$ for all $a, a^{\prime} \in \mathcal{A}_{it}$ and $s, s^{\prime} \in S_{it}$. Then, for every $a \in \mathcal{A}_{it}$ and $s \in S_{it}$, there exist $\tilde{a} \in \mathcal{A}_{it}^{\text{D}}$ and $\tilde{s} \in S_{it}^{\text{D}}$ such that $\left|\left|a - \tilde{a} \right| \right|_{\infty} \leq K 2^{- \text{D}}$ and $\left|\left|s - \tilde{s} \right| \right|_{\infty} \leq K 2^{- \text{D}}$.
	
	Remember that $\left(C_{it}^k\right)_{1 \leq k \leq G_{S_{it}^{\text{D}}}}$ partitions $S_{it}$. Let $s \in C_{it}^k$, then there is a unique $s^{k} \in S_{it}^{\text{D}} \cap C_{it}^k$. Define
	\begin{align*}
		\beta_{it}^{\text{D}}(s) = \argmin_{a_{it}^{Dk} \in \mathcal A_{it}^D} \left|\left|a_{it}^k - \beta_{it}(s^{k}) \right| \right|_{\infty} \text{ for } s \in C_{it}^k, 1 \leq k \leq G_{S_{it}}^{\text{D}}.
	\end{align*}
	Then $\beta_{it}^{\text{D}} \in \Sigma_{it}^{\text{D}}$ for every $\text{D} \in \mathbb{N}$. Finally, we get for all $s \in S_{it}$ that
	\begin{align*}
		\left|\left|\beta_{it}(s) - \beta_{it}^{\text{D}}(s^k) \right| \right|_{\infty} &\leq \left|\left|\beta_{it}(s) - \beta_{it}(s^k) \right| \right|_{\infty} + \left|\left|\beta_{it}(s^k) - \beta_{it}^{\text{D}}(s^k) \right| \right|_{\infty} \\
		& \leq L \left|\left|s - s^k \right| \right| + K2^{-\text{D}} \leq K \cdot (1 + L) 2^{- \text{D}}.
	\end{align*}
	As $K$ and $L$ are constants, and $2^{-\text{D}} \rightarrow 0$ for $\text{D} \rightarrow \infty$, we get the statement.
\end{proof}

We now turn to translate a close approximation of a strategy $\beta_i$ by a step function $\beta_i^{\prime}$ into closeness of the outcome distribution in the Wasserstein distance. For completeness, we restate some well-known results about the Wasserstein distance.

\begin{definition}[Wasserstein distance] \label{def:wasserstein-distance}
	\citep[p.93]{villaniOptimalTransport2009} Let $(X, d)$ be a Polish metric space. For any probability measures $\mu, \nu$ on $X$, the (1-)Wasserstein distance between $\mu$ and $\nu$ is defined by
	\begin{align*}
		d_W(\mu, \nu) = \inf_{\pi \in \Pi(\mu, \nu)} \int_X d(x, y) d \pi(x, y),
	\end{align*}
	where $\Pi(\mu, \nu)$ denotes the space of couplings between $\mu$ and $\nu$. That is $\pi \in \Pi(\mu, \nu)$ is a probability measure on $X \times X$, such that $\int_X \pi(x, y) d y = \mu(x)$ and $\int_X \pi(x, y) d x = \nu(x)$.
\end{definition}
In our applications, we always assume the metric $d$ to be induced by some norm $\left|\left| \, \cdot \,  \right| \right|$. As we consider finite-dimensional spaces by Assumption \ref{ass:bounded-signaling-and-action-spaces}, the choice of a norm is irrelevant.

\begin{lemma}[Kantorovich–Rubinstein duality] \label{thm:kantorovich-rubinstein-duality} 		\citep[p.59]{villaniOptimalTransport2009}
	Let $(X, d)$ be a Polish metric space. For any probability measures $\mu, \nu$ on $X$, and $K >0$, there holds the following equality
	\begin{align*}
		d_W(\mu, \nu) = \frac{1}{K}\sup_{||f||_{\text{Lip}}\leq K} \left\{\int f d \mu - \int f d\nu \right\},
	\end{align*}
	where $|| \argdot ||_{\text{Lip}}$ denotes the Lipschitz norm.
\end{lemma}

\begin{theorem}[Metrization of weak convergence] \label{thm:wasserstein-metrizes-weak-convergence}
	\citep[p.96]{villaniOptimalTransport2009} Let $(X, d)$ be a Polish metric space and $\left(\mu_k \right)_{k \in \mathcal{N}}$ is a sequence of measures in $P(X)$, and $\mu \in P(X)$, then the following two statements are equivalent 
	\begin{enumerate}
		\item $d_W(\mu_k, \mu) \rightarrow 0$
		\item For all bounded continuous functions $f: X  \rightarrow \mathbb{R}$, one has
		\begin{align*}
			\int f d \mu_k \rightarrow \int f d \mu.
		\end{align*}
	\end{enumerate}
\end{theorem}

To proof \autoref{thm:main-approximation-result-for-lipschitz-utilities}, we leverage \autoref{thm:wasserstein-metrizes-weak-convergence}. As the utilities are assumed to be continuous, it suffices to show that closeness to a Lipschitz continuous strategy translates to a small Wasserstein distance of the outcome distribution. More specifically, we show that under Assumptions \ref{ass:bounded-signaling-and-action-spaces} and \ref{ass:lipschitz-continuous-signals-and-ultimately-strategies}, for every $\beta_{i} \in \Sigma_{i}^{\text{Lip, p}}$ there exists a sequence $\left(\beta_i^{\text{D}} \right)_{\text{D} \in \mathcal{N}}$ with $\beta_i^{\text{D}} \in \Sigma_i^{\text{D}}$ such that $d_W\left( P\left(\argdot \, | \, \beta_i, \beta_{-i}\right), P(\argdot \, | \, \beta_i^{\text{D}}, \beta_{-i}) \right) \rightarrow 0$.
For this, we show some intermediate results first.

\begin{lemma} \label{thm:wasserstein-distance-of-product-measures-bounded-by-sum-of-distances}
	Let $\mu_1, \nu_1$ and $\mu_2, \nu_2$ be measures on $\mathbb{R}^n$ and $\mathbb{R}^m$ respectively. Define the product measures $\mu := \mu_1 \otimes \mu_2$, $\nu := \nu_1 \otimes \nu_2$. Then the following inequality holds
	\begin{align*}
		d_W(\mu, \nu) \leq d_W(\mu_1, \nu_1) + d_W(\mu_2, \nu_2)
	\end{align*}
\end{lemma}

\begin{proof}
	By Theorem 4.1 of \citep[p.43]{villaniOptimalTransport2009}, there exist optimal couplings $\pi_1, \pi_2$ for $\mu_1, \nu_1$ and $\mu_2, \nu_2$ respectively, such that 
	\begin{align*}
		d_W(\mu_1, \nu_1) = \int_{\mathbb{R}^{n} \times \mathbb{R}^{n}} d_{\mathbb{R}^n}(x_1, x_2) d \pi_1(x_1, x_2) \text{ , } d_W(\mu_2, \nu_2) = \int_{\mathbb{R}^{m} \times \mathbb{R}^{m}} d_{\mathbb{R}^m}(y_1, y_2) d \pi_2(y_1, y_2).
	\end{align*}
	Then $\pi := \pi_1 \otimes \pi_2$ is the trivial coupling for $\mu$ and $\nu$, which can be readily checked by
	\begin{align*}
		\int_{\mathbb{R}^n \times \mathbb{R}^m} \pi(x_1, x_2, y_1, y_2) d(x_1, y_1) = \int_{\mathbb{R}^n} \pi_1(x_1, x_2) d x_1 \int_{\mathbb{R}^m} \pi_2(y_1, y_2) d y_1 = \nu_1(x_2) \nu_2(y_2), \\
		\int_{\mathbb{R}^n \times \mathbb{R}^m} \pi(x_1, x_2, y_1, y_2) d(x_2, y_2) = \int_{\mathbb{R}^n} \pi_1(x_1, x_2) d x_2 \int_{\mathbb{R}^m} \pi_2(y_1, y_2) d y_2 = \mu_1(x_1) \mu_2(y_1).
	\end{align*}
	Therefore, we get
	\begin{align*}
		d_W(\mu, \nu) &\leq \int_{\mathbb{R}^{n+m} \times \mathbb{R}^{n+m}} d_{\mathbb{R}^{n+m}} (x_1, y_1, x_2, y_2) d \pi(x_1, x_2, y_1, y_2) \\
		&\leq \int_{\mathbb{R}^{n+m} \times \mathbb{R}^{n+m}} d_{\mathbb{R}^{n}} (x_1, y_1) +  d_{\mathbb{R}^{m}} (x_2, y_2) d \pi(x_1, x_2, y_1, y_2) \\
		&= \int_{\mathbb{R}^{n} \times \mathbb{R}^{n}} d_{\mathbb{R}^{n}} (x_1, y_1) d \pi_1(x_1, y_1) +  \int_{\mathbb{R}^{m} \times \mathbb{R}^{m}} d_{\mathbb{R}^{m}} (x_2, y_2) d \pi_2(x_2, y_2) \\
		& = d_W(\mu_1, \nu_1) + d_W(\mu_2, \nu_2)
	\end{align*}
\end{proof}

\begin{lemma} \label{thm:beta-beta-prime-lipschitz-condition-in-p-t}
	Let $\beta_{it}, \beta_{it}^{\prime} \in \Sigma_{it}^{\text{p}}$, $\beta_{-it} \in \Sigma_{-it}$ and $\beta_{<t} \in \Sigma_{<t}$. Under Assumption \ref{ass:bounded-signaling-and-action-spaces}, it holds that
	\begin{align*}
		d_W\left(P_{<t+1}\left(\argdot \, | \left(\beta_{i<t}, \beta_{it} \right), \left(\beta_{-i<t}, \beta_{-it}\right)\right), P_{<t+1}\left(\argdot \, | \left(\beta_{i<t}, \beta_{it}^{\prime} \right), \left(\beta_{-i<t}, \beta_{-it}\right)\right) \right) \leq || \beta_{it}- \beta_{it}^{\prime}||_{\infty}.
	\end{align*}
\end{lemma}

\begin{proof}
	We start with showing the following step first
	\begin{align} \label{equ:step-t-bound-on-wasserstein-distance-by-norm-difference}
		d_W\left(P_t\left(\argdot \, | \, a_{<t}, \beta_{it}, \beta_{-it}\right), P_t\left(\argdot \, | \, a_{<t}, \beta_{it}^{\prime}, \beta_{-it}\right) \right) \leq || \beta_{it}- \beta_{it}^{\prime}||_{\infty} \text{ for all } a_{<t} \in \mathcal{A}_{<t}.
	\end{align}
	Let $a_{<t} \in \mathcal{A}_{<t}$ be arbitrary. Then note first that as $\beta_{it}, \beta_{it}^{\prime}$ are pure strategies, $\beta_{it}(\argdot \, | \, \sigma_{it}(a_{<t}))$ and $\beta_{it}^{\prime}(\argdot \, | \,\sigma_{it}(a_{<t}))$ are Dirac-measures. Therefore, we can use the well-known fact that the Wasserstein distance between these is simply the distance between the points with positive measure (see Example 6.3 in \citep[p.94]{villaniOptimalTransport2009}), i.e.,
	\begin{align*}
		d_W\left(\beta_{it}(\argdot \, | \, \sigma_{it}(a_{<t})), \beta_{it}^{\prime}(\argdot \, | \,\sigma_{it}(a_{<t})) \right) = \left|\left|\beta_{it}(\sigma_{it}(a_{<t})) - \beta_{it}^{\prime}(\sigma_{it}(a_{<t})) \right| \right|_{\infty} \leq \left|  \left|\beta_{it} - \beta_{it}^{\prime} \right|\right|_{\infty},
	\end{align*}
	where we abused notation and treated $\beta_{it}, \beta_{it}^{\prime}$ once as mapping to Dirac-measures and once as mapping to elements in $\mathcal{A}_{it}$. By Assumption~\ref{ass:bounded-signaling-and-action-spaces}, $P_t$ is a product measure on $\mathbb{R}^m$ for some $m\in \mathbb{N}$. Therefore, one can use \autoref{thm:wasserstein-distance-of-product-measures-bounded-by-sum-of-distances} and get
	\begin{align*}
		&d_W\left(P_t\left(\argdot \, | \, a_{<t}, \beta_{it}, \beta_{-it}\right), P_t\left(\argdot \, | \, a_{<t}, \beta_{it}^{\prime}, \beta_{-it}\right) \right) \\
		& \leq d_W\left(\beta_{it}(\argdot \, | \, \sigma_{it}(a_{<t})), \beta_{it}^{\prime}(\argdot \, | \,\sigma_{it}(a_{<t})) \right) \leq \left|  \left|\beta_{it} - \beta_{it}^{\prime} \right|\right|_{\infty},
	\end{align*}
	which shows \autoref{equ:step-t-bound-on-wasserstein-distance-by-norm-difference}. Consequently, we get
	\begin{align*}
		&d_W\left(P_{<t+1}\left(\argdot \, | \left(\beta_{i<t}, \beta_{it} \right), \left(\beta_{-i<t}, \beta_{-it}\right)\right), P_{<t+1}\left(\argdot \, | \left(\beta_{i<t}, \beta_{it}^{\prime} \right), \left(\beta_{-i<t}, \beta_{-it}\right)\right) \right) \\
		&= \sup_{||f||_{\text{Lip}} \leq 1} \int_{\mathcal{A}_{<t}} \int_{\mathcal{A}_{\cdot t}} f\left(a_{<t}, a_{\cdot t}\right) d P_t\left(a_{\cdot t} \, | \, a_{<t}, \beta_{it}, \beta_{-it}\right) \\
		&- \int_{\mathcal{A}_{\cdot t}} f\left(a_{<t}, a_{\cdot t}\right) d P_t\left(a_{\cdot t} \, | \, a_{<t}, \beta_{it}^{\prime}, \beta_{-it}\right) d P_{<t}\left(a_{<t} \, | \beta_{<t}\right)	\\
		&\overset{(\ref{thm:kantorovich-rubinstein-duality})}{\leq} \int_{\mathcal{A}_{<t}} d_W\left(P_t\left(\argdot \, | \, a_{<t}, \beta_{it}, \beta_{-it}\right), P_t\left(\argdot \, | \, a_{<t}, \beta_{it}^{\prime}, \beta_{-it}\right) \right) d P_{<t}\left(a_{<t} \, | \beta_{<t}\right) \\
		&\overset{Equ. (\ref{equ:step-t-bound-on-wasserstein-distance-by-norm-difference})}{\leq} \int_{\mathcal{A}_{<t}} || \beta_{it}- \beta_{it}^{\prime}||_{\infty} d P_{<t}\left(a_{<t} \, | \beta_{<t}\right) = || \beta_{it}- \beta_{it}^{\prime}||_{\infty}.
	\end{align*}
\end{proof}

\begin{lemma} \label{thm:lipschitz-function-conditional-on-lipschitz-measure-is-again-lipschitz}
	Let $(X, d_X), (Y, d_Y)$ be metric Polish spaces, $f: X \times Y \rightarrow \mathbb{R}$ be a $K_f$-Lipschitz continuous function and $\mu(\argdot \, | \, x)$ be a measure on $Y$ for every $x \in X$. Furthermore, the mapping $x \mapsto \mu(\argdot \, | \, x)$ is $K_{\mu}$-Lipschitz continuous with respect to the Wasserstein distance $d_W$. Then it holds that
	\begin{align*}
		g_f:X \rightarrow \mathbb{R}, x \mapsto \int_Y f(x, y) d \mu(y | x) \text{ is } (K_f + K_f K_{\mu})-\text{Lipschitz}.
	\end{align*}
\end{lemma}

\begin{proof}
	Let $x, x^{\prime} \in X$, then
	\begin{align*}
		&| g_f(x) - g_f(x^{\prime}) | = \left|\int_Y f(x, y) d \mu(y| x) - \int_Y f(x^{\prime}, y) d \mu(y| x^{\prime})	\right| \\
		&\leq \left|\int_Y f(x, y) - f(x^{\prime}, y) d \mu(y| x)\right| + \left|\int_Y f(x^{\prime}, y) d \mu(y| x) - \int_Y f(x^{\prime}, y) d \mu(y| x^{\prime}) \right|\\
		&\leq \int_Y K_f \cdot d_{X\times Y}\left(\left(x, y\right)^T, \left(x^{\prime}, y\right)^T \right) d \mu(y|x) + \sup_{||g||_{\text{Lip}} \leq K_f} \left|\int_Y g(y) d \mu(y|x) - \int_Y g(y) d \mu(y | x^{\prime}) \right| \\
		& \overset{(\ref{thm:kantorovich-rubinstein-duality})}{=} K_f \int_Y d_X(x, x^{\prime}) d \mu(y|x) + K_f \cdot d_W\left(\mu\left(\argdot \, | \, x \right), \mu\left(\argdot \, | \, x^{\prime} \right) \right) \\
		&= K_f \left(d_X(x, x^{\prime}) + d_W\left(\mu\left(\argdot \, | \, x \right), \mu\left(\argdot \, | \, x^{\prime} \right) \right)	\right) \\
		&\leq K_f \left(d_X(x, x^{\prime}) + L_{\mu} d_X(x, x^{\prime})	\right) = \left(K_f + K_f K_{\mu}\right) d_X(x, x^{\prime}).
	\end{align*}
\end{proof}

\begin{lemma} \label{thm:backward-bound-for-wasserstein-with-lipschitz-continuous-strategies}
	Let Assumptions \ref{ass:bounded-signaling-and-action-spaces} and \ref{ass:lipschitz-continuous-signals-and-ultimately-strategies} hold, and let $\beta_{<t}, \beta_{<t}^{\prime} \in \Sigma_{<t}$ and $\beta_{\cdot t} \in \Sigma_{\cdot t}^{\text{Lip}}$. Then, there exists $K>0$ such that
	\begin{align*}
		d_W\left(P_{<t+1}\left(\argdot \, | \, \beta_{<t}, \beta_{\cdot t} \right),  P_{<t+1}\left(\argdot \, | \, \beta_{<t}^{\prime}, \beta_{\cdot t} \right) \right) \leq K \cdot d_W\left(P_{<t}\left(\argdot \, | \, \beta_{<t} \right),  P_{<t}\left(\argdot \, | \, \beta_{<t}^{\prime} \right) \right)
	\end{align*}
\end{lemma}

\begin{proof}
	By Assumption~\ref{ass:lipschitz-continuous-signals-and-ultimately-strategies}, there exist constants $K_{\sigma_{it}} > 0$ for $it \in L$, such that $\sigma_{it}$ is $K_{\sigma_{it}}$-Lipschitz continuous in $\mathcal{A}_{<t}$. Also, denote with $K_{0t}$ nature's Lipschitz constant with respect to $d_W$ in stage $t$. Similarly, as $\beta_{\cdot t} \in \Sigma_{\cdot t}^{\text{Lip}}$, there exist constants $K_{\beta_{it}} >0$ such that $\beta_{it}\left( \argdot \, | \, s_{it}\right)$ is $K_{\beta_{it}}$-Lipschitz with respect to the Wasserstein distance. Overall, we get for $it \in L$, $a_{<t}, a_{<t}^{\prime} \in \mathcal{A}_{<t}$
	\begin{align*}
		d_W\left(\beta_{it}\left(\argdot \, | \, \sigma_{it}(a_{<t}) \right), \beta_{it}\left(\argdot \, | \, \sigma_{it}(a_{<t}^{\prime}) \right) \right) \leq K_{\beta_{it}} K_{\sigma_{it}} \cdot \left| \left|a_{<t} - a_{<t}^{\prime} \right| \right|.
	\end{align*}
	Denote $K_t := K_{0t} + \sum_{i \in \mathcal{N}} K_{\beta_{it}} K_{\sigma_{it}}$. Then we get that the mapping $a_{<t} \mapsto P_t \left(\argdot \, | \, a_{<t}, \beta_{\argdot t} \right)$ is $K_t$-Lipschitz continuous with respect to $d_W$, which can be seen by
	\begin{align*}
		&d_W\left(P_t \left(\argdot \, | \, a_{<t}, \beta_{\argdot t} \right), P_t \left(\argdot \, | \, a_{<t}^{\prime}, \beta_{\argdot t} \right) \right) \\ &\overset{(\ref{thm:wasserstein-distance-of-product-measures-bounded-by-sum-of-distances})}{\leq} d_W\left(p_{0t}\left(\argdot \, | \, a_{<t} \right), p_{0t}\left(\argdot \, | \, a_{<t}^{\prime} \right)\right) + \sum_{i \in \mathcal{N}} d_W\left(\beta_{it}\left(\argdot \, | \, \sigma_{it}(a_{<t}) \right), \beta_{it}\left(\argdot \, | \, \sigma_{it}(a_{<t}^{\prime}) \right) \right) \\
		&\leq \left(K_{0t} + \sum_{i \in \mathcal{N}} K_{\beta_{it}} K_{\sigma_{it}} \right) \cdot \left| \left| a_{<t} - a_{<t}^{\prime} \right| \right|_{\infty},
	\end{align*}
	for all $a_{<t}, a_{<t}^{\prime} \in \mathcal{A}_{<t}$. Let $f: \mathcal{A}_{<t+1} \rightarrow \mathbb{R}$ be $1$-Lipschitz continuous. Then, we get by Lemma \ref{thm:lipschitz-function-conditional-on-lipschitz-measure-is-again-lipschitz}, that the function $g_f(a_{<t}) := \int_{\mathcal{A}_{<t}} f(a_{<t}, a_{\argdot t}) d P_t \left(a_{\argdot t} \, | \, a_{<t}, \beta_{\argdot t} \right)$ is $(1 + K_t)$-Lipschitz continuous. With this, we get
	\begin{align*}
		&d_W\left(P_{<t+1}\left(\argdot \, | \, \beta_{<t}, \beta_{\cdot t} \right),  P_{<t+1}\left(\argdot \, | \, \beta_{<t}^{\prime}, \beta_{\cdot t} \right) \right)\\
		& \overset{(\ref{thm:kantorovich-rubinstein-duality})}{=} \sup_{||f||_{\text{Lip}} \leq 1} \int_{\mathcal{A}_{<t+1}} f(a_{<t+1}) P_{<t+1}\left(a_{<t+1} \, | \, \beta_{<t}, \beta_{\cdot t} \right) - \int_{\mathcal{A}_{<t+1}} f(a_{<t+1}) P_{<t+1}\left(a_{<t+1} \, | \, \beta_{<t}^{\prime}, \beta_{\cdot t} \right) \\
		&= \sup_{||f||_{\text{Lip}} \leq 1} \int_{\mathcal{A}_{<t}} g_f(a_{<t}) P_{<t}\left(a_{<t} \, | \, \beta_{<t} \right) - \int_{\mathcal{A}_{<t}} g_f(a_{<t}) P_{<t}\left(a_{<t} \, | \, \beta_{<t}^{\prime} \right) \\
		&\overset{(\ref{thm:lipschitz-function-conditional-on-lipschitz-measure-is-again-lipschitz})}{\leq} \sup_{||g||_{\text{Lip}} \leq 1 + K_t} \int_{\mathcal{A}_{<t}} g(a_{<t}) P_{<t}\left(a_{<t} \, | \, \beta_{<t} \right) - \int_{\mathcal{A}_{<t}} g(a_{<t}) P_{<t}\left(a_{<t} \, | \, \beta_{<t}^{\prime} \right) \\
		& \overset{(\ref{thm:kantorovich-rubinstein-duality})}{=} \left(1 + K_t \right) \cdot d_W\left(P_{<t}\left(\argdot \, | \, \beta_{<t} \right),  P_{<t}\left(\argdot \, | \, \beta_{<t}^{\prime} \right) \right),
	\end{align*}
	which shows the statement. 
\end{proof}

\begin{lemma} \label{thm:close-strategies-in-infinity-imply-close-wasserstein-distance-of-ex-ante-measure}
	Let $\Gamma = \left(\mathcal{N}, T, S, \mathcal{A}, p, \sigma, u \right)$ be a multi-stage game, where Assumptions \ref{ass:bounded-signaling-and-action-spaces} and \ref{ass:lipschitz-continuous-signals-and-ultimately-strategies} hold. For strategies $\beta_{-i} \in \Sigma_{-i}^{\text{Lip}}$, $\beta_{i} \in \Sigma_{i}^{\text{Lip, p}}$ and $\epsilon >0$, there exists a $\delta>0$ such that for all $\beta_i^{\prime} \in \Sigma_{i}^{\text{p}}$ with $\left|\left|\beta_i - \beta_i^{\prime} \right| \right|_{\infty} < \delta$, it holds that 
	\begin{align*}
		d_W \left(P\left(\argdot \, | \, \beta_i, \beta_{-i} \right), P\left(\argdot \, | \, \beta_i^{\prime}, \beta_{-i} \right) \right) < \epsilon.
	\end{align*}
\end{lemma}
\begin{proof}
	Let $\epsilon > 0$, $\beta_{-i} \in \Sigma_{-i}^{\text{Lip}}$, $\beta_{i} \in \Sigma_{i}^{\text{Lip, p}}$, and $\beta_{i}^{\prime} \in \Sigma_{i}^{\text{p}}$. Then we get
	\begin{align*}
		&d_W \left(P\left(\argdot \, | \, \beta_i, \beta_{-i} \right), P\left(\argdot \, | \, \beta_i^{\prime}, \beta_{-i} \right) \right) \\
		& \overset{(\triangle \text{-inequ.})}{\leq} \sum_{t=1}^{T} d_W \left(P\left(\argdot \, | \, \left(\beta_{i<t}^{\prime}, \beta_{it}, \beta_{i>t} \right), \beta_{-i} \right), P\left(\argdot \, | \, \left(\beta_{i<t}^{\prime}, \beta_{it}^{\prime}, \beta_{i>t} \right), \beta_{-i} \right) \right) \\
		& \overset{(\ref{thm:backward-bound-for-wasserstein-with-lipschitz-continuous-strategies})}{\leq} \sum_{t=1}^{T} \left(\prod_{v>t} K_v \right)\cdot d_W \left(P_{<t+1}\left(\argdot \, | \, \left(\beta_{i<t}^{\prime}, \beta_{it}\right), \beta_{-i} \right), P_{<t+1}\left(\argdot \, | \, \left(\beta_{i<t}^{\prime}, \beta_{it}^{\prime} \right), \beta_{-i} \right) \right) \\
		&\overset{(\ref{thm:beta-beta-prime-lipschitz-condition-in-p-t})}{\leq} \sum_{t=1}^{T} K_{>t} \left|\left|\beta_{it} - \beta_{it}^{\prime} \right| \right|_{\infty},
	\end{align*}
	where $K_{>t} = \prod_{v=t+1}^{T} K_v$ with $K_t := K_{0t} + \sum_{i \in \mathcal{N}} K_{\beta_{it}} K_{\sigma_{it}}$ (see proof of \autoref{thm:backward-bound-for-wasserstein-with-lipschitz-continuous-strategies}) for all $1 \leq t \leq T$. By choosing $\delta < \max_{t \in T} \frac{\epsilon}{K_{>t} \cdot T}$, we get the statement.
\end{proof}

With the established results, we can finally give the proof of our main result.

\subsection{Proof of Theorem \ref{thm:main-approximation-result-for-lipschitz-utilities-informal}}

\begin{theorem} \label{thm:main-approximation-result-for-lipschitz-utilities}
	Let $\Gamma = \left(\mathcal{N}, T, S, \mathcal{A}, p, \sigma, u \right)$ be a multi-stage game, where Assumptions \ref{ass:bounded-signaling-and-action-spaces} and \ref{ass:lipschitz-continuous-signals-and-ultimately-strategies} hold, and that the utility function $u_i$ is continuous. Further, let $\beta_{-i} \in \Sigma_{-i}^{\text{Lip}}$, $\beta_i \in \Sigma_i$, and $A_{M_{\text{IS}}}$ and $B_{M_{\text{IS}}}$ be simulated data sets with initial simulation size $M_{\text{IS}} \in \mathbb{N}$ as described in Section \ref{sec:monte-carlo-integratoin-for-conditional-utilities}. Then, we have
	\begin{align*}
		\lim_{D \rightarrow \infty} \varepsilon_{\text{D}} \leq 0 \text{ and } \lim_{M_{\text{IS}} \rightarrow \infty} \varepsilon_{M_{\text{IS}}} = 0 \text{ almost surely.}
	\end{align*}
	Furthermore, we receive an upper bound on the utility loss over pure Lipschitz continuous strategies for the strategy profile $\beta = (\beta_i, \beta_{-i})$ by
	\begin{align*}
		\lim_{D \rightarrow \infty} \lim_{M_{\text{IS}} \rightarrow \infty} \ell_i^{\text{ver}}(\beta) = \lim_{D \rightarrow \infty} \lim_{M_{\text{IS}} \rightarrow \infty} \hat{u}_{i}^{\text{ver, D}}(\beta_{-i}\,|\,A_{M_{\text{IS}}}) - \hat{u}_i(\beta \, | \, B_{M_{\text{IS}}}) \\
		\geq \sup_{\beta^{\prime}_i \in \Sigma_i^{\text{Lip, p}}} \tilde{u}_i(\beta^{\prime}_i, \beta_{-i}) - \tilde{u}_i(\beta) = \tilde{\ell}_i^{\text{Lip, p}}(\beta) \text{ a. s.}
	\end{align*}
\end{theorem}

\begin{proof}
	By Lemma~\ref{thm:estimated-counterfactual-conditional-utility-converges-to-analytical-value}, we have almost sure convergence of $\lim_{M_{\text{IS}}  \rightarrow \infty} \varepsilon_{M_{\text{IS}}} = 0$. To finish the first statement, it remains to show
	$\lim_{\text{D} \rightarrow \infty} \varepsilon_{\text{D}} \leq 0$.
	
	Let $\epsilon > 0$ and $\bar{\beta}_i \in \Sigma_{i}^{\text{Lip, p}}$ such that $\sup_{\beta^{\prime}_i \in \Sigma_{i}^{\text{Lip, p}}} \tilde{u}_i\left(\beta_i^{\prime}, \beta_{-i}\right) - \tilde{u}_i\left(\bar{\beta}_i, \beta_{-i}\right) \leq \epsilon$. Then, by \autoref{thm:finite-step-functions-arbitrarily-close-to-continuous-functions}, there exists a sequence $\left\{\beta_i^{\text{D}}\right\}_{\text{D} \in \mathbb{N}}$ with $\beta_i^{\text{D}} \in \Sigma_{i}^{\text{D}}$ such that 
	\begin{align*}
		\lim_{\text{D} \rightarrow \infty} \left| \left|\bar{\beta}_i - \beta_{i}^{\text{D}} \right| \right|_{\infty} = 0.
	\end{align*}
	By \autoref{thm:close-strategies-in-infinity-imply-close-wasserstein-distance-of-ex-ante-measure}, we further get
	\begin{align*}
		\lim_{\text{D} \rightarrow \infty} d_W \left(P\left(\argdot \, | \, \bar{\beta}_i, \beta_{-i} \right), P\left(\argdot \, | \, \beta_i^{\text{D}}, \beta_{-i} \right) \right) = 0.
	\end{align*}
	The utility functions $u_i$ are bounded and continuous by assumption. Therefore, we can use \autoref{thm:wasserstein-metrizes-weak-convergence} and get
	\begin{align*}
		\lim_{\text{D} \rightarrow \infty} \tilde{u}_i\left(\beta_i^{\text{D}}, \beta_{-i}\right) = \lim_{\text{D} \rightarrow \infty}  \int_{\mathcal{A}} u_i(a) d P\left(a\,|\, \beta_i^{\text{D}}, \beta_{-i}\right) = \tilde{u}_i\left(\bar{\beta}_i, \beta_{-i}\right).
	\end{align*}
	As this holds for every $\epsilon>0$, we get for $\Sigma_{i}^{\text{SF}} := \bigcup_{\text{D} \in \mathbb{N}} \Sigma_{i}^{\text{D}}$
	\begin{align*}
		\sup_{\beta^{\prime}_i \in \Sigma_{i}^{\text{Lip, p}}} \tilde{u}_i\left(\beta_i^{\prime}, \beta_{-i}\right) \leq \sup_{\beta^{\text{SF}}_i \in \Sigma_{i}^{\text{SF}}} \tilde{u}_i\left(\beta_i^{\text{SF}}, \beta_{-i}\right),
	\end{align*}
	finishing the first statement.
	For the second statement, note that due to the boundedness of $u_i$, we can use Kolmogorov's law of large numbers and get
	\begin{align} \label{equ:monte-carlo-approximation-estimates-ex-ante-utility}
		\lim_{M_{\text{IS}}  \rightarrow \infty} \hat{u}_i(\beta \, | \, B_{M_{\text{IS}}}) = \tilde{u}_i(\beta).
	\end{align}
	Furthermore, we get that
	\begin{align*}
		\left| \ell_i^{\text{ver}}(\beta) - \tilde{\ell}_i^{\text{Lip, p}}(\beta) \right| &= 
		\left| \hat{u}_{i}^{\text{ver, D}}(\beta_{-i}\,|\,A_{M_{\text{IS}}}) - \hat{u}_i(\beta \, | \, B_{M_{\text{IS}}}) - \sup_{\beta^{\prime}_i \in \Sigma_i^{\text{Lip, p}}} \tilde{u}_i(\beta^{\prime}_i, \beta_{-i}) + \tilde{u}_i(\beta) \right| \\
		&\leq \left|\hat{u}_{i}^{\text{ver, D}}(\beta_{-i}\,|\,A_{M_{\text{IS}}}) - \sup_{\beta^{\prime}_i \in \Sigma_{i}^{\text{D}}} \tilde{u}_i(\beta^{\prime}_i, \beta_{-i}) \right| + \left| \sup_{\beta^{\prime}_i \in \Sigma_{i}^{\text{D}}} \tilde{u}_i(\beta^{\prime}_i, \beta_{-i}) - \sup_{\beta^{\prime}_i \in \Sigma_i^{\text{Lip, p}}} \tilde{u}_i(\beta^{\prime}_i, \beta_{-i}) \right| \\
		& + \left|\hat{u}_i(\beta \, | \, B_{M_{\text{IS}}}) - \tilde{u}_i(\beta) \right| \\
		&= \varepsilon_{\text{D}} + \varepsilon_{M_{\text{IS}}} + \left|\hat{u}_i(\beta \, | \, B_{M_{\text{IS}}}) - \tilde{u}_i(\beta) \right|.
	\end{align*}
	From the first statement and using the relation of Equation \ref{equ:monte-carlo-approximation-estimates-ex-ante-utility}, we get
	\begin{align*}
		\lim_{D \rightarrow \infty} \lim_{M_{\text{IS}} \rightarrow \infty} \left| \ell_i^{\text{ver}}(\beta) - \tilde{\ell}_i^{\text{Lip, p}}(\beta) \right| \geq 0,
	\end{align*}
	finishing the statement.
\end{proof}

\section{Convergence Results to Equilibrium} \label{sec:discussion-convergence-in-games}
The literature on convergence in games is vast and rapidly growing. Here, we state some central results from this field, with a focus on policy gradients algorithms and continuous multi-stage games. For a broader discussion, refer to the survey articles by \citet{yang2020overview} and \citet{zhang2021multi}.

In games with finitely many states and actions, some algorithms are guaranteed to converge to equilibrium in specific game classes. The time-average policies of no-regret algorithms converge to the set of coarse correlated equilibria~\citep{blumLearningRegretMinimization2007}. In two-player zero-sum games, this implies that the time-average policies converge to the set of Nash equilibria. \citet{hennesNeuralReplicatorDynamics2020} show that a popular policy gradient algorithm, where the policy is tabular and parametrized by a softmax function, satisfies the no-regret property. Other MARL variants with this property include neural fictitious play~\citep{heinrichDeepReinforcementLearning2016} and neural replicator dynamics~\citep{hennesNeuralReplicatorDynamics2020}. Additionally, some approaches use a meta-game solver, where RL agents learn best-responses in an inner loop to compute best-responses against an average of previous policies~\citep{lanctotUnifiedGameTheoreticApproach2017, mcaleerAnytimePSROTwoPlayer2022}. These works consider convergence in average policy. Some algorithms also achieve last-iterate convergence in two-player zero-sum games with finitely many states and actions~\citep{lockhartComputingApproximateEquilibria2019, farinaKernelizedMultiplicativeWeights2022}.

Another game class learnable by MARL algorithms is Markov potential games~\citep{mardenStateBasedPotential2012, macuaLearningParametricClosedloop2018, leonardosGlobalConvergenceMultiagent2022}. \citet{leonardosGlobalConvergenceMultiagent2022} provide convergence guarantees to Nash equilibrium policies for independent policy gradient algorithms. \citet{dingIndependentPolicyGradient2022} extend these results to independent policy gradient methods with function approximation, offering sharper convergence rates.

For general-sum stochastic games, convergence remains elusive. \citet{giannouConvergencePolicyGradient2022} describe properties of Nash equilibrium policies in finite Markov games, such as strictness and second-order stationarity, that ensure policy gradient algorithms in self-play converge to NE with high probability.

Convergence guarantees in multi-stage games with continuous signals and actions are even rarer. \citet{zhangPolicyOptimizationProvably2019} show that policy gradient algorithms in self-play converge to NE in two-player zero-sum quadratic games. However, finding the NE of zero-sum Markov games generally becomes a nonconvex-nonconcave saddle-point problem~\citep{mazumdarFindingLocalNash2019, buGlobalConvergencePolicy2019, chasnovConvergenceAnalysisGradientBased2020}. This inherent difficulty persists even in the simplest linear quadratic setting with linear function approximation~\citep{buGlobalConvergencePolicy2019, chasnovConvergenceAnalysisGradientBased2020}. Consequently, most convergence results are local, addressing behavior around local NE points~\citep{meschederNumericsGans2017, daskalakisLimitPointsOptimistic2018, mertikopoulosOptimisticMirrorDescent2019}. Moreover, policy gradient updates in MARL often fail to converge to local NEs due to non-convergent behaviors such as limit cycling~\citep{meschederNumericsGans2017, daskalakisLimitPointsOptimistic2018, mertikopoulosOptimisticMirrorDescent2019} or the existence of non-Nash stable limit points~\citep{mazumdarFindingLocalNash2019}.

While some results about Markov potential games apply to games with continuous states~\citep{leonardosGlobalConvergenceMultiagent2022, dingIndependentPolicyGradient2022}, we know of no formal extension of potential games to games with multiple stages and continuous actions or imperfect information. 

In summary, the above results are either restricted to games with finitely many states and actions, two-player zero-sum, or Markov potential games. The class of games we consider, namely, multi-stage games with continuous signals and actions, does not fall into any of these categories. Therefore, to the best of our knowledge, no convergence guarantees directly apply to our setting. 

\end{APPENDICES}

\end{document}